%% file: main.tex
\documentclass[11pt]{article}
\usepackage[latin9]{inputenc}
\usepackage{mathtools}
\usepackage{bm}
\usepackage{amsmath}
\usepackage{amsthm}
\usepackage{amssymb}
\usepackage{xcolor}
\definecolor{utburnt}{HTML}{BF5700}
\definecolor{charcoal}{HTML}{333333}
\definecolor{goldenrod}{HTML}{D99E00}
\definecolor{tealblue}{HTML}{007C91}

\usepackage[backref=page,bookmarks]{hyperref}
\hypersetup{
  colorlinks=true,
  linkcolor=utburnt,
  citecolor=tealblue,
  urlcolor=tealblue
}
\usepackage{booktabs}
\usepackage[nameinlink,capitalize]{cleveref}
\usepackage[numbers,sort&compress]{natbib}

\makeatletter
\theoremstyle{plain}
\newtheorem{thm}{\protect\theoremname}
\theoremstyle{definition}
\newtheorem{defn}[thm]{\protect\definitionname}
\newtheorem{corollary}[thm]{Corollary}
\theoremstyle{remark}
\newtheorem{claim}[thm]{\protect\claimname}
\theoremstyle{plain}
\newtheorem{lem}[thm]{\protect\lemmaname}
\theoremstyle{plain}

\theoremstyle{remark}
\newtheorem{remark}{\protect\remarkname}
\newcommand{\Dglm}{D_{\mathrm{GLM}}}
\usepackage{epsfig}
\usepackage{color}
\usepackage{times}

\usepackage{fullpage}
\usepackage{enumitem}

\usepackage[linesnumbered,ruled,vlined]{algorithm2e}
\usepackage[noend]{algpseudocode}
\newcommand{\sigmoidclipped}{\sigmoid_{\text{clipped}}}
\newcommand{\cube}[1]{\{\pm 1\}^{#1}}

\usepackage{tikz}
\usetikzlibrary{positioning,arrows.meta,calc}

\tikzset{
  box/.style   = {rectangle, rounded corners, draw, align=center,
                  minimum width=3.8cm, minimum height=1.2cm},
  arrow/.style = {-{Stealth[length=3mm]}, thick},
  conn/.style  = {-, thick}, 
  lab/.style   = {midway, fill=white, inner sep=1.5pt, align=center}
}

\newcommand{\spn}[1]{\text{Span}\left( #1 \right)}
\DeclareMathOperator*{\argmaxop}{arg\,max}
\DeclareMathOperator*{\argminop}{arg\,min}
\usepackage{bbm}

\newcommand{\myeqold}[1]{\stackrel{\mathclap{\normalfont\mbox{#1}}}{=}}
\newcommand{\myleqold}[1]{\stackrel{\mathclap{\normalfont\mbox{#1}}}{\leq}}
\newcommand{\mygeqold}[1]{\stackrel{\mathclap{\normalfont\mbox{#1}}}{\geq}}

\newcommand{\setalphazero}{C(\log \frac{\Delta}{\eps}+1)}

\DeclareMathOperator*{\E}{\mathbb{E}}%

\makeatother

\providecommand{\claimname}{Claim}
\providecommand{\definitionname}{Definition}
\providecommand{\factname}{Fact}
\providecommand{\lemmaname}{Lemma}
\providecommand{\theoremname}{Theorem}
\providecommand{\remarkname}{Remark}

\begin{document}
\global\long\def\red#1{\text{{\color{red}#1}}}%

\global\long\def\dimvc{\dim_{\text{VC}}}%

\global\long\def\spn#1{\text{Span}\left(#1\right)}%

\global\long\def\argmax{\argmaxop}%

\global\long\def\argmin{\argminop}%

\global\long\def\poly{\mathrm{poly}}%

\global\long\def\R{\mathbb{R}}%

\global\long\def\P{\mathbb{P}}%

\global\long\def\C{\mathbb{C}}%

\global\long\def\Z{\mathbb{Z}}%

\global\long\def\F{\mathcal{F}}%

\global\long\def\B{\mathcal{B}}%

\global\long\def\Ep{\mathcal{E}}%

\global\long\def\sign{\mathrm{{sign}}}%

\global\long\def\opt{\mathrm{{opt}}}%

\global\long\def\myeq#1{\myeqold{#1}}%

\global\long\def\myleq#1{\myleqold{#1}}%

\global\long\def\mygeq#1{\mygeqold{#1}}%

\global\long\def\indicator{\mathbbm{1}}%

\global\long\def\norm#1{\left\Vert #1\right\Vert }%

\global\long\def\bra#1{\left\langle #1\right\rangle }%

\global\long\def\abs#1{\left|#1\right|}%

\global\long\def\d{\,d}%

\global\long\def\conditional{\bigg\vert}%

\global\long\def\vect#1{\bm{#1}}%

\global\long\def\mydot#1{\stackrel{\mathbf{.}}{#1}}%

\global\long\def\mydoubledot#1{\stackrel{\mathbf{.}\mathbf{.}}{#1}}%

\global\long\def\mytripledot#1{\stackrel{\mathbf{.}\mathbf{.}\mathbf{.}}{#1}}%

\global\long\def\mybar#1{\overline{#1}}%

\global\long\def\myhat#1{\widehat{#1}}%

\global\long\def\spn#1{\mathop{\text{span}}\left(#1\right)}%

\global\long\def\dtv#1{\mathop{\mathrm{dist}_{\mathrm{TV}}}\left(#1\right)}%

\global\long\def\conv#1{\mathop{\text{conv}}\left(#1\right)}%

\global\long\def\righ{\rightarrow}%

\global\long\def\var{\mathrm{var}}%

\global\long\def\Inf{\mathrm{Inf}}%

\global\long\def\N{\mathcal{N}}%

\global\long\def\dpairs{D}%

\global\long\def\dmarg{\mathcal{D}}%

\global\long\def\setalpha{\frac{\eps^{25}}{C\left(\log\frac{1}{\eps}\right)^{C}}}%

\global\long\def\setk{\frac{C^{0.1}\log^{16}\frac{1}{\eps}}{\eps^{8}}}%

\global\long\def\setSparseDatasetSize{\sqrt{C}\frac{k}{\eps^{2}}\log^{2}\frac{k}{\eps}}%

\global\long\def\setEpsilonForGLM{\frac{C^{0.01}\log^{2}k}{\sqrt{k}}}%

\newcommand{\vhead}{\vect v_{\mathrm{head}}}
\newcommand{\vtail}{\vect v_{\mathrm{tail}}}

\newcommand{\vtailhat}{\widehat{\vect v}_{\mathrm{tail}}}

\newcommand{\learnhalfspace}{\textsc{LearnHalfspace}}

\newcommand{\x}{\vect x}
\newcommand{\tautemp}{t}
\newcommand{\hatv}{\widehat{\vect v}}
\newcommand{\hattau}{\widehat{\tau}}
\newcommand{\hatI}{\widehat{\mathrm{I}}}
\newcommand{\hati}{j}

\newcommand{\I}{\mathcal{I}}

\newcommand{\findheavycoefficients}{\textsc{FindHeavyCoefficients}}
\newcommand{\findregularcoefficients}{\textsc{FindRegularCoefficients}}

\newcommand{\eps}{\varepsilon}
\newcommand{\clippingthreshold}{\rho}
\newcommand{\Gauss}{\mathcal{N}}
\newcommand{\psiclip}{\psi_{\mathrm{clp}}}

\newcommand{\V}{\mathcal{V}}
\newcommand{\relu}{\mathsf{ReLU}}
\newcommand{\sigmoid}{\sigma^{\mathrm{sig}}}
\newcommand{\ind}{\mathbbm{1}}

\newcommand{\Sinp}{S_{\mathrm{inp}}}
\newcommand{\Sref}{S_{\mathrm{ref}}}
\newcommand{\Sadv}{S_{\mathrm{adv}}}
\newcommand{\Scln}{S_{\mathrm{cln}}}
\newcommand{\Srem}{S_{\mathrm{rem}}}
\newcommand{\Sout}{S_{\mathrm{filt}}}

\newcommand{\Sinplabeled}{\bar{S}_{\mathrm{inp}}}
\newcommand{\Sreflabeled}{\bar{S}_{\mathrm{ref}}}
\newcommand{\Sadvlabeled}{\bar{S}_{\mathrm{adv}}}
\newcommand{\Sclnlabeled}{\bar{S}_{\mathrm{cln}}}
\newcommand{\Sremlabeled}{\bar{S}_{\mathrm{rem}}}
\newcommand{\Soutlabeled}{\bar{S}_{\mathrm{filt}}}

\newcommand{\findregularcontaminated}{\textsc{FindRegularContaminated}}

\newcommand{\hsparse}{h_{\mathrm{sp}}}
\newcommand{\hatvsparse}{\hatv_{\mathrm{sp}}}
\newcommand{\hatvheavy}{\hatv_{\mathrm{hv}}}
\newcommand{\hatvhead}{\hatv_{\mathrm{head}}}
\newcommand{\hatvtail}{\hatv_{\mathrm{tail}}}

\newcommand{\AregGLM}{\mathcal{A}_{\mathrm{regGLM}}}
\newcommand{\Amisspecified}{\mathcal{A}_{\mathrm{misspGLM}}}
\newcommand{\AGLM}{\mathcal{A}_{\mathrm{GLM}}}

\newcommand{\zalphazero}{\alpha}
\newcommand{\etatwo}{\eta}

\title{A Fully Polynomial-Time Algorithm for Robustly \\ Learning Halfspaces over the Hypercube}
\author{
    \begin{tabular}{cc}
        \begin{tabular}{c}
            Gautam Chandrasekaran\thanks{Supported by the NSF AI Institute for Foundations of Machine Learning (IFML).} \\ \texttt{gautamc@cs.utexas.edu} \\ UT Austin
        \end{tabular} &
        \begin{tabular}{c}
            Adam R. Klivans\thanks{Supported by NSF award AF-1909204 and the NSF AI Institute for Foundations of Machine Learning (IFML).}\\ \texttt{klivans@cs.utexas.edu} \\ UT Austin 
        \end{tabular}
        \\\\
        \begin{tabular}{c}
             Konstantinos Stavropoulos\thanks{Supported by the NSF AI Institute for Foundations of Machine Learning (IFML) and by the 2025 Apple Scholars in AI/ML PhD fellowship.} \\ \texttt{kstavrop@utexas.edu} \\ UT Austin
        \end{tabular}
         & 
         \begin{tabular}{c}
              Arsen Vasilyan\thanks{Supported by the NSF AI Institute for Foundations of Machine Learning (IFML).} \\ \texttt{arsenvasilyan@gmail.com} \\ UT Austin
         \end{tabular}
    \end{tabular}
}
\date{}
\maketitle

\begin{abstract}

We give the first fully polynomial-time algorithm for learning halfspaces with respect to the uniform distribution on the hypercube in the presence of contamination, where an adversary may corrupt some fraction of examples and labels arbitrarily.  We achieve an error guarantee of $\eta^{O(1)}+\eps$ where $\eta$ is the noise rate.  Such a result was not known even in the agnostic setting, where only labels can be adversarially corrupted.  All prior work over the last two decades has a superpolynomial dependence in $1/\eps$ or succeeds only with respect to continuous marginals (such as log-concave densities). 

Previous analyses rely heavily on various structural properties of continuous distributions such as anti-concentration.  Our approach avoids these requirements and makes use of a new algorithm for learning Generalized Linear Models (GLMs) with only a polylogarithmic dependence on the activation function's Lipschitz constant.  More generally, our framework shows that supervised learning with respect to discrete distributions is not as difficult as previously thought.

\end{abstract}

\thispagestyle{empty}
\newpage
\setcounter{page}{1}
\input{intro}

\section{Brief Technical Overview}

\paragraph{Algorithm Outline.} Our algorithm receives $m$ i.i.d. examples drawn from some distribution over pairs $(\x,y)\in \cube{d}\times \cube{}$ such that $\x$ follows the uniform distribution over the $d$-dimensional hypercube and there is an unknown halfspace $f^*(\x) = \sign(\vect v^* \cdot \x + \tau^*)$ such that $\Pr[y\neq f^*(\x)] = \opt$. The output of the algorithm is guaranteed to have error $O(\opt^{1/c})+\eps$ for some constant $c$.\footnote{
In our proofs, we assume that $\opt \le \eps^c$ for some absolute constant $c$, which is without loss of generality since one can always redefine $\eps$ as a new value $\eps' \ge \opt^{1/c}$.}$^,$\footnote{The constant we provide is $c=25$, but it is possible that it can be improved.} We proceed in the following phases (see also Algorithm \ref{algorithm:learn-halfspace} for more details).
\begin{itemize}
    \item\label{item:algo-phase-1} \textbf{Phase 1: Finding indices of influential variables.} We first find a set of indices $j_1,j_2,\dots,j_k\in[d]$ that correspond to the $k$ largest (in absolute value) coordinates of some vector $\tilde{\vect v}^*$ such that:
    \begin{equation}
        \Pr[f^*(\x) \neq \sign(\tilde{\vect v}^*\cdot \x + \tau^*)] \le 2k\cdot \opt \label{equation:halfspace-with-known-head}
    \end{equation}
    \item\label{item:algo-phase-2} \textbf{Phase 2: Learn a structured halfspace.} We find a halfspace $h_{\mathrm{r}}(\x) = \sign(\hatv_{\mathrm{head}}\cdot \x + \hatv_{\mathrm{tail}} \cdot \x + \hattau)$, where $\hatv_{\mathrm{head}}$ is supported on $j_1,\dots,j_k$ and $\hatv_{\mathrm{tail}}$ is supported on the remaining coordinates. 
    \begin{enumerate}[label=(\alph*)]
        \item First, we compute $(\hatv_{\mathrm{head}},\hattau)$ through Algorithm \ref{algorithm:find-heavy-coefficients}.
        \item Then, we compute $\hatv_{\mathrm{tail}}$ through Algorithm \ref{algorithm:find-regular-coefficients} using $(\hatv_{\mathrm{head}},\hattau)$ as part of the input.
    \end{enumerate} 
        \item\label{item:algo-phase-3} \textbf{Phase 3: Learn a sparse halfspace.} We attempt to find $\hsparse(\x) = \sign(\hatvsparse\cdot \x+\hattau)$ defined by a coefficient vector $\hatvsparse$ which is only supported on the coordinates $j_1,j_2,\dots,j_k$ that perfectly classifies a set of $\poly(k,1/\eps)$ input examples. This can be done through linear programming.
\end{itemize}

\paragraph{Proof Outline.} Our proof is quite technical and splits into several cases as depicted by the tree of reductions in \Cref{figure:proof-outline-0}. 
It starts with a robust correspondence between coefficients of large magnitude and the coordinate influences of the unknown halfspace (see \Cref{section:technical-overview-influential} for a more thorough overview). This enables us to reduce the initial halfspace learning problem to one where the coordinates corresponding to the $k$ largest coefficients are known, by incurring an $O(k)$ error amplification (Phase 1).
We proceed with the \textit{critical index} machinery introduced by Servedio \cite{servedio2006} (we use a statement from \cite{gopalan2010fooling}), which asserts that any halfspace over the Boolean hypercube is either close to a sparse halfspace or is {\em structured}, meaning that beyond a small number of \emph{heavy (or head)} coefficients, independent of the input dimension, all remaining weights are of roughly comparable magnitude (the \emph{regular or tail} coefficients). This allows us to decompose the learning task into two subproblems: learning either a structured halfspace (Phase 2) or a sparse halfspace (Phase 3) --- we try both cases in parallel and check afterwards to see which case was successful.

\paragraph{Structured Case.} The structured case comprises the bulk of the technical work of the paper and requires several new ideas. First, we establish a connection between learning structured halfspaces and learning generalized linear models (GLMs), which allows us to learn the heavy coefficients (see \Cref{section:technical-overview-heavy} for more background on GLMs and an overview of our GLM learning results). Concretely, we treat the aggregate contribution of the tail variables as approximately Gaussian noise and represent each example as $(\tilde{\x}, y)$, where $\tilde{\x}$ corresponds to the heavy variables and $\E[y|\tilde{\x}]$ is given by a GLM with (approximately) known activation. This viewpoint is justified by the regularity property of the tail weights, which enables the use of the Berry-Esseen theorem to predict the collective behavior of the inner product between the tail variables and the corresponding (regular) coefficients.

A key observation is that the number of heavy coefficients $k$ can be chosen small enough relative to the input noise level so that the effect of noise becomes negligible when learning the corresponding coefficients. Specifically, if the error of some halfspace ${f}$ is $\eta$, then for $m \ll 1/\eta$ labeled i.i.d. samples, with high probability all labels are consistent with $f$. 

In our analysis, we choose the halfspace $f$ to be the one whose $k$ largest coefficients are known and has error $\eta \le (2k+1)\opt$ (Phase 1, Eq. \eqref{equation:halfspace-with-known-head}). The choice $m\ll 1/\eta$ is feasible because the effective dimension of the problem is $\dim(\tilde{\x}) = k \ll 1/\opt^c$ for some sufficiently large constant $c$. 
Moreover, according to the critical index lemma, if $f$ is structured, then the heavy coefficients coincide with the $k$ largest coefficients of $f$, whose positions are known. This enables us to identify the heavy variables $\tilde{\x}$.

Finally, once the heavy coefficients are fixed, we estimate the tail coefficients robustly and efficiently by minimizing a convex surrogate loss (see \Cref{section:technical-overview-regular} for more details). The crucial observation here is that the tail distribution is anti-concentrated --- again, by Berry-Esseen --- ensuring that the value of the hinge loss is representative of the misclassification error of the optimal halfspace.

\paragraph{Sparse Case.} Once more, we use the fact that the number of relevant variables $k$ in the sparse case can be chosen small enough relative to the input noise level. 
Since the positions of the relevant variables are known (due to Phase 1), we can use a simple linear program to recover a sparse halfspace consistent with the data and guaranteed to be close to the optimum.

\subsection{Learning the Heavy Coefficients: A Connection with GLM Learning}\label{section:technical-overview-heavy}

Recall that in the structured case (Phase 2 of our algorithm), we first learn the heavy coefficients. We will now sketch the main idea behind the algorithm used for this purpose. First, we focus on a seemingly unrelated problem called GLM (Generalized Linear Model) learning, which turns out to be closely related to the problem of learning the heavy coefficients in the structured case. In GLM learning, there is a distribution over pairs $(\x,y)\in \R^d\times\cube{}$, where the conditional expectation $\E[y|\x]$ is specified by some function of the form $\sigma(\vect w^*\cdot \x+ \tau^*)$, where $\sigma:\R\to [-1,1]$ is known, and $\vect w, \tau$ are unknown. The goal is to learn some $\widehat{w},\hattau$ such that:
\[
    \E\Bigr[\bigr(\sigma(\vect w^* \cdot \x + \tau^*) - \sigma(\widehat{\vect w} \cdot \x + \hattau)\bigr)^2\Bigr] \le \eps\,.
\]

One of our main technical contributions is a GLM learning algorithm whose sample complexity is independent of the magnitude $w_{\max}$ of associated coefficients, and whose runtime scales logarithmically with $w_{\max}$. In particular, we obtain the first efficient algorithm for learning sigmoidal GLMs (activation $\sigmoid(t) = (1-e^{-t})/(1+e^{-t})$) with sample complexity independent of the weights:

\begin{thm}\label{theorem:sigmoid-main}
    There is an algorithm that, given sample access to a distribution over pairs $(\x,y)\in \R^d\times \{\pm 1\}$ where $\|\x\|_2 \le \sqrt{d}$ and $\E[y | \x] = \sigmoid(\vect w^* \cdot \x + \tau^*)$ for some unknown $\vect w^*,\tau^*$ with $\|\vect w^*\|_2 + |\tau^*| \le w_{\max}$, outputs $\widehat{\vect w},\hattau$ such that with probability at least $0.99$ we have:
    \[
        \E\Bigr[\bigr(\sigmoid(\vect w^* \cdot \x + \tau^*) - \sigmoid(\widehat{\vect w} \cdot \x + \hattau)\bigr)^2\Bigr] \le \eps\,.
    \]
    The sample complexity of the algorithm is $\tilde{O}(d / \eps^2)$ and its time complexity is $\poly(d, 1/\eps, \log w_{\max})$.
\end{thm}

The proof of the above theorem combines two key ingredients: (1) a known connection between GLM learning and a class of surrogate losses often referred to as \emph{matching losses} (see, e.g., \cite{kanade2018note}), and (2) the guarantees provided by the ellipsoid algorithm. To ensure that the sample complexity does not depend on $w_{\max}$, we first establish the result for GLMs with \emph{clipped activations} --- that is, activations $\sigma(t)$ whose absolute value is fixed at $1$ whenever $|t|$ exceeds a threshold $\zalphazero$. We then argue that for sufficiently large $\zalphazero$, since the sigmoid approaches $1$ exponentially fast, the distribution induced by the clipped sigmoid GLM is within vanishing total variation distance of the original GLM distribution, relative to the number of samples needed for generalization. Hence, the result for the clipped model directly extends to the original GLM.

A modified version of the above result (see \Cref{cor: the GLM component theorem nonhom}) is crucial in our approach here, but we believe that \Cref{theorem:sigmoid-main} is of independent interest, since, to our knowledge, it is the first GLM learning result to achieve a logarithmic runtime dependence on the magnitude of the coefficients, which essentially quantifies how steep the transition of $\E[y|\x]$ from $-1$ to $1$ is, as $\x$ moves along the direction of $\vect w^*$. See section \Cref{section:heavy} for the proof of \Cref{cor: the GLM component theorem nonhom} and \Cref{appendix:sigmoid} for the proof of \Cref{theorem:sigmoid-main}.

In order to apply the result in our case, we identify GLM learning as a particular subproblem of learning a halfspace $f(\x) = \sign(\tilde{\vect v}^* \cdot \x + \tau^*)$ over the $d$-dimensional hypercube whose $k$-largest coefficients are in known positions. Specifically, due the critical index lemma, one important scenario is that $\tilde{\vect v}^* = \vhead^* + \vtail^*$, where $\vhead^*$ is supported on the $k$ known coordinates (where $k$ does not depend on $d$), and $\vtail^*$ is supported on the remaining coordinates and is regular, i.e., $\|\vtail^*\|_2 = 1$ and $\|\vtail^*\|_4 \ll 1$. By focusing on the coordinates corresponding to $\vhead^*$ --- which are known due to the definition of $\tilde{\vect v}^*$ in Phase 1, Eq. \eqref{equation:halfspace-with-known-head} --- we obtain a new GLM learning problem over a distribution of pairs $(\tilde{\x}, y)$ where $\tilde{\x}\in \R^k$ and:
\begin{equation}
    \E[y|\tilde{\x}] = \E_{\x' \sim \cube{{d-k}}}\Bigr[\sign\bigr(\vect u_{\mathrm{head}}^*\cdot \tilde{\x} + \vect u_{\mathrm{tail}}^*\cdot \x' + \tau^*\bigr)\Bigr]\label{equation:main-glm}\,,
\end{equation}
where $\vect u_{\mathrm{head}}^*$ is the restriction of $\vhead^*$ on the $k$ known coordinates and $\vect u_{\mathrm{tail}}^*$ the restriction of $\vtail^*$ to the remaining coordinates.
We show that a solution $(\hatv_{\mathrm{head}},\hattau)$ to the corresponding GLM problem can be used in place of $(\vhead^*,\tau^*)$, increasing the classification distance from the initial halfspace by only a small amount. There are several additional technical complications. 

First, the GLM defined by Equation \eqref{equation:main-glm} has unknown activation, since $\vtail^*$ is unknown, but the surrogate loss technique requires knowledge of the activation. To address this, we apply Berry-Esseen theorem (\Cref{lemma:berry-esseen}), which implies that the distribution of $\vtail^*\cdot \x'$ is very close to the standard Gaussian. This leads to a sigmoid-like GLM, which has all the required properties to obtain a result similar to \Cref{theorem:sigmoid-main}, but is slightly misspecified.\footnote{The actual activation is $\psi(z)\coloneq \E_{x\sim \Gauss(0,1)}[\sign(z+x)]$, which has similar properties to the sigmoid.} This is because we do not know the activation values exactly, but only approximately. Nevertheless, we show that our technique is robust to such misspecification. 

Second, we still require some bound on $\|\vhead^*\|_2$, since the runtime of our GLM algorithm scales polylogarithmically to the magnitude of the coefficients. Such a bound can be obtained through a classical result in Boolean analysis showing that every halfspace over the $d$-dimensional hypercube can be exactly represented by integer coefficients of magnitude at most $w_{\max} = 2^{O(d \log d)}$ \cite{servedio2006}. Since we obtain a logarithmic dependence on $w_{\max}$, our runtime remains polynomial.

Finally, Equation \eqref{equation:main-glm} holds only when the input distribution is realized by the halfspace $f^*$, yet our input distribution is actually noisy. Nevertheless, since the effective dimension of the corresponding GLM problem is $k \ll 1/\opt^c$ for some constant $c$ (due to the critical index lemma),
the number of samples required to solve the problem is small compared to the label noise rate ($\opt$) and the labels are effectively realized by $f^*$. Therefore, our method is robust to noise. 

\subsection{Learning the Regular Coefficients: Hinge Loss Minimization}\label{section:technical-overview-regular}

Another important ingredient of our approach is an algorithm that learns the tail coefficients, once we have estimates for the heavy coefficients. A simpler version of the result we use in our proof is the following result, which solves a particular case of the noisy halfspace learning problem over the hypercube when the coefficients are regular.

\begin{thm}\label{theorem:hinge-main}
    There is an algorithm that, given sample access to a distribution over pairs $(\x,y)\in \cube{d}\times \cube{}$ where $\x$ is uniform over $\cube{d}$, outputs $\hatv\in \R^d$ such that with probability at least $0.99$ we have:
    \[
        \Pr\Bigr[ \sign(\hatv\cdot \x) \neq y \Bigr] \le \frac{C}{\eps}\cdot \log \frac{1}{\eps}\cdot \opt_\eps+ C\cdot \eps\,, \;\text{ for }\; \opt_\eps := \min_{f\in \mathcal{C}_\eps} \Pr\Bigr[ f(\x) \neq y \Bigr] 
    \]
    where $\mathcal{C}_\eps$ is the class of $\eps$-regular halfspaces, i.e., $f(\x) = \sign(\vect v\cdot \x)$ where $\|\vect v\|_2 = 1$ and $\|\vect v\|_4^2 \le \eps$ and $C$ is some sufficiently large constant. The sample and time complexity of the algorithm is $\poly(d, 1/\eps)$.
\end{thm}

The algorithm minimizes an $\eps$-rescaled version of the hinge loss, i.e., it minimizes the empirical expected value of the function $\ell_\eps(y,\vect v\cdot \x) = \relu(1 - y(\vect v\cdot \x)/\eps)$ over the input samples $(\x,y)$. It is a well-known fact that the hinge loss is a surrogate for the missclassification error, i.e., $\ind\{y\neq \sign(\vect v\cdot \x)\} \le \ell_\eps(y,\vect v\cdot \x)$. Therefore, achieving a small value for the hinge loss leads to low missclassification error. 

The non-trivial part of our proof is showing that the minimum value of the hinge loss is indeed shrinking with $\opt_\eps$. To this end, we focus on the value of the hinge loss on $\vect v^*_\eps$, where $\vect v^*_\eps$ defines the regular halfspace $f^*_\eps\in \mathcal{C}_\eps$ that achieves $\opt_\eps$. Due to regularity of $f^*_\eps(\x) = \sign(\vect v^*_\eps\cdot \x)$ and uniformity of $\x$, we may apply the Berry-Esseen theorem (\Cref{lemma:berry-esseen}) to show that the distribution of $\vect v^*_\eps\cdot \x$ is anticoncentrated. Moreover, the distribution of $\vect v^*_\eps\cdot \x$ is concentrated, since $\x$ is subgaussian. These two properties suffice to relate $\opt_\eps$ and $\E[\ell_\eps(y,\vect v^*_\eps\cdot \x)]$.

In our proof, we use a slightly more general version of \Cref{theorem:hinge-main} (\Cref{thm: with known head vars and weights}), which enables the learner to fix the heavy coefficients $\vect v_{\mathrm{heavy}}$ and the bias $\tau$ to their estimated values $(\hatv_{\mathrm{heavy}},\hattau)$ (from the GLM component) and only optimize over the regular tail $\vect v_{\mathrm{tail}}$. The additional complication here is that there could be some points $\x$ where $|\widehat{\vect v}_{\mathrm{head}}\cdot \x + \hattau| \gg |\vect v_{\mathrm{tail}}\cdot \x|$ leading to large values of the surrogate loss, even for the optimal choice of $\vect v_{\mathrm{tail}}$. To overcome this obstacle, we may simply filter out the points $\x$ for which $|\widehat{\vect v}_{\mathrm{head}}\cdot \x + \hattau| \gg \log(1/\eps)$, as the missclassification error on them is not influenced by the choice of $\vect v_{\mathrm{tail}}$, because the sign is determined by $\widehat{\vect v}_{\mathrm{head}}\cdot \x + \hattau$. See \Cref{section:regular} for more details.

\subsection{Finding the Positions of Influential Variables}\label{section:technical-overview-influential}

One crucial part of our proof is finding the coordinates $i_1,\dots,i_k\in [d]$ that correspond to coefficients of the optimum halfspace with the $k$ largest magnitudes. This is important for robustness, since it enables us to learn a sparse halfspace in the sparse case, as well as the heavy coefficients in the structured case, by solving appropriate low-dimensional problems (see Phase 1, Eq. \eqref{equation:halfspace-with-known-head}). If the dimensionality of the problem is low compared to the noise rate, then robustness can be guaranteed in a black-box manner, due to a total variation distance-based argument.

The coordinates $i_1,\dots,i_k$, however, cannot be recovered exactly, due to the label noise, as well as sampling errors. Instead, we provide the following robust coordinate identification result, which is sufficient for our purposes.

\begin{thm}[Simplified version of \Cref{theorem:influential}]\label{theorem:influential-main}
    There is an algorithm that, given sample access to a distribution over pairs $(\x,y)\in \cube{d}\times \cube{}$ where $\x$ is uniform over $\cube{d}$ and $\Pr[y\neq \sign(\vect v^*\cdot \x + \tau^*)] \le \eta$, outputs a set of $k$ indices $\{j_1,\dots, j_k\}$ such that with probability at least $0.99$:
    \[
        \Pr\Bigr[ \sign(\vect v^*\cdot \x + \tau^*) \neq \sign(\tilde{\vect v}^*\cdot \x + \tau^*) \Bigr] \le 4\eta k\,,
    \]
    where $\tilde{\vect v}^*$ is such that the $k$ largest (in absolute value) coefficients correspond to the indices $j_1,\dots,j_k$. The algorithm has time and sample complexity $\poly(d,1/\eta)$.
\end{thm}

The above result is essentially the strongest possible statement one could hope for, at least from a qualitative perspective. It implies that for all purposes, we may assume that the returned indices $j_1,\dots,j_k$ correspond to the coefficients of largest magnitudes, by only incurring an error amplification factor of $4k$. Moreover, the critical index lemma (\Cref{lem: critical index}) precisely identifies the large-magnitude coefficients with the relevant coordinates in the sparse case and the heavy coefficients in the structured case.

The proof of the above theorem is based on carefully combining the following three main ingredients. First, we use the fact that the ordering of the coordinate influences of a halfspace over the hypercube coincides with the ordering of the absolute values of its coefficients. Second, we show that the absolute difference between the influences of two coordinates $i,j$ gives an upper bound on the disagreement between two halfspaces that correspond to coefficient vectors $\vect v, \vect v'$, where $\vect v$ and $\vect v'$ are the same in all coordinates except $i,j$, where the magnitudes are swapped, i.e., $v_i' = |v_j|\cdot \sign(v_i)$ and $v_j' = |v_i|\cdot \sign(v_j)$. Finally, we observe that the coordinate influences can be efficiently estimated through the Chow parameters (see \cite{o2008chow}). See \Cref{section:chow-estimation} for more details.

\subsection{Tolerating Contamination}

Our method works even in the more challenging setting where both the features and the labels are adversarially contaminated for a bounded fraction of the input examples. In particular, our results apply to the following model called bounded contamination (or nasty noise), initially defined by \cite{BSHOUTY2002255}.

\begin{defn}[Learning with Bounded Contamination \cite{BSHOUTY2002255}]\label{definition:bounded-contamination}
            We say that a set $\Sinp$ of $N$ labeled samples is an $\eta$-contaminated (uniform) sample with respect to some class $\mathcal{C} \subseteq\{\cube{d} \to \cube{}\}$, where $N\in \mathbb{N}$ and $\eta\in(0,1)$, if it is formed by an adversary as follows.
    \begin{enumerate}
        \item The adversary receives a set of $N$ clean i.i.d. labeled examples $\Scln$, drawn from the uniform distribution over $\cube{d}$ and labeled by some unknown concept $f^*$ in $\mathcal{C}$. \item The adversary removes an arbitrary set $\Srem$ of $\lfloor\eta N\rfloor$ labeled examples from $\Scln$ and substitutes it with an adversarial set of $\lfloor\eta N\rfloor$ labeled examples $\Sadv$.
    \end{enumerate} 
    In this model, given $\epsilon\in(0,1)$ and $\Sinp$ for large enough $N$, 
    the goal of the learner is to output, with probability $0.95$, a hypothesis $h:\cube{d}\to \cube{}$ such that $\Pr_{\x\sim\cube{d}}[h(\x) \neq f^*(\x)] \le \poly(\eta) + \epsilon$.
\end{defn}

Once more, we obtain a polynomial-time algorithm for the problem based on a similar approach as the one for the label noise case, using the following additional ingredients. See \Cref{section:contamination} for more details.

\paragraph{Equivalence Between Adaptive and Oblivious Adversaries.} It follows from results in \cite{blanc2025adaptive} that learning some class $\cal C$ with bounded contamination over the hypercube is computationally equivalent to learning using samples from some distribution $D$ over $\cube{d}\times \cube{}$ whose total variation distance from the clean distribution $D_{\mathrm{clean}}$ is $\eta$. The marginal of $D_{\mathrm{clean}}$ on $\cube{d}$ is uniform and the clean labels are generated by some $f\in \cal C$. This is useful for our analysis, since it allows us to consider the input samples independent.

\paragraph{Finding Positions of Influential Variables under Contamination.} We provide an analogue of \Cref{theorem:influential-main} for the setting of (oblivious) contamination. The crucial observation here is that the way we estimate the positions of influential variables is by estimating the expected value of the quantities $y x_i$ for $i\in [d]$, whose absolute value is bounded by $1$ for any $y\in\cube{}, \x = (x_1,\dots,x_d)\in \cube{d}$. Therefore, the estimation error is bounded by the total variation distance between the input distribution and the clean distribution.

\paragraph{Learning the Heavy Coefficients or a Sparse Halfspace under Contamination.}  Recall that in both the sparse case and in learning the heavy coefficients of the structured case, the required sample size $m$ can be chosen to be vanishingly small compared to the noise rate. This enables safely ignoring the label noise through a total variance argument between the joint distribution of $m$ samples with or without noise. The same argument works even in the presence of contamination and, therefore, the noise can be ignored in this case as well.

\paragraph{Learning the Regular Coefficients under Contamination.} To learn the regular coefficients, we once more use an appropriate hinge loss minimization routine. However, we must ensure that the minimum value of the hinge loss remains bounded by a function of the contamination rate $\eta$. This requires two properties of the input distribution: (i) concentration in every direction, and (ii) anti-concentration along the direction of the optimal tail coefficient vector. Regarding the anti-concentration property, we use the fact that the input distribution is close to the uniform distribution over the hypercube, and that the optimum tail coefficient vector is regular (since we are in the structured case). To ensure concentration, we perform an outlier removal procedure on the input samples, ensuring that the moments of constant degree are bounded along any direction. Such procedures have been widely used in the literature of learning with contamination \cite{klivans2009learning,awasthi2017power,diakonikolas2018learning,goel2024tolerant,klivans2024learningac0}. Here, we use the formulation of \cite{klivans2024learningac0} as a black box, and obtain an algorithm for learning the regular coefficients that is robust to contamination.

\section{Notation and Roadmap}
\label{section:roadmap}

\paragraph{Notation.} We use bold letters to denote vectors. For a vector $\x\in\R^d$, we denote with $x_i$ the value of the $i$-th coordinate, and for $H\subseteq[d]$, we denote with $\x_H$ the vector in $\R^{|H|}$ corresponding to the restriction of $\x$ on the coordinates in $H$, i.e., $\x_{H} = (x_i)_{i\in H}$. For vectors that are denoted with subscript notation, e.g., $\x_{\mathrm{example}}$, we use the notation $(\x_{\mathrm{example}})_H$ to denote their restriction.

We use $D$ to denote a labeled distribution over pairs $(\x,y)\in \cube{d}\times \cube{}$, and $D_X$ to denote the marginal of $D$ on $\cube{d}$. Moreover, we use the notation $D_{y|\x}$ to denote the conditional distribution of $y$ given $\x$ when $(\x,y)$ is drawn according to $D$. We use the notation $(\x,y)\sim D$ to assert that the pair $(\x,y)$ is drawn according to the distribution $D$, and for a set $S$ of points in $\cube{d}\times \cube{}$, we use the notation $(\x,y)\sim S$ to assert that $(\x,y)$ is drawn from the uniform distribution over the elements in $S$.

We denote with $\sign:\R \to \cube{}$ the sign function and with $\ind_E\in\{0,1\}$ the indicator of the event $E$.

\paragraph{Roadmap.} Our main result is achieved by Algorithm \ref{algorithm:learn-halfspace}.
In the following sections, we provide a proof for our main result \Cref{thm: main theorem}, following the structure of Algorithm \ref{algorithm:learn-halfspace}, as well as the proof outline presented in \Cref{figure:proof-outline-0}. In particular:
\begin{itemize}
    \item In \Cref{section:chow-estimation} we provide an important tool (\Cref{theorem:influential}) that reduces the problem of learning an arbitrary halfspace over the hypercube to learning a halfspace which is slightly more noisy, but whose large-magnitude coefficients are in known positions.
    \item In \Cref{section:heavy} we show how to learn the heavy coefficients of a structured halfspace, i.e., a halfspace such that, except from a small number of heavy coefficients whose locations are known, has coefficients of roughly comparable magnitudes.
    \item In \Cref{section:regular} we prove that the regular coefficients of a structured halfspace can be learned robustly and efficiently, as long as we are given accurate estimates for the heavy coefficients.
    \item Finally, in \Cref{sec:combineall} we show how all these ingredients can be combined with the critical index lemma (\Cref{lem: critical index}) to obtain our main result of \Cref{thm: main theorem}.
\end{itemize}
Moreover, in \Cref{section:contamination} we show how to extend our results to the setting of learning under bounded contamination, where a bounded fraction of both datapoints $\x$ and their labels $y$ can be adversarially corrupted.

\begin{algorithm}[htbp]
\caption{$\learnhalfspace(\eps, D)$}\label{algorithm:learn-halfspace}
    \KwIn{Parameter $\eps\in (0,1)$, an oracle to independent samples $(\x, y)\sim D$}
\KwOut{A halfspace $\widehat{h}(\x) = \sign(\widehat{\vect v}\cdot \x +\widehat{\tau})$ described by $\hatv \in \R^d, \hattau\in \R$}
\BlankLine
\BlankLine
\tcc{Initialization}
\BlankLine
Let $C\ge 1$ be a sufficiently large universal constant\;
Set $k = C^{0.1} \log^{16} (1/\eps) / \eps^8$ and $\etatwo = \eps^{25}/C$; \Comment{$\etatwo$ is an upper bound for $\opt$}\\
\BlankLine
\BlankLine

\tcc{Finding Indices of Influential Variables (see Section \ref{section:chow-estimation})}
\BlankLine
Let $S_1$ be a set of  $C\sqrt{\log d} / \etatwo^2$ i.i.d. samples from $D$\; 
Compute the quantities $\hatI_i = \E_{(\x,y)\sim S_1}[y x_i] $ for $i\in[d]$\;
Let $H_k = \{\hati_1, \hati_2, \dots, \hati_k\}$ be such that $|\hatI_{i}|\le |\hatI_{j}|$ for all $i\in[d]$, $j\in H_k$;\Comment{$k$ largest values $|\hatI_i|$}\\
\BlankLine
\BlankLine

\tcc{Finding Structured Halfspace}
\BlankLine
\For{$\Delta = 0,1,2,\dots, k$}{
    \BlankLine
    \tcc{Part I: Heavy Coefficients (see Section \ref{section:heavy})}
    \BlankLine
    Let $H_{\Delta} = \{\hati_1,\hati_2,\dots,\hati_{\Delta}\}$ and let $D_\Delta$ the distribution of $(\x_{H_{\Delta}} , y)$ where $(\x,y)\sim D$\;
    Set $\eps_{\mathrm{hv}} = C^{0.01} \log^2 (k) / \sqrt{k}$ and $u_{\max} = d2^{20d \log d}$\;
    Let $(\widehat{\vect u}, \hattau_\Delta)$ be the output of $\findheavycoefficients(\eps_{\mathrm{hv}}, \Delta, u_{\max}, D_\Delta)$\;
    Let ${\hatvheavy}\in\R^d$ such that $({\hatvheavy})_{H_\Delta} = \widehat{\vect u}$ and $({\hatvheavy})_{[d]\setminus H_\Delta} = 0$\;
    \BlankLine
    \tcc{Part II: Regular Coefficients (see Section \ref{section:regular})}
    \BlankLine
    Set $\eps_{\mathrm{reg}} = \eps/C^{0.01}$\;
    Let $\hatv_\Delta$ be the output of $\findregularcoefficients(\eps_{\mathrm{reg}},d, H_\Delta, {\hatvheavy}, \hattau_\Delta, D)$\;
    Let $h_\Delta(\x) = \sign(\hatv_{\Delta} \cdot \x + \hattau_\Delta)$\;
}
\BlankLine
\BlankLine
\tcc{Finding Sparse Halfspace (see Section \ref{sec:combineall}). }
\BlankLine
Let $S_2$ be a set of $\sqrt{C} \frac{k}{\eps^2} \log^2 \frac{k}{\eps}$ i.i.d. samples from $D$\;
Use linear programming to attempt to find a halfspace ${\hsparse}(\x) = \sign({\hatvsparse} \cdot \x + \hattau)$, where ${\hatvsparse}\in\R^d$ is supported on the indices in $H_k$ and $\hattau\in \R$, such that $(\x,{\hsparse}(\x)) = (\x,y)$ for all $(\x,y)\in S_2$\;
If the linear program is infeasible then set ${\hsparse} \equiv +1$. 
\BlankLine
\BlankLine
\tcc{Selecting the Best Hypothesis. }
\BlankLine
Let $S_3$ be a set of $C \sqrt{\log k} / \eps^2$ i.i.d. samples from $D$\;
Choose $\widehat{h}$ to be the hypothesis in $\{h_0,\ldots,h_k\} \cup \{{\hsparse}\}$ that minimizes the empirical error on $S_3$\; 
\end{algorithm}

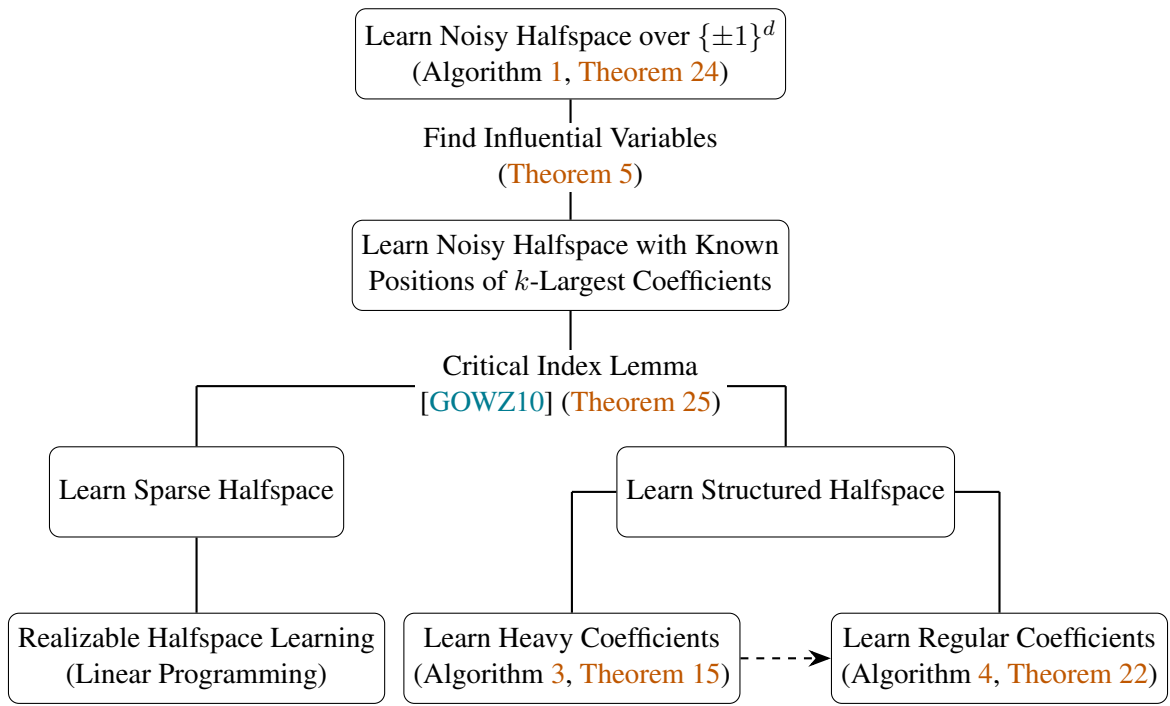
\begin{figure}[ht ]
    \centering
    \begin{tikzpicture}[node distance=1.6cm]

    \node[box] (top0) {Learn Noisy Halfspace over $\cube{d}$ \\ (Algorithm \ref{algorithm:learn-halfspace}, \Cref{thm: main theorem})};

    \node[box, below=1.6cm of top0] (top) {Learn Noisy Halfspace with Known \\ Positions of $k$-Largest Coefficients};

    \draw[conn] (top0.south) -- node[lab, below=-0.5cm of top0] {Find Influential Variables \\ (\Cref{theorem:influential})} (top.north);

    \node[box, below left=1.8cm and 0.1cm of top]  (L1) {Learn Sparse Halfspace};
    \node[box, below right=1.8cm and -2.3cm of top] (R1) {Learn Structured Halfspace};

    \path (top.south) ++(0,-1cm) coordinate (J);

    \path (L1.north |- J) coordinate (Lsplit);
    \path (R1.north |- J) coordinate (Rsplit);

    \draw[conn] (top.south) -- (J);

    \draw[conn] (J) -- (Lsplit);
    \draw[conn] (J) -- (Rsplit);

    \draw[conn] (Lsplit) -- (L1.north);
    \draw[conn] (Rsplit) -- (R1.north);

    \node[fill=white, inner sep=2pt, align=center] at (J) {Critical Index Lemma \\ \cite{gopalan2010fooling} (\Cref{lem: critical index})};


    \node[box, below=1cm of L1] (L2) {Realizable Halfspace Learning \\ (Linear Programming)};
    \draw[conn] (L1.south) -- (L2.north);

    \node[box, below right=1cm and -1.65cm of R1] (R2) {Learn Regular Coefficients \\ (Algorithm \ref{algorithm:find-regular-coefficients}, \Cref{thm: with known head vars and weights})};
    \node[box, below left=1cm and -1.65cm of R1] (R3) {Learn Heavy Coefficients \\ (Algorithm \ref{algorithm:find-heavy-coefficients}, \Cref{cor: the GLM component theorem nonhom})};

    \path (R3.north |- R1) coordinate (R3up);
    \path (R2.north |- R1) coordinate (R2up);

    \draw[conn] (R1.west) -- (R3up);
    \draw[conn] (R1.east) -- (R2up);

    \draw[conn] (R3up) -- (R3.north);
    \draw[conn] (R2up) -- (R2.north);
    \draw[arrow, dashed] (R3.east) -- (R2.west);

    \end{tikzpicture}
    \caption{Reduction tree for the proof of our main result (\Cref{thm: main theorem}). Each internal node corresponds to a learning problem that reduces to solving all of its children. The dashed arrow indicates a dependency between two leaf nodes, where the output of one serves as input to the other.}
    \label{figure:proof-outline-0}
\end{figure}

\section{Finding Indices of Influential Variables}\label{section:chow-estimation}

In this section, we prove the following, slightly more technical version of \Cref{theorem:influential-main}, which reduces the problem of learning an arbitrary noisy halfspace over the hypercube to learning a halfspace that is slightly more noisy, but whose large magnitude coefficients are in known positions.

\begin{thm}\label{theorem:influential}
    Let $D$ be a distribution over pairs $(\x,y)\in \cube{d}\times \cube{}$ where $\x$ is uniform over $\cube{d}$ and $\Pr[y\neq \sign(\vect v^*\cdot \x + \tau^*)] \le \eta$. Let $H_k = \{j_1,\dots,j_k\}$ be such that:
    \[
        \Bigr|\E_{(\x,y)\sim S}[y x_{i}]\Bigr|\le \Bigr|\E_{(\x,y)\sim S}[y x_{j}] \Bigr| \;,\text{ for all }i\in[d]\setminus H_k\text{ and }j\in H_k\,,
    \]
    where $S$ is a set of $C\sqrt{\log d} / \eta^2$ i.i.d. samples from $D$ for some sufficiently large universal constant $C\ge 1$. Then, with probability at least $0.99$, we have:
    \[
        \Pr_{\x\sim \cube{d}}\Bigr[ \sign(\vect v^*\cdot \x + \tau^*) \neq \sign(\tilde{\vect v}^*\cdot \x + \tau^*) \Bigr] \le 4\eta k\,,
    \]
    where $\tilde{\vect v}^*$ is such that the $k$ largest (in absolute value) coefficients correspond to the indices $j_1,\dots,j_k$. The algorithm has time and sample complexity $\poly(d,1/\eta)$.
\end{thm}

We give a more explicit construction for the vector $\tilde{\vect v}^*$ of \Cref{theorem:influential}. To this end, we will make use of the following definitions.

\begin{defn}
For any $\vect v\in\R^{d}$ and a pair of distinct $i,j\in[d]$, let
$\vect v_{i\leftrightarrow j}$ be defined as the vector one obtains
by swapping the absolute values of the weights $v_{i}$ and $v_{j}$
but keeping their signs. Formally:
\[
\left(\vect v_{i\leftrightarrow j}\right)_{\ell}=\begin{cases}
v_{\ell} & \text{if \ensuremath{\ell\notin\left\{ i,j\right\} ,}}\\
\sign(v_{i})\abs{v_{j}} & \text{if }\ell=i,\\
\sign(v_{j})\abs{v_{i}} & \text{if }\ell=j.
\end{cases}
\]
If $i=j$, then $\vect v_{i\leftrightarrow j}$ is defined to be equal
to $\vect v$. 
\end{defn}

The following definition extends the notion of $\vect v_{i\leftrightarrow j}$
to larger permutations on $[d]$:
\begin{defn}
\label{def: permuting abs values of weights}Let $\pi:\left[d\right]\rightarrow[d]$
be a permutation, then for a vector $\vect v$ we define a vector
$\vect v^{\pi}$ as follows 
\[
\left(\vect v^{\pi}\right)_{i}=\sign(v_{i})\abs{v_{\pi(i)}}.
\]
For a linear threshold function $g(\vect x)=\sign(\vect v\cdot  \vect x +\tau)$ we define
\[
g^\pi(\vect x)=\sign(\vect v^\pi\cdot \vect x +\tau)
\]
\end{defn}

To state our result, we will need some observations regarding weights of a halfspace and influences of the coordinates:
\begin{defn}
For a vector $\vect v\in\R^{d}$, we denote $\Inf_{i}\left[\vect v\right]$
the influence of the function $\sign\left(\vect v\cdot\vect x\right)$
defined as 
\[
\Inf_{i}\left[\vect v\right]=\Pr_{\vect x\sim\left\{ \pm1\right\} ^{d}}\left[\sign\left(\vect v\cdot\vect x\right)\neq\sign\left(\vect v\cdot\vect x^{\otimes i}\right)\right],
\]
where $\vect x^{\otimes i}$ denotes the vector $\vect x$ with the
$i$-th coordinate negated. More generally, for a function $g:\{\pm 1\}^d\rightarrow \{\pm 1\}$ we define
\[
\Inf_{i}\left[g\right]=\Pr_{\vect x\sim\left\{ \pm1\right\} ^{d}}\left[g(\vect x)\neq g\left(\vect x^{\otimes i}\right)\right].
\]
\end{defn}

We provide the following result, as we will show shortly, implies \Cref{theorem:influential}.

\begin{lem}
\label{lem: almost-maximal influences are ok}Let $\etatwo\in(0,1)$,
let $\vect v$ be a vector on $\R^{d}$, let $i_{1},i_{2},\cdots i_{k}\in[d]$
be the indices with the $k$ largest weights for the linear classifier $g(\vect x)=\sign\left(\vect v\cdot\vect x+\tau\right)$,
i.e.
\[
i_{j}\leftarrow\argmax_{i\in[d]\setminus\left\{ i_{1},i_{2},\cdots i_{j-1}\right\} }\left[\abs{v_{i}}\right],
\]
and let $\widehat{i}_{1},\widehat{i}_{2},\cdots\widehat{i}_{k}\in[d]$ satisfy
\[
\Inf_{\widehat{i}_{j}}\left[g\right]\geq\max_{i\in[d]\setminus\left\{ \widehat{i}_{1},\widehat{i}_{2},\cdots\widehat{i}_{j-1}\right\} }\Inf_{i}\left[g\right]-\etatwo,
\]
then there exists a permutation $\pi:\left[d\right]\rightarrow[d]$
such that:
\begin{enumerate}
\item $\pi$ permutes only indices inside $\left\{ i_{1},i_{2},\cdots i_{k}\right\} \cup \{ \widehat{i}_{1},\widehat{i}_{2},\cdots\widehat{i}_{k}\}$.
\item For all $j\in[k]$, it holds that $\pi\left(\widehat{i}_{j}\right)=i_{j}$
and therefore 
\begin{equation}
\left(\vect v^{\pi}\right)_{\widehat{i}_{j}}=\sign\left(v_{\widehat{i}_{j}}\right)\abs{v_{i_{j}}}.\label{eq: weights after permuting}
\end{equation}
\item We have $\Pr_{\vect x\sim\left\{ \pm1\right\} ^{d}}\left[g(\vect x)\neq g^\pi(\vect x)\right]\leq\etatwo k$.
\end{enumerate}
\end{lem}

Before proving \Cref{lem: almost-maximal influences are ok}, we state the following observation regarding influences of a
halfspace:
\begin{claim}
For any $\vect v\in\R^{d}$, $\tau\in \R$ and $i,j\in[d]$, if it is the case that
$\abs{v_{i}}\geq\abs{v_{j}}$ then the halfspace $f=\sign(\vect v \cdot \vect x +\tau)$ satisfies
\[
\Inf_{i}\left[f\right]\geq\Inf_{j}\left[ f\right]
\]
\label{claim: swapping weights vs influences}
Furthermore, it holds that 
\[
\Pr_{\vect x\sim\left\{ \pm1\right\} ^{d}}\left[\sign\left(\vect v\cdot\vect x+\tau\right)\neq\sign\left(\vect v_{i\leftrightarrow j}\cdot\vect x+\tau\right)\right]\leq\abs{\Inf_{i}\left[ f\right]-\Inf_{j}\left[ f\right]}.
\]
\end{claim}

\begin{proof}
Without loss of generality, assume that $i=1$, $j=2$ and $v_1 \geq v_2 \geq 0$. 

First, we consider the case when $d=2$. By inspecting all functions $\{\pm 1\}^2\rightarrow \{\pm 1\}$, we see that $f(x_1,x_2)=\sign(v_1 x_1+v_2 x_2+\tau)$ is one of the following functions:
\begin{itemize}
    \item $f_1$ that equals to  $-1$ on all of $\{\pm 1\}^2$.
    \item $f_2$ that equals to  $+1$ on $\{+1,+1\}$ and equals to  $-1$ on all other points of $\{\pm 1\}^2$.
    \item $f_3$ that equals to  $+1$ on $\{+1,+1\}$ and $\{+1,-1\}$, and equals to  $-1$ on all other points of $\{\pm 1\}^2$. 
     \item $f_4$ that equals to  $+1$ on $\{+1,+1\}$, $\{+1,-1\}$ and $\{-1,+1\}$, and equals to  $-1$ on $\{-1,-1\}$.
     \item $f_5$ that equals to  $-1$ on $\{+1,+1\}$, $\{+1,-1\}$ and $\{-1,+1\}$, and equals to  $-1$ on all of $\{\pm 1\}^2$.
\end{itemize}
Note that $\sign(v_1 x_1+v_2 x_2+\tau)$ with $v_1 \geq v_2 \geq 0$ cannot represent  the function $g$ that equals to  $+1$ on $\{+1,+1\}$ and $\{-1,+1\}$, and equals to  $-1$ on all other points of $\{\pm 1\}^2$. This is the case because $g(-1,+1)>g(+1,-1)$ whereas $\sign(-v_1 +v_2+\tau )\leq \sign(v_1 -v_2+\tau )$ because $v_1\geq v_2\geq 0$. Thus, $f(x_1,x_2)$ has to be one of the five functions above.

We see that for each function $f$ among the functions in $\{f_1,f_2,f_3,f_4,f_5\}$ we have 
\begin{equation}
\label{eq: influence ineq for two-d functions}
\Inf_{1}[f]=
\Pr_{\vect x\sim\left\{ \pm1\right\} ^{2}}\left[f(\vect x) \neq f(\vect x^{\otimes 1})\right]
\geq 
\Pr_{\vect x\sim\left\{ \pm1\right\} ^{2}}\left[f(\vect x) \neq f(\vect x^{\otimes 2})\right]
=\Inf_{2}[f]
,
\end{equation}
and, again by considering all functions in $\{f_1,f_2,f_3,f_4,f_5\}$ that $f(x)=\sign(v_1x_2+v_2x_2+\tau)$ can represent, we check that for each of them
it holds that 
\begin{align}
\label{eq: flipping weights for 2-d}
\Pr_{\vect x\sim\left\{ \pm1\right\} ^{2}}\left[\sign\left(v_1x_1+v_2x_2+\tau\right)\neq\sign\left(v_2x_1+v_1x_2+\tau\right)\right] &=\Pr_{\vect x\sim\left\{ \pm1\right\} ^{2}}\left[f(x_1,x_2)\neq f(x_2,x_1)\right] \notag\\
&\leq
\Inf_{1}\left[ f\right]-\Inf_{2}\left[ f\right].
\end{align}
finishing the proof for $d=2$.

Now we consider the case of general $d$ (but still, without loss of generality, assume $i=1$, $j=2$ and $v_1\geq v_2 \geq 0$). Let $f^{x_3,\cdots x_d}$ be the
restriction of $f$ with specific settings of $x_3,\cdots x_d$ (while keeping $x_1$ and $x_2$ free). We have
\[
\Inf_{1}[f]
=\E_{x_3,\cdots x_d\sim\{\pm 1\}^{d-2}}
[
\Inf_{1}\left[f^{x_3,\cdots x_d}\right]
]
\]
\[
\Inf_{2}[f]
=\E_{x_3,\cdots x_d\sim\{\pm 1\}^{d-2}}
[
\Inf_{2}\left[f^{x_3,\cdots x_d}\right]
]
\]
Applying Equation \ref{eq: influence ineq for two-d functions} to each two-dimensional linear threshold function $f^{x_3,\cdots x_d}$ we get $\Inf_{1}[f]\geq \Inf_{2}[f]$ as required.

Combining the two equalities above, and using Equation \ref{eq: flipping weights for 2-d}
we further have 
\begin{align*}
\Inf_{1}[f]-\Inf_{2}[f]
&=\E_{x_3,\cdots x_d\sim\{\pm 1\}^{d-2}}
[
\Inf_{1}\left[f^{x_3,\cdots x_d}
\right]-
\Inf_{2}\left[f^{x_3,\cdots x_d}\right]
]
\\
&\ge \E_{x_3,\cdots x_d\sim\{\pm 1\}^{d-2}}
\left[
\Pr_{\vect x\sim\left\{ \pm1\right\} ^{2}}\left[f^{x_3,\cdots x_d}(x_1,x_2)\neq f^{x_3,\cdots x_d}(x_2,x_1)\right]
\right]\\
&=\Pr_{\vect x\sim\left\{ \pm1\right\} ^{d}}\left[\sign\left(\vect v\cdot\vect x+\tau\right)\neq\sign\left(\vect v_{i\leftrightarrow j}\cdot\vect x+\tau\right)\right],
\end{align*}
finishing the proof of the claim.
\end{proof}
The claim above allows us to use an alternative characterization of
$i_{1},i_{2},\cdots i_{k}$ in Lemma \ref{lem: almost-maximal influences are ok}
\[
i_{j}=\argmax_{i\in[d]\setminus\left\{ i_{1},i_{2},\cdots i_{j-1}\right\} }\left[\Inf_{i}\left[ g\right]\right].
\]

To prove Lemma \ref{lem: almost-maximal influences are ok},
we also will need the following basic observation:
\begin{claim}
\label{claim: swapping almost-sorted sequences of integers}Suppose
we are given a sequence of positive numbers $\left\{ b_{1},b_{2},\cdots,b_{d}\right\} $
such that for every $\ell\in[k]$ we have 
\begin{equation}
b_{\ell}\geq\max_{i\in[d]\setminus\left\{ 1,2,\cdots,\ell-1\right\} }\left[b_{i}\right]-\etatwo.\label{eq: invariant C inductive-1}
\end{equation}
suppose we obtain a new sequence $\left\{ b_{1}',b_{2}',\cdots,b_{d}'\right\} $
by swapping a pair $b_{j_{1}}$ and \textbf{$b_{j_{2}}$ }for which
$j_{2}>j_{1}$ and $b_{j_{2}}\geq b_{j_{1}}$ (and keeping all other
values the same). Then for every $\ell$ in $\ell\in[k]$ we again
have 
\begin{equation}
b_{\ell}'\geq\max_{i\in[d]\setminus\left\{ 1,2,\cdots,\ell-1\right\} }\left[b_{i}'\right]-\etatwo.\label{eq: invariant C inductive-1-1}
\end{equation}
\end{claim}

\begin{proof}
The claim follows by considering the following cases for $\ell\in[k]$
:
\begin{itemize}
\item If $\ell>j_{2}$ or $\ell<j_{1}$, both sides of Equation \ref{eq: invariant C inductive}
are the same as for Equation \ref{eq: invariant C inductive-1}.
\item If $\ell=j_{1}$, we have
\[
b_{\ell}'=b_{j_{1}}'=b_{j_{2}}\geq b_{j_{1}}\geq\max_{i\in[d]\setminus\left\{ 1,2,\cdots,\ell-1\right\} }\left[b_{i}\right]-\etatwo\geq\max_{i\in[d]\setminus\left\{ 1,2,\cdots,\ell\right\} }\left[b_{i}'\right]-\etatwo,
\]
which implies 
\[
b_{\ell}'\geq\max_{i\in[d]\setminus\left\{ 1,2,\cdots,\ell-1\right\} }\left[b_{i}'\right]-\etatwo.
\]
\item If $\ell\in\left(j_{1},j_{2}\right)$ we have 
\[
b_{\ell}'=b_{\ell}\geq\max_{i\in[d]\setminus\left\{ 1,2,\cdots,\ell-1\right\} }\left[b_{i}\right]-\etatwo\geq\max_{i\in[d]\setminus\left\{ 1,2,\cdots,\ell-1\right\} }\left[b_{i}'\right]-\etatwo.
\]
\item If $\ell=j_{2}$, we have 
\[
b_{j_{2}}'=b_{j_{1}}\geq\max_{i\in[d]\setminus\left\{ 1,2,\cdots,j_{1}-1\right\} }\left[b_{i}\right]-\etatwo\geq\max_{i\in[d]\setminus\left\{ 1,2,\cdots,j_{2}-1\right\} }\left[b_{i}\right]-\etatwo\geq\max_{i\in[d]\setminus\left\{ 1,2,\cdots,j_{2}-1\right\} }\left[b_{i}'\right]-\etatwo.
\]
\end{itemize}
We have concluded the proof of \Cref{claim: swapping almost-sorted sequences of integers}.
\end{proof}
Now, we are ready to to prove \Cref{lem: almost-maximal influences are ok}.
\begin{proof}[Proof of Lemma \ref{lem: almost-maximal influences are ok}.]
Without loss of generality, assume that $\widehat{i}_{j}=j$. We will
use the procedure \ref{algorithm:permutation-procedure} to form $\pi$.

\begin{algorithm}[ht]
\caption{Permutation Procedure}\label{algorithm:permutation-procedure}
    \KwIn{Vector $\vect v\in \R^d$, set of indices $\{i_1,\dots,i_k\}\subseteq [d]$}
\KwOut{Vectors $\vect v^1,\dots, \vect v^k\in \R^d$}
\BlankLine

Set $\vect v^0 \leftarrow \vect v$\;
Set $i_\ell^0 \leftarrow i_\ell$ for all $\ell \in [k]$; \Comment{These values keep track of where we mapped indices in $\left\{ i_{1},\cdots,i_{k}\right\} $}\\

\BlankLine

\For{$j = 1,2,\dots, k$}{
Set $a_{j}\leftarrow i_{j}^{j-1}$; \Comment{This is where the permutation which
we built so far maps $i_{j}$}\\
Set $\vect v^{j}\leftarrow\left(\vect v^{j-1}\right)_{a_{j}\leftrightarrow j}$\;
Set $i_{j}^{j}\leftarrow j$; \textbf{If} $j=i_{\ell}^{j-1}$ for some $\ell$ \textbf{then} set $i_{\ell}^{j}\leftarrow a_{j}$; \Comment{Keeping track of the swap}\\
For all $\ell\in[k]\setminus\left\{ a_{j},j\right\} $, set $i_{\ell}^{j}\leftarrow i_{\ell}^{j-1}$\;
}
\end{algorithm}
Procedure \ref{algorithm:permutation-procedure} implicitly defines a sequence of permutations
$\pi^{0},\pi^{1},\cdots,\pi^{k}$ for which $\vect v^{j}=\vect v^{\pi^{j}}.$ We also have a sequence of linear classifiers $g^0, g^1, \cdots g^k$ for which
\[
g^j( \vect x )=\sign(\vect v^j \cdot \vect x +\tau)=g^{\pi^j}(\vect x).
\]
From the construction it is immediate that the procedure only permutes
indices inside $\left\{ i_{1},i_{2},\cdots i_{k}\right\} \cup\left\{ 1,2,\cdots,k\right\} ,$
and does not permute any other indices.

We claim that the following invariants hold throughout the process
above:
\begin{itemize}
\item Invariant A: for every step $j$ and index $\ell\in[k]$ we have $i_{\ell}^{j}=\pi^{j}(\ell)$.
\item Invariant B: At every step $j$, for $\ell\in\left\{ 1,\cdots,j\right\} $
we have $i_{\ell}^{j}=\ell$, it is the case that
\[
\left(\vect v^{j}\right)_{\ell}=\sign\left(v_{\ell}\right)\abs{v_{i_{\ell}}},
\]
and additionally for $\ell\in\left\{ j+1,\cdots,k\right\} $ we have
$i_{\ell}^{j}>\ell$.
\item Invariant C: At every step $j$, for every $\ell\in[k]$ we have 
\begin{equation}
\Inf_{\ell}\left[g^{j}\right]\geq\max_{i\in[d]\setminus\left\{ 1,2,\cdots,\ell-1\right\} }\left[\Inf_{i}\left[g^{j}\right]\right]-\etatwo.\label{eq: invariant C}
\end{equation}
\end{itemize}
Invariants A and B follow directly by inspecting the procedure and
induction on $j$. Invariant B in particular implies that for all
$j\in[k]$, it holds that $\pi^{j}\left(j\right)=i_{j}$ as required
in part (2) of the claim (then Equation \ref{eq: weights after permuting}
follows by Definition \ref{def: permuting abs values of weights}). 

Invariant C also can be shown by induction as follows. First of all,
the base case $j=0$ holds by the premise of the claim. At every step
$j$ we obtain $\vect v^{j}$ as $\left(\vect v^{j-1}\right)_{a_{j}\leftrightarrow j}$.
Invariant A allows us to conclude that $\Inf_{a}\left[g^{j-1}\right]=\Inf_{i_{j}}\left[\vect g\right]$
and the inductive assumption implies that for every $\ell\in[k]$
we have 
\begin{equation}
\Inf_{\ell}\left[g^{j-1}\right]\geq\max_{i\in[d]\setminus\left\{ 1,2,\cdots,\ell-1\right\} }\left[\Inf_{i}\left[g^{j-1}\right]\right]-\etatwo.\label{eq: invariant C inductive}
\end{equation}
Since at step $j$ we have $a_{j}=i_{j}^{j-1},$ invariants A and
B (together with the definition of $i_{j}$ in the claim statement)
imply that 
\begin{equation}
\Inf_{a_{j}}\left[g^{j-1}\right]=\Inf_{i_{j}}\left[g\right]=\max_{i\in[d]\setminus\left\{ 1,2,\cdots,j-1\right\} }\left[\Inf_{i}\left[g^{j-1}\right]\right],\label{eq: a-j has largest influence in the tail}
\end{equation}
so in particular since Invariant B implies that $a_{j}\geq j$ we
have
\[
\Inf_{a_{j}}\left[g^{j-1}\right]\geq\Inf_{j}\left[g^{j-1}\right].
\]
The three observations that (a) $a_{j}\geq j$ (b) $\Inf_{a_{j}}\left[g^{j-1}\right]\geq\Inf_{j}\left[g^{j-1}\right]$
and (c) $\vect v^{j}\leftarrow\left(\vect v^{j-1}\right)_{a_{j}\leftrightarrow j}$
allow us to use Claim \ref{claim: swapping almost-sorted sequences of integers}
to conclude that Equation \ref{eq: invariant C} holds (note that
for the case $a_{j}=j$ this follows trivially because then $\vect v^{j}=\vect v^{j-1}$). This establishes Invariant C.

Claim \ref{claim: swapping weights vs influences} implies that for
every $j$ in $[k]$ we have
\[
\Pr_{\vect x\sim\left\{ \pm1\right\} ^{d}}\left[
g^{j-1}\left(\vect x\right)
\neq
g^{j}\left(\vect x\right)
\right]\leq\abs{\Inf_{j}\left[g^{j-1}\right]-\Inf_{a_{j}}\left[g^{j-1}\right]}.
\]
By Equation \ref{eq: a-j has largest influence in the tail}: $\Inf_{a_{j}}\left[\vect v^{j-1}\right]\geq\Inf_{j}\left[\vect v^{j-1}\right]$. By Invariant C we have $\Inf_{j}\left[g^{j-1}\right]\geq\Inf_{a_{j}}\left[g^{j-1}\right]-\etatwo$,
which implies that 
\[
\Pr_{\vect x\sim\left\{ \pm1\right\} ^{d}}\left[g^{j-1}\left(\vect x\right)
\neq
g^{j}\left(\vect x\right)\right]\leq\etatwo.
\]
Summing over $j$, we obtain:
\begin{align*}
\Pr_{\vect x\sim\left\{ \pm1\right\} ^{d}}
\left[g\left(\vect x\right)\neq  g^{\pi}\left(\vect x\right)\right]
&=
\Pr_{\vect x\sim\left\{ \pm1\right\} ^{d}}
\left[g^{\pi^0}\left(\vect x\right)\neq  g^{\pi^k}\left(\vect x\right)\right]\\
&\le
\sum_{j=1}^{k}\Pr_{\vect x\sim\left\{ \pm1\right\} ^{d}}\left[g^{j-1}\left(\vect x \right)\neq g^{j}\left(\vect x \right)\right] \\
&\leq k\etatwo,
\end{align*}
which concludes the proof of \Cref{lem: almost-maximal influences are ok}.
\end{proof}

We are now ready to prove \Cref{theorem:influential}, by using \Cref{lem: almost-maximal influences are ok}.

\begin{proof}[Proof of \Cref{theorem:influential}]
    By the premise of the theorem, we have that
for all $i$
\[
\hatI_{i}:= \E_{(\x,y)\sim S}[yx_i]=\E_{(\vect x,y)\sim D}[f^*(\vect x)x_{i}]\pm2\etatwo\,.
\]
The following claim is a standard result that can be found, for example, in \cite{o2014analysis}:

\begin{claim}
\label{claim: influences as chow coefficients}For any linear threshold function $g(\vect x)=\sign(\vect v\cdot \x+\tau)$
and $i\in[d]$, it is the case that 
\[
\E_{\vect x\sim\left\{ \pm1\right\} ^{d}}\left[x_{i}\cdot g(\vect x)\right]=\sign\left(v_{i}\right)\Inf_{i}\left[ g\right].
\]
\end{claim}

Therefore, we have the following:
\[
\abs{\hatI_{i}}=\Inf_{i}\left[f^*\right]\pm2\etatwo,
\]
which by the definition of $j_1,\dots,j_k$ tells
us that for every $t\in\left\{ 0,\cdots,k\right\} $
\[
\Inf_{\hati_{t}}\left[f^*\right]\geq\max_{i\in[d]\setminus\left\{ \hati_{1},\hati_{2},\cdots\hati_{t-1}\right\} }\left[\Inf_{i}\left[ f^*\right]\right]-4\etatwo.
\]
The proof follows by applying \Cref{lem: almost-maximal influences are ok} on the choice $(\widehat{i}_1,\dots,\widehat{i}_k) = (j_1,\dots,j_k)$.
\end{proof}

\section{Learning the Heavy Coefficients: The GLM component}\label{section:heavy}

 In this section, we will present the algorithm and proof for the task of learning Generalized Linear Models (GLMs) with run-time scaling polylogarithmically in the magnitude of weights. This is an important sub-routine in our final algorithm for learning halfspaces: it is the main tool we use to find the heavy coefficients (see \Cref{figure:proof-outline-0}). Using the same techniques, we also prove a theorem of independent interest on learning sigmoidal activations with run-time scaling polylogarithmically in the magnitudes of the weights (see \Cref{theorem:sigmoid-main}).
\subsection{Algorithm and Theorem Statement}
Recall the definition of a generalized linear model.
\label{defn:GLM}
\begin{defn}[Generalized linear model (GLM)]
    A distribution $\Dglm$ on $\R^{\Delta}\times \cube{}$ is called a generalized linear model if there exists a monotone function $\sigma$ and a vector $\vect w^*\in \R^{\Delta}$ such that 
    \[
    \E_{(\vect x,y)\sim D}[y\mid \vect x]= \sigma(\vect w^*\cdot \vect x).
    \]
\end{defn}
We consider the class of $\epsilon$-regular GLM's which are important for our application of learning halfspaces. 
\begin{defn}
\label{def: regular activation}We say that a function $\sigma^{\text{reg}
}:\R\rightarrow\left[-1,1\right]$
is $\eps$-regular if for some integer $d$ and unit vector $\vect u\in\R^{d}$
satisfying $\norm{\vect u}_{2}=1$ and $\norm{\vect u}_{4}^2\leq\eps$
we have
\[
\sigma^{\text{reg}
}(z)=\E_{\vect x\sim\left\{ \pm1\right\} ^{d}}\left[\sign\left(z+\vect x\cdot\vect u\right)\right]
\]
\end{defn}

Note that $d$ in the definition above can be arbitrarily large.
The goal of this section is to show the following theorem.
\begin{thm}
\label{cor: the GLM component theorem nonhom}There exists an algorithm ${{\findheavycoefficients}}(\eps,\Delta,w_{\max},D)$
that takes as inputs $\eps\in(0,1)$, positive integer $\Delta$,
a value $w_{\max}\in\R_{\geq1}$, and a sample access to a distribution
$D$ supported on $\left\{ \vect x\in\R^{\Delta}:\norm{\vect x}_2\leq\sqrt{\Delta}\right\} \times\left\{ \pm1\right\} $,
satisfying the following:
\begin{itemize}
\item The algorithm ${{\findheavycoefficients}}$ takes $O\left(\frac{\Delta}{\eps^2}\cdot \left(\log \frac{\Delta}{\eps}\right)^3\right)$
samples from $D$ and runs in time $\poly\left(\frac{\Delta}{\eps}\log w_{\max}\right)$.
\item When $D$ is a GLM with an arbitrary $\R^{\Delta}$-marginal supported
on $\left\{ \vect x\in\R^{\Delta}:\norm{\vect x}_2\leq\sqrt{\Delta}\right\} $,
and an activation $\sigma^{\text{reg}
}(\vect{v^{*}}\cdot\vect x+\tau^*),$ where $\sigma^{\text{reg}
}$
is a $\eps$-regular function (Definition \ref{def: regular activation}), and $\norm{\vect v^{*}}_2+|\tau^*|\leq w_{\max}$. Then
with probability at least $0.95$ ${\findheavycoefficients}$
will outputs $\vect v_{0}\in \R^{\Delta}$ and $t_0 \in \R$ for which 
\[
\E_{(x,y)\sim D}\left[\left(\sigma^{\text{reg}
}(\vect{v^{*}}\cdot\vect x+\tau^*)-\sigma^{\text{reg}
}(\vect v_{0}\cdot\vect x+t_0)\right)^{2}\right]\leq O\left(\eps \cdot  \log (\Delta/\eps)\right).
\]
\end{itemize}
\end{thm}

\begin{algorithm}[ht]
\caption{$\findheavycoefficients(\eps, \Delta, w_{\max}, {D})$}\label{algorithm:find-heavy-coefficients}
    \KwIn{Parameters $\eps\in (0,1), w_{\max} > 0$, and an oracle to independent examples $(\x,y)\sim { D}$}
\KwOut{A pair $(\hatv,\hattau)$ where $\hatv \in \R^\Delta, \hattau\in \R$}
\BlankLine
\tcc{Initialization}
\BlankLine
Let $C\ge 1$ be a sufficiently large universal constant and set $\clippingthreshold = C + C \log\Delta/\eps$\;
Let $S$ be a set of $C \frac{\Delta \clippingthreshold^2}{\eps^2} \log(\Delta \clippingthreshold /\eps) + C$ i.i.d. examples from $D$\;
Let $\psi(z) = \E_{x\sim \Gauss(0,1)}[\sign(x+z)]$ and let $\psiclip: \R \to [-1,1]$ be the function that can be computed as follows for any given $z\in \R$:
\[
    \psiclip(z) =
    \begin{cases}
        \psi(z)\,, \text{ if } |z| \le \clippingthreshold - 1 \\
        \psi(\clippingthreshold - 1) + (1 - \psi(\clippingthreshold - 1)) (z - \clippingthreshold + 1)\,, \text{ if }z\in (\clippingthreshold-1, \clippingthreshold] \\
        \psi(-\clippingthreshold + 1) - (1 + \psi(-\clippingthreshold + 1)) (z - \clippingthreshold + 1)\,, \text{ if }z\in [-\clippingthreshold, -\clippingthreshold+1) \\
        \sign(z)\,, \text{ if } |z|> \clippingthreshold
    \end{cases}
\]
\BlankLine
\tcc{Optimization via Ellipsoid Algorithm}
\BlankLine
Let $\tilde{S} = \{((\x,1) , y) : (\x,y) \in S\}$\;
Compute the solution $\widehat{\vect w} \in \R^{\Delta+1}$ of the following convex program using the ellipsoid algorithm:
\[
    \argmin_{\vect w: \|\vect w\|_2 \le w_{\max}} \sum_{(\tilde{\vect x}, y)\in \tilde{S}} \int_{0}^{\vect w \cdot \tilde{\vect x}} (\psiclip(z) - y) \, dz
\]\\
Set $(\hatv ,\hattau) = \widehat{\vect w}$\;
\end{algorithm}

We highlight a few distinctions between our result and the standard realizable GLM learning results from literature. The most important feature of our algorithm is that run-time is polynomial in $\log w_{\max}$ as opposed to polynomial
in $w_{\max}$. We note that sample complexity does not depend on $w_{\max}$. Another strength of our algorithm is that it does not need to know the function $\sigma$ in advance: it works universally over the whole family of $\eps$-regular $\sigma$. We will exploit this strength in our halfspace learning algorithm. A final distinction is that the activations we consider are not Lipschitz (although we will see that in a strong sense they are well-approximated by a Lipschitz activation). The overall GLM learning algorithm can be found in Algorithm~\ref{algorithm:find-heavy-coefficients}.

\subsection{Proof}
 We first state a homogeneous (no bias term) version of the result. \Cref{cor: the GLM component theorem nonhom} will immediately follow due to distribution free nature of the result.
\begin{thm}
\label{thm: the GLM component theorem}There exists an algorithm ${\AregGLM}(\eps,\Delta,w_{\max},D)$
that takes as inputs $\eps\in(0,1)$, positive integer $\Delta$,
a value $w_{\max}\in\R_{\geq1}$, and sample access to a distribution
$D$ over $\{ \vect x\in\R^{\Delta}:\norm{\vect x}_2\leq\sqrt{\Delta}\} \times\left\{ \pm1\right\} $,
satisfying the following:
\begin{itemize}
\item The algorithm ${\AregGLM}$ takes $O\left(\frac{\Delta}{\eps^2}\cdot \left(\log \frac{\Delta}{\eps}\right)^3\right)$
samples from $D$ and runs in time $\poly\left(\frac{\Delta}{\eps}\log w_{\max}\right)$.
\item When $D$ is a GLM with an arbitrary $\R^{\Delta}$-marginal supported
on $\left\{ \vect x\in\R^{\Delta}:\norm{\vect x}_2\leq\sqrt{\Delta}\right\} $,
and an activation $\sigma^{\text{reg}
}(\vect{v^{*}}\cdot\vect x),$ where $\sigma^{\text{reg}
}$
is $\eps$-regular and $\norm{\vect v^{*}}_2\leq w_{\max}$. Then
with probability at least $0.95$ ${\AregGLM}$
will output a vector $\vect v_{0}$ for which 
\[
\E_{(x,y)\sim D}\left[\left(\sigma^{\text{reg}
}(\vect{v^{*}}\cdot\vect x)-\sigma^{\text{reg}
}(\vect v_{0}\cdot\vect x)\right)^{2}\right]\leq O\left(\eps \cdot  \log (\Delta/\eps)\right).
\]
\end{itemize}
\end{thm}
We will now complete the proof of \Cref{cor: the GLM component theorem nonhom} given the above theorem. 
\begin{proof}[Proof of \Cref{cor: the GLM component theorem nonhom}]
Define a new distribution $D'$ supported on $\R^{\Delta+1}\times \{\pm 1\}$ and formed from $D$ as follows: sample $(\vect x,y)\sim D$ and output $((\vect x,1),y)$. Clearly, $D'$ is a GLM with an $\eps$-regular activation $\sigma(\vect z)\coloneq \sigma^{{\text{reg}}}(\vect w^*\cdot \vect z)$ where $\vect w^*=(\vect v^*,\tau^*)$. The algorithm $\findheavycoefficients(\eps,\Delta,w_{\mathrm{max}}, D)$ (see Algorithm \ref{algorithm:find-heavy-coefficients}) implements the ${\AregGLM}(\eps,\Delta,w_{\mathrm{max}}, D')$ from \Cref{thm: the GLM component theorem}.
Now, applying Theorem~\ref{thm: the GLM component theorem} to $D'$, we obtain a vector $\widehat{\vect w}=(\widehat{\vect v},\widehat{t})$ such that:
\begin{align*}
    \E_{(z,y)\sim D'}[(\sigma(\widehat{\vect w}\cdot \vect z)-\sigma(\vect w^*\cdot \vect z))^2] &=\E_{(x,y)\sim D}\left[\left(\sigma^{\text{reg}
}(\vect{v^{*}}\cdot\vect x+\tau^*)-\sigma^{\text{reg}
}(\widehat{\vect v}\cdot \vect x+\widehat{t})\right)^{2}\right] \\
&\leq O(\eps\cdot \log (\Delta/\eps))
\end{align*}
The bound on runtime and sample complexity follows from the analogous bounds in Theorem~\ref{thm: the GLM component theorem}. 
\end{proof}

In order to prove \Cref{thm: the GLM component theorem}, we prove the following theorem on misspecified GLM learning corresponding to Lipschitz, truncated activations. 
\begin{thm}
\label{thm: misspecified truncated GLM}
Let $\zalphazero\in\R_{\geq1}$
and suppose $\sigma$ is a an activation function $\R\rightarrow[-1,1]$
that satisfies the following three conditions:
\begin{itemize}
\item $\sigma$ is monotone non-decreasing.
\item $\sigma$ is 1-Lipschitz.
\item $\sigma(z)=\sign(z)$ whenever $\abs z\geq\zalphazero$.
\end{itemize}
Then, there exists an algorithm $\Amisspecified^{\sigma}\left(\eps,\zalphazero,w_{\max},D\right)$
that
\begin{itemize}
\item Draws $O\left(\frac{\Delta\zalphazero^{2}}{\eps^{2}}\ln\frac{\Delta\zalphazero}{\eps}\right)$
samples $\left\{ x_{i}\right\} $ from a distribution $D$ over $\left\{ \vect x\in\R^{\Delta}:\norm{\vect x}_2\leq\sqrt{\Delta}\right\} \times\left\{ \pm1\right\} $
and outputs a vector $\vect v\in\R^{\Delta}$.
\item Runs in time $\poly\left(\frac{\zalphazero\Delta}{\eps}\log w_{\max}\right)$.
\item When $D$ is a GLM with an arbitrary $\R^{\Delta}$-marginal supported
on $\left\{ \vect x\in\R^{\Delta}:\norm{\vect x}_2\leq\sqrt{\Delta}\right\} $,
and a monotone activation $\sigma'(\vect{v^{*}}\cdot\vect x),$ where
\begin{itemize}
\item $\norm{\vect v^{*}}_2\leq w_{\max}$, 
\item $\sigma'(z)=\sign(z)$ whenever $\abs z\geq\zalphazero$,
\item for all $z\in\R$ we have $\abs{\sigma'(z)-\sigma(z)}\leq\eps$
\end{itemize}
\item[] then with probability at least $0.99$ the algorithm $\Amisspecified^{\sigma}\left(D,\eps,\zalphazero,w_{\max}\right)$
will output a vector $\vect v_{0}$ satisfying 
\begin{equation}
\E_{\vect x\sim D_{X}}\left[\left(\sigma'(\vect x\cdot\vect{v^{*}})-\sigma'(\vect x\cdot\vect v_{0})\right)^{2}\right]\leq24\zalphazero\eps\label{eq: goal of thm for slightly mis-specified sigma}
\end{equation}
where $D_{X}$ is the $\R^{\Delta}$-marginal of $D$. (Note that no assumption
is made on $D_{X}$.)
\end{itemize}
\end{thm}

The proof of \Cref{thm: misspecified truncated GLM} follows the lines of Varun Kanade's notes \cite{kanade2018note} on GLM learning through minimizing a special kind of surrogate loss, typically called the matching loss. Our approach deviates from this paradigm in that (1) we use the ellipsoid method for the corresponding convex minimization problem instead of iterative descent to obtain a logarithmic runtime dependence on the parameter $w_{\max}$, and (2) our analysis crucially uses a truncated version of the surrogate loss in order to provide a sample complexity bound that does not depend on $w_{\max}$. 

\begin{proof}
We define the surrogate losses $\ell^{\sigma}$
and $\ell^{\sigma'}$ as follows:

\[
\ell^{\sigma}(\vect v;\vect x,y)=\int_{0}^{\vect v\cdot\vect x}(\sigma(z)-y)\d z.
\]

\[
\ell^{\sigma'}(\vect v;\vect x,y)=\int_{0}^{\vect v\cdot\vect x}(\sigma(z)-y)\d z.
\]
The surrogate loss is convex because $\sigma$ is monotone. The algorithm
$\Amisspecified^{\sigma}\left(\eps,\zalphazero,w_{\max},\dpairs\right)$
does the following:
\begin{itemize}
\item Draw a collection $S$ of $C\frac{\Delta\zalphazero^{2}}{\eps^{2}}\ln\frac{\Delta\zalphazero}{\eps}+C$
samples from $D$, where $C$ is a sufficiently large absolute constant.
\item Using the Ellipsoid algorithm, find
\[
\vect v_{0}\leftarrow\argmin_{\vect v:\ \norm{\vect v}_2\leq w_{\max}}\E_{(\vect x,y)\in S}\left[\ell^{\sigma}(\vect v;\vect x,y)\right]
\]
\item Output $\vect v_{0}$.
\end{itemize}
The run-time is dominated by the second step, which runs in time $\poly\left(\frac{\zalphazero\Delta}{\eps}\log w_{\max}\right)$.
To analyze the algorithm, we define the following auxilliary functions:
\[
\text{trunc}(z):=\min(\zalphazero,\max(-\zalphazero,z))
\]
\[
\ell_{\text{trunc}}^{\sigma}(\vect v;\vect x,y)=\int_{0}^{\text{trunc}(\vect v\cdot\vect x)}(\sigma(z)-y)\d z.
\]

Since $\sigma(z)=\sign(z)$ and $\sigma'(z)=\sign(z)$ whenever $\abs z\geq\zalphazero$,
for every $z\in\R$ we have
\begin{equation}
\sigma(z)=\sigma(\text{trunc}(z))=\sigma'(z).\label{eq: sigma compatible with truncation-1}
\end{equation}
We now make the following observation:
\begin{claim}
\label{claim: truncation does nothing for ground truth surrogate loss-1}For
every $\left(\vect x,y\right)$ in the support of $D$ we have
\[
\ell^{\sigma}(\vect{v^{*}};\vect x,y)=\ell_{\text{trunc}}^{\sigma}(\vect{v^{*}};\vect x,y)\in[-2\zalphazero,2\zalphazero]
\]
\end{claim}

\begin{proof}
The proof is deferred to the appendix where it is identical to the proof of Claim \ref{claim: truncation does nothing for ground truth surrogate loss}.
\end{proof}
Now, since $\abs{\ell^{\sigma}(\vect{v^{*}};\vect x,y)}\leq2\zalphazero$,
we can use the Hoeffding bound to conclude that with probability 0.99 over the
choice of our sample set $S$ we have

\[
\E_{(\vect x,y)\sim S}\left[\ell^{\sigma}(\vect{v^{*}};\vect x,y)\right]=\E_{(\vect x,y)\sim D}\left[\ell^{\sigma}(\vect{v^{*}};\vect x,y)\right]\pm O\left(\frac{\zalphazero}{\sqrt{|S|}}\right),
\]
Since $v_{0}$ minimizes the empirical surrogate loss $\ell^{\sigma}$
among all vectors whose norm is at most $w_{\max}$, and $\norm{\vect v^{*}}_2\leq w_{\max}$
we also have 
\begin{align}
\E_{(\vect x,y)\sim S}\left[\ell^{\sigma}(\vect v_{0};\vect x,y)\right]&\leq\E_{(\vect x,y)\sim S}\left[\ell^{\sigma}(\vect{v^{*}};\vect x,y)\right] \notag\\
&=\E_{(\vect x,y)\sim D}\left[\ell^{\sigma}(\vect{v^{*}};\vect x,y)\right]\pm O\left(\frac{\zalphazero}{\sqrt{|S|}}\right) \notag\\
&=\E_{(\vect x,y)\sim D}\left[\ell^{\sigma}(\vect{v^{*}};\vect x,y)\right]\pm\frac{\eps}{4},\label{eq: the minimizing vector has good empirical surrogate loss-1}
\end{align}
where in the last step we substituted $|S|=C\frac{\Delta\zalphazero^{2}}{\eps^{2}}\ln\frac{\Delta\zalphazero}{\eps}+C$
and took $C$ to be a sufficiently large absolute constant.
We now show some useful properties of the truncated surrogate loss.
\begin{claim}
\label{claim:surrogate_loss_properties}
For every $\vect v,\vect x,y$, the following hold:
   \begin{enumerate}
\item 
$\abs{\ell_{\text{trunc}}^{\sigma}(\vect v;\vect x,y)}\leq\int_{0}^{\text{trunc}(\vect v\cdot\vect x)}\abs{\sigma(z)-y}\d z\leq2\zalphazero$
\item 
$\ell_{\text{trunc}}^{\sigma}(\vect v;\vect x,y)\leq\ell^{\sigma}(\vect v;\vect x,y)\label{eq: truncated surrate loss at most surrogate loss-1}$.
\end{enumerate}
\end{claim}
\begin{proof} 
The first property follows immediately from definition of $\text{trunc}$ and the fact that $\sigma$ and $y$ are bounded by $1$. We show the second property in cases:
\begin{itemize}
\item If $\vect v\cdot\vect x\in[-\zalphazero,\zalphazero]$, we have $\text{trunc}(\vect v\cdot\vect x)=\vect v\cdot\vect x$
and we have $\ell_{\text{trunc}}^{\sigma}(\vect v;\vect x,y)=\ell^{\sigma}(\vect v;\vect x,y)$.
\item If $\vect v\cdot\vect x>\zalphazero$ we have 
\[
\ell^{\sigma}(\vect v;\vect x,y)=\int_{0}^{\text{trunc}(\vect v\cdot\vect x)}(\sigma(z)-y)\d z+\int_{\text{trunc}(\vect v\cdot\vect x)}^{\vect v\cdot\vect x}(1-y)\d z\geq\int_{0}^{\text{trunc}(\vect v\cdot\vect x)}(\sigma(z)-y)\d z
\]
\item If $\vect v\cdot\vect x<-\zalphazero$ we have 
\[
\ell^{\sigma}(\vect v;\vect x,y)=\int_{0}^{\text{trunc}(\vect v\cdot\vect x)}(\sigma(z)-y)\d z+\int_{\text{trunc}(\vect v\cdot\vect x)}^{\vect v\cdot\vect x}(-1-y)\d z\geq\int_{0}^{\text{trunc}(\vect v\cdot\vect x)}(\sigma(z)-y)\d z
\]
\end{itemize}
\end{proof}
We also need the following claim.
\begin{claim}
For a sufficiently large value of the absolute constant $C$, with
probability at least $0.999$ over the choice of the set $S$, we
have the following uniform convergence bound for $\ell_{\text{trunc}}^{\sigma}$:
\begin{equation}
\max_{\vect v'\in\R^{d}}\abs{\E_{(\vect x,y)\sim D}\left[\ell_{\text{trunc}}^{\sigma}(\vect v';\vect x,y)\right]-\E_{(\vect x,y)\sim S}\left[\ell_{\text{trunc}}^{\sigma}(\vect v';\vect x,y)\right]}\leq\frac{\eps}{4}\label{eq: uniform convergence for truncated surrogate loss-1}
\end{equation}
\end{claim}

\begin{proof}
The proof is exactly the same as in the case of Equation \ref{eq: uniform convergence for truncated surrogate loss} and is deferred to the appendix.
\end{proof}

For the GLM $D$,
define $D_{y\vert \x}$ to be the  distribution of the label $y$ conditioned on a specific point $\vect x \in \R^d$. 
We now compare the value of the truncated surrogate loss for a candidate vector $\vect v$ and the optimal vector $\vect v^*$ 
as follows (this argument is similar to the argument in \cite{kanade2018note}):
\begin{align}
\E_{y\sim D_{y\vert \x}}&\left[\ell_{\text{trunc}}^{\sigma}(\vect v;\vect x,y)\right]-\E_{y\sim D_{y\vert \x}}\left[\ell_{\text{trunc}}^{\sigma}(\vect{v^{*}};\vect x,y)\right]=\E_{y\sim D_{y\vert \x}}\left[\int_{\text{trunc}(\vect{v^{*}}\cdot\vect x)}^{\text{trunc}(\vect v\cdot\vect x)}(\sigma(z)-y)\d z\right] \notag\\
&=\int_{\text{trunc}(\vect{v^{*}}\cdot\vect x)}^{\text{trunc}(\vect v\cdot\vect x)}(\sigma(z)-\sigma'(\vect{v^{*}}\cdot\vect x))\d z=\int_{\text{trunc}(\vect{v^{*}}\cdot\vect x)}^{\text{trunc}(\vect v\cdot\vect x)}(\sigma(z)-\sigma'(\text{trunc}(\vect{v^{*}}\cdot\vect x)))\d z\notag\\
&\geq\int_{\text{trunc}(\vect{v^{*}}\cdot\vect x)}^{\text{trunc}(\vect v\cdot\vect x)}(\sigma(z)-\sigma(\text{trunc}(\vect{v^{*}}\cdot\vect x)))\d z-\int_{\text{trunc}(\vect{v^{*}}\cdot\vect x)}^{\text{trunc}(\vect v\cdot\vect x)}\abs{\sigma(\text{trunc}(\vect{v^{*}}\cdot\vect x))-\sigma'(\text{trunc}(\vect{v^{*}}\cdot\vect x))}\d z\notag\\
&\geq\int_{\text{trunc}(\vect{v^{*}}\cdot\vect x)}^{\text{trunc}(\vect v\cdot\vect x)}(\sigma(z)-\sigma(\text{trunc}(\vect{v^{*}}\cdot\vect x)))\d z-2\zalphazero\eps\notag\\
&\geq\frac{1}{2}\left(\sigma(\text{trunc}(\vect v\cdot\vect x))-\sigma(\text{trunc}(\vect{v^{*}}\cdot\vect x)))\right)^{2}-2\zalphazero\eps=\frac{1}{2}\left(\sigma(\vect v\cdot\vect x)-\sigma(\vect{v^{*}}\cdot\vect x))\right)^{2}-2\zalphazero\eps.\label{eq: kanade notes step-1}
\end{align}
The first equality is due to linearlity of expectation and the definition
of $\ell_{\text{trunc}}^{\sigma}$.
 The second equality is because $D_{y\vert \x}$ is the conditional probability distribution of a GLM with vector $\vect{v^{*}}$
and activation $\sigma'$.
The third equality is due to Equation \ref{eq: sigma compatible with truncation-1}.
The fourth inequality is due to the triangle inequality.
The fifth inequality uses the premise that $\abs{\sigma(z')-\sigma'(z')}\leq\eps$
for all $z'\in\R$.
 The inequality $\int_{\text{trunc}(\vect{v^{*}}\cdot\vect x)}^{\text{trunc}(\vect v\cdot\vect x)}(\sigma(z)-\sigma(\text{trunc}(\vect{v^{*}}\cdot\vect x)))\d z\geq\frac{1}{2}\left(\sigma(\text{trunc}(\vect v\cdot\vect x))-\sigma(\text{trunc}(\vect{v^{*}}\cdot\vect x)))\right)^{2}$
is due to $\sigma$ being monotonically increasing and $1$-Lipschitz.
 The last equality is due to Equation \ref{eq: sigma compatible with truncation-1}.
 
We therefore have for any vector $\vect v\in\R^{d}$
\begin{align*}
\frac{1}{2}\E_{(\vect x,y)\sim D}\left[\left(\sigma(\vect v\cdot\vect x)-\sigma(\vect{v^{*}}\cdot\vect x))\right)^{2}\right]-2\zalphazero\eps&\leq\E_{(\vect x,y)\sim D}\left[\ell_{\text{trunc}}^{\sigma}(\vect v;\vect x,y)-\ell_{\text{trunc}}^{\sigma}(\vect{v^{*}};\vect x,y)\right]\\
&=\E_{(\vect x,y)\sim D}\left[\ell_{\text{trunc}}^{\sigma}(\vect v;\vect x,y)-\ell^{\sigma}(\vect{v^{*}};\vect x,y)\right]\\
&\leq\E_{(\vect x,y)\sim S}\left[\ell_{\text{trunc}}^{\sigma}(\vect v;\vect x,y)\right]-\E_{(\vect x,y)\sim D}\left[\ell^{\sigma}(\vect{v^{*}};\vect x,y)\right]+\eps/4\\
&\leq\E_{(\vect x,y)\sim S}\left[\ell^{\sigma}(\vect v;\vect x,y)\right]-\E_{(\vect x,y)\sim D}\left[\ell^{\sigma}(\vect{v^{*}};\vect x,y)\right]+\eps/4.
\end{align*}
The first step uses Equation \ref{eq: kanade notes step-1}. The second uses Claim \ref{claim: truncation does nothing for ground truth surrogate loss-1}. The third uses Equation \ref{eq: uniform convergence for truncated surrogate loss-1} and the fourth step uses Claim~\ref{claim:surrogate_loss_properties} (2).
Finally, combining with Equation \ref{eq: the minimizing vector has good empirical surrogate loss-1},
we have
\begin{equation}
\label{eq: closeness in sigma distance  non-misspecified}
\E_{\vect x\sim D_{X}}\left[\left(\sigma(\vect x\cdot\vect{v^{*}})-\sigma(\vect x\cdot\vect v_{0})\right)^{2}\right]\leq8\zalphazero\eps
\end{equation}
Recalling that
for all $z\in\R$ we have $\abs{\sigma'(z)-\sigma(z)}\leq\eps$, which implies that for every $z_1$ and $z_2$ in $\R$ we have 
\begin{align*}
|
(
\sigma'(z_1)-&\sigma'(z_2)
)^2
-
\left(
\sigma(z_1)-\sigma(z_2)
\right)^2
|
\\
&\leq 2\abs{
\sigma'(z_1)-\sigma(z_1)
-(\sigma'(z_2)-\sigma(z_2))
}
\cdot
\abs{
\sigma'(z_1)-\sigma'(z_2)
+(\sigma(z_1)-\sigma(z_2))
}
\leq 
16 \eps. 
\end{align*}
Setting $z_1=\vect x\cdot\vect{v^{*}}$, $z_2=\vect x\cdot\vect v_{0}$, and
combining with Equation \ref{eq: closeness in sigma distance  non-misspecified}
yields the desired Equation \ref{eq: goal of thm for slightly mis-specified sigma}.
\end{proof}
\paragraph{Putting it all together: the case of  $\eps$-regular activation $\sigma^{\text{reg}}$.}
We are now ready to prove \Cref{thm: the GLM component theorem}. In \Cref{thm: misspecified truncated GLM}, we gave a polynomial time learning algorithm for GLMs which were misspecified, truncated and had weights with magnitued potentially exponential in dimension. Recall that the goal of this section was to learn $\epsilon$-regular GLMs, which naturally arose in the context of our application to halfspace learning. We now argue that learning $\epsilon$-regular GLMs can be reduced to learning mis-specified GLMs with activation $\psi(z)\coloneq \E_{x\sim \Gauss(0,1)}[\sign(z+x)]$.
\begin{proof}
Recall that in the setting of Theorem \ref{thm: the GLM component theorem}, the algorithm is given inputs $\eps\in(0,1)$, positive integer $\Delta$,
a value $w_{\max}\in\R_{\geq1}$, and a sample access to a distribution
$D$ supported on $\{ \vect x\in\R^{\Delta}:\norm{\vect x}_2\leq\sqrt{\Delta}\} \times\left\{ \pm1\right\} $. 
The distribution $D$ is a GLM with an arbitrary $\R^{\Delta}$-marginal supported
on $\{ \vect x\in\R^{\Delta}:\norm{\vect x}_2\leq\sqrt{\Delta}\} $,
and an activation $\sigma^{\text{reg}
}(\vect{v^{*}}\cdot\vect x),$ where $\sigma^{\text{reg}
}$
is $\eps$-regular and $\norm{\vect v^{*}}_2\leq w_{\max}$. With probability at least $0.95$ ${\AregGLM}$
the algorithm should output a vector $\vect v_{0}$ for which 
\begin{equation}
    \label{eq: goal for GLM component}
\E_{(x,y)\sim D}\left[\left(\sigma^{\text{reg}
}(\vect{v^{*}}\cdot\vect x)-\sigma^{\text{reg}
}(\vect v_{0}\cdot\vect x)\right)^{2}\right]\leq O\left(\eps \cdot \log (\Delta/\eps)\right).
\end{equation}

We now state the algorithm ${\AregGLM}$ we use for Theorem \ref{thm: the GLM component theorem}. The algorithm ${\AregGLM}(\eps,\Delta,w_{\max},D)$ will use the algorithm $\Amisspecified$ from Theorem \ref{thm: misspecified truncated GLM}, specifically running $\Amisspecified^{\psiclip}\left(4\eps,\rho,w_{\max},D\right)$ with the choice of activation $\sigma=\psiclip$ and parameter $\rho$ to be defined below. Afterwards, the algorithm ${\AregGLM}$ outputs the same vector $\widehat{v}$ that it received from $\Amisspecified$.

Below, we describe the function $\psiclip: \R \rightarrow[-1,1]$ used in our algorithm and prove Theorem \ref{thm: the GLM component theorem}. 

Recall that, since $\sigma^{\text{reg}}$ is $\eps$-regular (Definition \ref{def: regular activation}), for some integer $d$ and unit vector $\vect u\in\R^{d}$
satisfying $\norm{\vect u}_{2}=1$ and $\norm{\vect u}_{4}^2\leq\eps$
we have
\begin{equation}
\label{eq: epsilon-regular activation}
\sigma^{\text{reg}
}(z)=\E_{\vect x\sim\left\{ \pm1\right\} ^{d}}\left[\sign\left(z+\vect x\cdot\vect u\right)\right]
\end{equation}
In addition to activation  $\sigma^{\text{reg}}$, we define the following function
\[
\psi(z)=\E_{x\sim \N(0,1)}\left[\sign\left(z+x\right)\right].
\]
Since $\sigma^{\text{reg}}$ is an $\eps$-regular activation, the Berry-Esseen theorem (Fact \ref{fact: Berry-Esseen for boolean}) implies that
\begin{equation}
    \label{eq: Berry-Esseen for reg activation}
    \max_{z
    \in \R}
    \abs{\psi(z)-\sigma^{\text{reg}}(z)} \leq 2\eps
\end{equation}

Now, we are ready to define the function $\psiclip$ that is used in our algorithm. Let $\clippingthreshold = \setalphazero$ be a real-valued parameter, where $C$ is a sufficiently large absolute constant. Define the following ``clipped'' versions of $\sigma^{\text{reg}}$ and  $\psi$:
\begin{equation}
    \sigma^{\text{reg}}_{\text{clipped}}(z)
    =
    \begin{cases}
        \sigma^{\text{reg}}(z) & \text{if} \abs{z}\leq \clippingthreshold \\
        \sign(z) & \text{if} \abs{z} > \clippingthreshold 
    \end{cases}
\end{equation}
\begin{equation}
    \psiclip(z)
    =
    \begin{cases}
        \psi(z) & \text{if} \abs{z}\leq \clippingthreshold-1\\
         \psi(\clippingthreshold-1)+(1-\psi(\clippingthreshold-1))\cdot (z-(\clippingthreshold-1))
         & \text{if } z\in(\clippingthreshold-1, \clippingthreshold]\\
         \psi(-(\clippingthreshold-1))-(1+\psi(-(\clippingthreshold-1)))\cdot (-z-(\clippingthreshold-1))
         & \text{if } z\in[-\clippingthreshold, -\clippingthreshold+1)\\
        \sign(z) &\text{if} \abs{z} > \clippingthreshold 
    \end{cases}
\end{equation}
Additionally, recall that 

The function $\sigma^{\text{reg}}_{\text{clipped}}(z)$ equals to $\sigma^{\text{reg}}(z)$ in $[-(\clippingthreshold-1),(\clippingthreshold-1)]$ and equals to $\sign(z)$ outside of this interval. The function $\psiclip(z)$ likewise equals to $\psi(z)$ in $[-(\clippingthreshold-1),(\clippingthreshold-1)]$ and equals $\sign(z)$  in $\R \setminus[-\clippingthreshold,\clippingthreshold]$. In the intervals $[-\clippingthreshold, -\clippingthreshold+1]$ and $[\clippingthreshold-1, \clippingthreshold]$, the function $\psiclip(z)$ interpolates linearly between the values of $\psi(z)$ and $\sign(z)$ on the respective edges of these two intervals.

Also, note that the function $\psiclip(z)$ is defined without any reference to $\sigma^{\text{reg}}$, and thus one can run the algorithm $\Amisspecified^{\psiclip}\left(4\eps,\rho,w_{\max},D\right)$ without knowing  $\sigma^{\text{reg}}$ in advance. Below, we will show that for any $\eps$-regular $\sigma^{\text{reg}}$, with probability at least $0.95$ the resulting vector $\vect{v}_0$ will satisfy Equation \ref{eq: goal for GLM component}.

Recall, that the distribution $D$ is a GLM with an arbitrary $\R^{\Delta}$-marginal supported
on $\{ \vect x\in\R^{\Delta}:\norm{\vect x}_2\leq\sqrt{\Delta}\} $,
and an activation $\sigma^{\text{reg}
}(\vect{v^{*}}\cdot\vect x)$. 

From Hoeffding inequality and the standard Gaussian tail bound respectively, we have the following two inequalities
\begin{equation}
    \label{eq: tail bound for Gaussian}
    \Pr_{\vect x\sim\left\{ \pm1\right\} ^{d}}\left[\abs{\vect x\cdot\vect u}\geq \clippingthreshold-1\right] \leq 2e^{-(\clippingthreshold-1)^2/2},
\end{equation}
\begin{equation}
    \label{eq: tail bound for Boolean}
    \Pr_{x\sim \N(0,1)}\left[|x|\geq \clippingthreshold-1\right]
    \leq
    2e^{-(\clippingthreshold-1)^2/2},
\end{equation}
which allow us to conclude that
\begin{equation}
    \label{eq: clipping does not change regular activation much}
    \max_{z
    \in \R}
    \abs{\sigma^{\text{reg}}_{\text{clipped}}(z)-\sigma^{\text{reg}}(z)} \leq 2e^{-(\clippingthreshold-1)^2/2}
\end{equation}
\begin{equation}
    \label{eq: clipping does not change Gaussian activation much}
    \max_{z
    \in \R}
    \abs{\psiclip(z)-\psi(z)} \leq 2e^{-(\clippingthreshold-1)^2/2}.
\end{equation}

We define $D_{\text{clipped}}$ to be a GLM with the same $\R^{\Delta}$-marginal as $D$ but activation $\sigma^{\text{reg}}_{\text{clipped}}$. 
We are now ready to use Theorem \ref{thm: misspecified truncated GLM} to prove the following claim about $D_{\text{clipped}}$:
\begin{claim}
\label{claim: running A regular on truncated distribution}
    Let $\vect{v}_0$ be obtained by running $\Amisspecified^{\psiclip}(4\eps,\rho,w_{\max},D_{\text{clipped}})$. Then, with probability at least $0.99$, the vector  $\vect{v}_0$  satisfies Equation \ref{eq: goal for GLM component}.
\end{claim}
\begin{proof}
We see that the choice of $\sigma=\psiclip$ satisfies the requirements on $\sigma$ in Theorem \ref{thm: misspecified truncated GLM}. Indeed,
\begin{itemize}
    \item The function $\psiclip$ is monotone non-decreasing, because so is $\psi$. (Monotonicity of $\psi$ follows directly from the definition $\psi$.)
    \item The function $\psiclip$ is $1$-Lipschitz because, 
    \begin{itemize}
        \item For $z$ in $[-(\clippingthreshold-1),\clippingthreshold-1]$ we have $\psiclip(z)=\psi(z)$ and we have
        \[
        \frac{d}{\d z} \psi(z)
        =
        \frac{1}{\sqrt{2\pi}}
        e^{-z^2/2} \leq 1,
        \]
        and hence $\psiclip$ is $1$-Lipschitz in this interval.
        \item In the interval $(\clippingthreshold-1,\clippingthreshold]$, the function $\psiclip$ is linear, with slope $1-\psi(\clippingthreshold-1)$. Since $\clippingthreshold\geq 1$, it follows from the definition of $\psi$ that $\psi(\clippingthreshold-1)\geq 0$, and thus $\psiclip$ is $1$-Lipschitz in the interval  $(\clippingthreshold-1,\clippingthreshold]$. The case of interval $[-\clippingthreshold, \clippingthreshold-1)$ is fully analogous.
        \item Outside the interval $[-\clippingthreshold, \clippingthreshold]$ we have $\psiclip(z)=\sign(z)$, which is trivially $1$-Lipschitz.
    \end{itemize}
    \item By construction, we have $\psiclip(z)=\sign(z)$ whenever $z$ is not in $[-\clippingthreshold, \clippingthreshold]$. 
\end{itemize}
By combining Equation \ref{eq: clipping does not change regular activation much}, Equation \ref{eq: clipping does not change Gaussian activation much} and Equation \ref{eq: Berry-Esseen for reg activation}, using the triangle inequality and substituting $\clippingthreshold=\setalphazero$ we see that for sufficiently large absolute constant $C$ we have 
\[
\max_{z
    \in \R}
    \abs{\psiclip(z)-\sigma^{\text{reg}}_{\text{clipped}}(z)} \leq 2\eps + 4e^{-(\clippingthreshold-1)^2/2} \leq 4\eps.
\]

Overall, we see that all the premises of Theorem \ref{thm: misspecified truncated GLM} are satisfied for $\sigma=\psiclip$ and $\sigma'=\sigma^{\text{reg}}_{\text{clipped}}$. Therefore, with probability at least $0.99$, the vector $\vect{v}_0$, obtained by running $\Amisspecified^{\psiclip}(4\eps,\rho,w_{\max},D_{\text{clipped}})$, satisfies 
\[
\E_{\vect x\sim D_{X}}\left[\left(\sigma^{\text{reg}}_{\text{clipped}}(\vect x\cdot\vect{v^{*}})-\sigma^{\text{reg}}_{\text{clipped}}(\vect x\cdot\vect v_{0})\right)^{2}\right]\leq 96 \rho \eps.
\]
Combining the inequality above with Equation \ref{eq: clipping does not change regular activation much}, using the basic inequality $\abs{(a+b)^2-a^2}\leq 2\abs{a b}+b^2$, and recalling that $\abs{\sigma^{\text{reg}}(z)}\leq 1$ we see that 
\[
\E_{\vect x\sim D_{X}}\left[\left(\sigma^{\text{reg}}(\vect x\cdot\vect{v^{*}})-\sigma^{\text{reg}}(\vect x\cdot\vect v_{0})\right)^{2}\right]\leq 96 \rho\eps
+8e^{-(\clippingthreshold-1)^2/2}+4e^{-(\clippingthreshold-1)^2}.
\]
Substituting $\clippingthreshold=\setalphazero$ we see that for sufficiently large absolute constant $C$ the inequality above implies Equation \ref{eq: goal for GLM component}, which finishes the proof. 
\end{proof}
Finally, we compare the distributions $D$ and $D_{\text{clipped}}$. First, for every specific $z$, the TV distance between a pair of $\{\pm 1\}$-valued random variables with expectations $\sigma^{\text{reg}}_{\text{clipped}}(z)$ and $\sigma^{\text{reg}}(z)$ respectively is bounded by $\abs{\sigma^{\text{reg}}_{\text{clipped}}(z)-\sigma^{\text{reg}}(z)}/2$, which is at most $e^{-(\clippingthreshold-1)^2/2}$ by Equation \ref{eq: clipping does not change regular activation much}. Comparing this with the definitions of $D$ and $D_{\text{clipped}}$ we see that 
\[
\mathrm{dist}_{\mathrm{TV}}
(D,D_{\text{clipped}})
\leq \frac{1}{2} \max_{z
    \in \R}
    \abs{\sigma^{\text{reg}}_{\text{clipped}}(z)-\sigma^{\text{reg}}(z)} \leq e^{-(\clippingthreshold-1)^2/2}.
\]
Theorem \ref{thm: misspecified truncated GLM} implies that the algorithm  $\Amisspecified^{\psiclip}$ uses at most $O\left(\frac{\Delta\rho^{2}}{\eps^{2}}\ln\frac{\Delta\rho}{\eps}\right)$
samples from the distribution it accesses. This, together with the data-processing inequality and the bound above, lets us conclude that
\begin{align*}
\mathrm{dist}_{\mathrm{TV}}
&\left(\Amisspecified^{\psiclip}(4\eps,\rho,w_{\max},D),
\Amisspecified^{\psiclip}(4\eps,\rho,w_{\max},D_{\text{clipped}})
\right)\\
&\leq
O\left(\frac{\Delta\rho^{2}}{\eps^{2}}\ln\frac{\Delta\rho}{\eps}\right)
\mathrm{dist}_{\mathrm{TV}}
(D,D_{\text{clipped}})
\leq 
O\left(\frac{\Delta\rho^{2}}{\eps^{2}}\ln\frac{\Delta\rho}{\eps}\right)
\cdot 
e^{-(\clippingthreshold-1)^2/2}.
\end{align*}
Substituting $\clippingthreshold=\setalphazero$, we see that for sufficiently large absolute constant $C$ the expression above is at most $0.05$. Combining this with Claim \ref{claim: running A regular on truncated distribution}, we see that the vector $\vect{v}_0$ given by $\Amisspecified^{\psiclip}(4\eps,\rho,w_{\max},D)$ with probability at least $0.99-0.05=0.95$ satisfies  Equation \ref{eq: goal for GLM component}, finishing the proof of correctness for the algorithm ${\AregGLM}(\eps,\Delta,w_{\max},D)$ we use for Theorem \ref{thm: the GLM component theorem} (recall ${\AregGLM}(\eps,\Delta,w_{\max},D)=\Amisspecified^{\psiclip}(4\eps,\rho,w_{\max},D)$). 

Finally, we check the required sample complexity and run-time bounds. Referring to Theorem \ref{thm: misspecified truncated GLM}, and substituting $\clippingthreshold=\setalphazero$, we see that the sample complexity of the resulting algorithm is $O\left(\frac{\Delta}{\eps^2}\cdot \left(\log \frac{\Delta}{\eps}\right)^3\right)$, as required by Theorem \ref{thm: the GLM component theorem}. The resulting run-time is $\poly\left(\frac{\Delta}{\eps}\log w_{\max}\right)$, as required by Theorem \ref{thm: the GLM component theorem}.
\end{proof}

\section{Learning the Regular Coefficients}\label{section:regular}

In this section, we give an algorithm for learning the regular tail coefficients once we have a good enough estimate for the head coefficients in the structured case (see \Cref{figure:proof-outline-0}). We provide the following result.

\begin{thm}
\label{thm: with known head vars and weights}
There exists an algorithm
$\findregularcoefficients(\eps,d,H,\vect v_{\mathrm{head}},\tau,D)$ (Algorithm \ref{algorithm:find-regular-coefficients})
that takes as inputs $\eps\in(0,1)$, positive integer $d$, a
set $H\subset[d]$, a vector $\vect v_{\mathrm{head}}\in\R^{d}$ that is supported
on $H$, and a sample access to a distribution $D$ supported on $\left\{ \pm1\right\} ^{d}\times\left\{ \pm1\right\} $
whose $\left\{ \pm1\right\} ^{d}$-marginal is the uniform distribution.
The algorithm runs in time $\poly(d/\eps)$ and outputs, with probability at least $0.995$, a vector
$\widehat{\vect v}\in\R^{d}$ for which 
\begin{equation}
\Pr_{(\vect x,y)\sim D}\left[\sign(\widehat{\vect v}\cdot\vect x+\tau)\neq y\right]\leq O\left(\frac{\opt_{\mathrm{reg}}}{\eps}\log\frac{1}{\eps}+\eps\right),\label{eq: bound with known head vars and weights}
\end{equation}
where $\opt_{\mathrm{reg}}=\opt_{\mathrm{reg}}(H,\eps)$ is defined as follows for ${\cal V}(H,\eps) = \{\vect v\in \R^d: \|\vect v\|_2 = 1, \|\vect v\|_4^2 \le \eps, \vect v_H = 0\}$:
\begin{equation}
\opt_{\mathrm{reg}}:=\min_{\vect v\in {\cal V}(H,\eps)}\Pr_{(\vect x,y)\sim D}\left[\sign(\vect v_{\mathrm{head}}\cdot\vect x+\vect v\cdot\vect x+\tau)\neq y\right]\label{eq: def of opt reg tail}
\end{equation}
\end{thm}

\begin{algorithm}[ht]
\caption{$\findregularcoefficients(\eps, d, H, \vect v_{\mathrm{head}}, \tau, D)$}\label{algorithm:find-regular-coefficients}
    \KwIn{Parameters $\eps\in (0,1), d\in \mathbb{N}$, $ H\subseteq [d], \vect v_{\mathrm{head}}\in \R^d, \tau\in \R$ where $(\vect v_{\mathrm{head}})_{[d]\setminus H} = 0$, and an oracle to independent samples $(\x,y)\sim D$.}
\KwOut{A vector $\hatv\in \R^d$}
\BlankLine
\tcc{Initialization}
\BlankLine
Let $C\ge 1$ be a sufficiently large universal constant\;
Let $S$ be a set of $C(d/\eps)^C$ i.i.d. samples from $D$\;
Let $S' = \{(\x,y)\in S : |\vect v_{\mathrm{head}} \cdot \x + \tau| \le \log 1/\eps\}$\;
\BlankLine
\tcc{Hinge Loss minimization}
\BlankLine
Let $\V = \{\vect v\in\R^d : \|\vect v\|_2 \le 1, \vect v_{H} = 0\}$\;
Compute the solution $\widehat{\vect v}_{\mathrm{tail}} \in \R^{d}$ of the following convex program:
\[
    \argmin_{\vect v\in \V} \sum_{(\x, y)\in S'} \relu(\eps - y (\vect v_{\mathrm{head}}\cdot \x + \vect v \cdot \x + \tau) )
\]\\
Set $\hatv = \widehat{\vect v}_{\mathrm{tail}}+ \vect v_{\mathrm{head}}$\;
\end{algorithm}

\begin{proof}[Proof of \Cref{thm: with known head vars and weights}]
We define the following loss (which
can be interpreted as a modified hinge loss):
\[
\lambda(z)=\begin{cases}
0 & \text{if }z\geq\eps,\\
1-\frac{z}{\eps} & \text{otherwise}.
\end{cases}
\]
Observe that $\lambda(z) = \frac{1}{\eps} \relu(\eps - z)$.
Moreover, we let $\phi:\R^d \to \R$ be the function defined as follows. 
\begin{equation}
    \phi(\vect x) = \vhead\cdot \vect x + \tau\label{equation:definition-phi}
\end{equation}
\begin{claim}
\label{claim: bound on hinge-like loss for ground truth tail}Let
$\vtail^{*}$ be a minimizer in Equation \ref{eq: def of opt reg tail},
then with probability at least $0.999$, it is the case that 
\begin{equation}
\frac{1}{|S|}\sum_{(\vect x,y)\in {\Sout}}\left[\lambda\left((\phi(\vect x)+\vtail^{*}\cdot\vect x)y\right)\right]\leq O\left(\frac{\opt_{\mathrm{reg}}}{\eps}\log\frac{1}{\eps}+\eps\right).\label{eq: bound on hingle-like loss}
\end{equation}
\end{claim}

\begin{proof}
We see that $\lambda(z)\leq\frac{\abs z}{\eps}+1$, and therefore
\begin{align*}
\left(\indicator_{\abs{\phi(\vect x)}\leq\log\frac{1}{\eps}}\cdot\lambda\left((\phi(\vect x)+\vtail^{*}\cdot\vect x)y\right)\right)^{2} &\leq\indicator_{\abs{\phi(\vect x)}\leq\log\frac{1}{\eps}}O\left(\frac{\left(\phi(\vect x)\right)^{2}+\left(\vtail^{*}\cdot\vect x\right)^{2}}{\eps^{2}}\right)\\
&\le O\left(\frac{\left(\log\frac{1}{\eps}\right)^{2}+\left(\vtail^{*}\cdot\vect x\right)^{2}}{\eps^{2}}\right).
\end{align*}
Taking the expectations of both sides and recalling that $\norm{\vtail^{*}}_{2}=1$,
we see that 
\[
\E_{(\vect x,y)\sim D}\left[\left(\indicator_{\abs{\phi(\vect x)}\leq\log\frac{1}{\eps}}\cdot\lambda\left((\phi(\vect x)+\vtail^{*}\cdot\vect x)y\right)\right)^{2}\right]=O\left(\frac{\left(\log\frac{1}{\eps}\right)^{2}+1}{\eps^{2}}\right),
\]
which allows us to use the Chebyshev's inequality to conclude that
if $C$ is a sufficiently large absolute constant the $C\left(\frac{d}{\eps}\right)^{C}$-sized
sample set $S$ with probability at least $0.999$ satisfies
\[
\frac{1}{|S|}\sum_{(\vect x,y)\in {\Sout}}\left[\lambda\left((\phi(\vect x)+\vtail^{*}\cdot\vect x)y\right)\right]=\E_{(\vect x,y)\sim D}\left[\indicator_{\abs{\phi(\vect x)}\leq\log\frac{1}{\eps}}\cdot\lambda\left((\phi(\vect x)+\vtail^{*}\cdot\vect x)y\right)\right]\pm\eps.
\]

Picking $B\geq1$ to be a positive parameter to be set later, we can
write 
\begin{multline}
\E_{(\vect x,y)\sim D}\left[\indicator_{\abs{\phi(\vect x)}\leq\log\frac{1}{\eps}}\cdot\lambda\left((\phi(\vect x)+\vtail^{*}\cdot\vect x)y\right)\right]\leq\\
\E_{(\vect x,y)\sim D}\left[\indicator_{\abs{\vtail^{*}\cdot\vect x}\leq B}\indicator_{\abs{\phi(\vect x)}\leq\log\frac{1}{\eps}}\cdot\lambda\left((\phi(\vect x)+\vtail^{*}\cdot\vect x)y\right)\right]\\
+\E_{(\vect x,y)\sim D}\left[\indicator_{\abs{\vtail^{*}\cdot\vect x}\geq B}\indicator_{\abs{\phi(\vect x)}\leq\log\frac{1}{\eps}}\cdot\lambda\left((\phi(\vect x)+\vtail^{*}\cdot\vect x)y\right)\right]\label{eq: breaking expected hinge loss into two parts}
\end{multline}
We now bound the two terms above as follows. To bound the first term,
we make the following three observations:
\begin{enumerate}
\item From the definition of $\lambda$, we see that $\lambda(z)\leq\frac{\abs z}{\eps}+1$,
and that $\lambda(z)=0$ if $z\geq\eps$.
\item By Berry-Esseen (Fact \ref{fact: Berry-Esseen for boolean}) 
and the fact that $\norm{\vtail^{*}}_{4}^2\leq\eps$
and $\norm{\vtail^{*}}_{2}=1$, for any fixed $a\in\R$ we
have 
\[
\Pr_{\vect x\sim\left\{ \pm1\right\} ^{d}}\left[\abs{\vtail^{*}\cdot\vect x+a}\leq\eps\right]\leq O(\eps)
\]
Further recalling that $\vtail^{*}$ and $\vhead$
are supported on disjoint indices in $[d]$, we see that $\vtail^{*}\cdot\vect x$
and $\phi(\vect x)$ are indipendent random variables
when $\vect x$ is drawn uniformly from $\left\{ \pm1\right\} ^{d}$.
This allows us to average over $a=\phi(\vect x)$ to
conclude that 
\[
\Pr_{\vect x\sim\left\{ \pm1\right\} ^{d}}\left[\abs{\vtail^{*}\cdot\vect x+\phi(\vect x)}\leq\eps\right]\leq O(\eps)
\]
\item We see that $(\phi(\vect x)+\vtail^{*}\cdot\vect x)y<0$
if and only if $\sign(\vhead \cdot \vect x+\vtail^{*}\cdot \vect x+ \tau)\neq y$, since we have defined $\phi(\vect x) = \vhead \cdot \vect x + \tau$.
\end{enumerate}
The three observations above, together with the observation that $\lambda(z)=0$ whenever $z>\eps$, allow us to conclude that 

\begin{align}
\E_{(\vect x,y)\sim D}&\left[\indicator_{\abs{\vtail^{*}\cdot\vect x}\leq B}\indicator_{\abs{\phi(\vect x)}\leq\log\frac{1}{\eps}}\cdot\lambda\left((\phi(\vect x)+\vtail^{*}\cdot\vect x)y\right)\right]\leq \notag\\
&\le \Pr_{\vect x\sim\left\{ \pm1\right\} ^{d}}\left[y\left(\vtail^{*}\cdot\vect x+\phi(\vect x)\right)\in[0,\eps]\right]+\left(\frac{B}{\eps}+1\right)\cdot\Pr_{(\vect x,y)\sim D}\left[(\phi(\vect x)+\vtail^{*}\cdot\vect x)y<0\right] \notag \\
&= O(\eps)+\left(\frac{B}{\eps}+1\right)\cdot\Pr_{(\vect x,y)\sim D}\left[\sign(\phi(\vect x)+\vtail^{*}\cdot \vect x)\neq y\right]\notag \\
&\le O(\eps)+\left(\frac{B}{\eps}+1\right)\opt_{\mathrm{reg}},\label{eq: bounding first term of hinge loss}
\end{align}
where the first inequality follows from the first inequality above,
the second inequality from the other two observations, and the last
equality is due to $\vtail^{*}$ being a minimizer in Equation
\ref{eq: def of opt reg tail}. 

Using again the inequality $\lambda(z)\leq\frac{\abs z}{\eps}+1$,
we bound the second term in Equation \ref{eq: breaking expected hinge loss into two parts}
as follows:
\begin{align*}
\E_{(\vect x,y)\sim D}&\left[\indicator_{\abs{\vtail^{*}\cdot\vect x}\geq B}\indicator_{\abs{\phi(\vect x)}\leq\log\frac{1}{\eps}}\cdot\lambda\left((\phi(\vect x)+\vtail^{*}\cdot\vect x)y\right)\right]\leq\\
&\le \E_{(\vect x,y)\sim D}\left[\indicator_{\abs{\vtail^{*}\cdot\vect x}\geq B}\left(\frac{1}{\eps}\left(\log\frac{1}{\eps}+\abs{\vtail^{*}\cdot\vect x}\right)+1\right)\right]\\
&= \left(\frac{1}{\eps}\log\frac{1}{\eps}+1\right)\Pr_{(\vect x,y)\sim D}\left[\abs{\vtail^{*}\cdot\vect x}\geq B\right]+\frac{1}{\eps}\E_{(\vect x,y)\sim D}\left[\indicator_{\abs{\vtail^{*}\cdot\vect x}\geq B}\abs{\vtail^{*}\cdot\vect x}\right].
\end{align*}
Combining this with the Hoeffding inequality, we see
\[
\Pr_{(\vect x,y)\sim D}\left[\abs{\vtail^{*}\cdot\vect x}\geq B\right]\leq2\exp\left(-2B^{2}\right),
\]
\begin{align*}
\E_{(\vect x,y)\sim D}\left[\indicator_{\abs{\vtail^{*}\cdot\vect x}\geq B}\abs{\vtail^{*}\cdot\vect x}\right]& \leq B\cdot\Pr\left[\abs{\vtail^{*}\cdot\vect x}\geq B\right]+\int_{z=B}^{\infty}\Pr\left[\abs{\vtail^{*}\cdot\vect x}\geq z\right]\d z \\
& \leq2B\exp\left(-2B^{2}\right)+2\int_{z=B}^{\infty}\exp\left(-2z^{2}\right)\d z \\
& \leq 2B\exp\left(-2B^{2}\right)+2\int_{z=B}^{\infty}z\exp\left(-2z^{2}\right)\d z \\
& \leq O\left(B\exp\left(-2B^{2}\right)\right).
\end{align*}
therefore
\begin{align*}
\E_{(\vect x,y)\sim D}\left[\indicator_{\abs{\vtail^{*}\cdot\vect x}\geq B}\indicator_{\abs{\phi(\vect x)}\leq\log\frac{1}{\eps}}\cdot\lambda\left((\phi(\vect x)+\vtail^{*}\cdot\vect x)y\right)\right]\leq
O\left(\frac{B}{\eps}\log\frac{1}{\eps}\exp\left(-2B^{2}\right)\right),
\end{align*}
which combined with Equations \ref{eq: breaking expected hinge loss into two parts}
and \ref{eq: bounding first term of hinge loss} yields
\[
\E_{(\vect x,y)\sim D}\left[\indicator_{\abs{\phi(\vect x)}\leq\log\frac{1}{\eps}}\cdot\lambda\left((\phi(\vect x)+\vtail^{*}\cdot\vect x)y\right)\right]\leq \left(\frac{B}{\eps}+1\right)\opt_{\mathrm{reg}}+O\left(\eps + \frac{B}{\eps}\log\left(\frac{1}{\eps}\right)e^{-2B^{2}}\right).
\]
Substituting $B=\log\frac{1}{\eps}$ yields the desired Equation
\ref{eq: bound on hingle-like loss}.
\end{proof}
Using Claim \ref{claim: bound on hinge-like loss for ground truth tail},
and the observation that $\lambda\left((\phi(\vect x)+\vtailhat\cdot\vect x)y\right)\geq\indicator_{\sign\left(\phi(\vect x)+\vtailhat\cdot\vect x\right)\neq y}$,
we see that:
\begin{align}
O\left(\frac{\opt_{\mathrm{reg}}}{\eps}\log\frac{1}{\eps}+\eps\right) &\geq \frac{1}{|S|}\sum_{(\vect x,y)\in {\Sout}}\left[\lambda\left((\phi(\vect x)+\vtail^{*}\cdot\vect x)y\right)\right] \notag\\
&\ge \frac{1}{|S|}\sum_{(\vect x,y)\in {\Sout}}\left[\lambda\left((\phi(\vect x)+\vtailhat\cdot\vect x)y\right)\right] \notag\\
&\ge \frac{1}{|S|}\sum_{(\vect x,y)\in {\Sout}}\ind_{\sign\left(\phi(\vect x)+\vtailhat\cdot\vect x\right)\neq y} \notag\\
&= \Pr_{(\vect x,y)\in S}\left[\left(\abs{\phi(\vect x)}\leq\log\frac{1}{\eps}\right)\text{\ensuremath{\land}}\left(\sign\left(\phi(\vect x)+\vtailhat\cdot\vect x\right)\neq y\right)\right]\label{eq: bound on loss inside the band}
\end{align}
Since $\norm{\vtail^{*}}_{2}$ and $\norm{\vtailhat}_{2}$
are bounded by $1$, we can conclude (using the Hoeffding inequality
for $\vect x\cdot\vtail^{*}$ and $\vtailhat\cdot\vect x$
respectively) that:
\begin{multline*}
\opt_{\mathrm{reg}}\geq\Pr_{(\vect x,y)\sim D}\left[\left(\abs{\phi(\vect x)}>\log\frac{1}{\eps}\right)\land\left(\sign(\phi(\vect x)+\vtail^{*}\cdot\vect x)\neq y\right)\right]=\\
\Pr_{(\vect x,y)\sim D}\left[\left(\abs{\phi(\vect x)}>\log\frac{1}{\eps}\right)\land\left(\sign(\phi(\vect x))\neq y\right)\right]\pm\underbrace{\Pr_{(\vect x,y)\sim D}\left[\abs{\vtail^{*}\cdot\vect x}\geq\log\frac{1}{\eps}\right]}_{\leq O(\eps)},
\end{multline*}
and by the same token

\begin{multline*}
\Pr_{(\vect x,y)\sim D}\left[\left(\abs{\phi(\vect x)}>\log\frac{1}{\eps}\right)\land\left(\sign(\phi(\vect x)+\vtailhat\cdot\vect x)\neq y\right)\right]=\\
\Pr_{(\vect x,y)\sim D}\left[\left(\abs{\phi(\vect x)}>\log\frac{1}{\eps}\right)\land\left(\sign(\phi(\vect x))\neq y\right)\right]\pm\underbrace{\Pr_{(\vect x,y)\sim D}\left[\abs{\vtailhat\cdot\vect x}\geq\log\frac{1}{\eps}\right]}_{\leq O(\eps)},
\end{multline*}
which together tell us that 
\[
\Pr_{(\vect x,y)\sim D}\left[\left(\abs{\phi(\vect x)}>\log\frac{1}{\eps}\right)\land\left(\sign(\phi(\vect x)+\vtailhat\cdot\vect x)\neq y\right)\right]\leq\opt_{\mathrm{reg}}+O(\eps).
\]
Using the above together with Chebyshev's inequality we see that if
$C$ is a sufficiently large absolute constant, then the $C\left(\frac{d}{\eps}\right)^{C}$-sized
sample set $S$ with probability at least $0.999$ satisfies
\[
\Pr_{(\vect x,y)\sim S}\left[\left(\abs{\phi(\vect x)}>\log\frac{1}{\eps}\right)\land\left(\sign(\phi(\vect x)+\vtailhat\cdot\vect x)\neq y\right)\right]\leq\opt_{\mathrm{reg}}+O(\eps),
\]
which together with \ref{eq: bound on loss inside the band} and the definition of $\phi$ \eqref{equation:definition-phi}
implies that 
\begin{equation}
\Pr_{(\vect x,y)\sim S}\left[\left(\sign(\vhead\cdot\vect x+\vtailhat\cdot\vect x + \tau)\neq y\right)\right]\leq O\left(\frac{\opt_{\mathrm{reg}}}{\eps}\log\frac{1}{\eps}+\eps\right).\label{eq: hypothesis good on sample set}
\end{equation}
Finally, a standard uniform-convergence VC-dimension argument for
halfspaces implies that if $C$ is a sufficiently large absolute
constant, then the $C\left(\frac{d}{\eps}\right)^{C}$-sized sample
set $S$ with probability at least $0.999$ satisfies the following inequality, uniformly over $\vtail\in\R^{d}$:
\[
\abs{\Pr_{(\vect x,y)\sim S}\left[\left(\sign(\vhead\cdot \vect x+\vtail\cdot\vect x + \tau)\neq y\right)\right]-\Pr_{(\vect x,y)\sim D}\left[\left(\sign(\vhead\cdot \vect x+\vtail\cdot\vect x + \tau)\neq y\right)\right]}\leq\eps,
\]
which together with Equation \ref{eq: hypothesis good on sample set}
implies Equation \ref{eq: bound with known head vars and weights}
finishing the proof of Theorem \ref{thm: with known head vars and weights}.
\end{proof}

\section{Putting everything together} \label{sec:combineall}

In this section, we complete the proof of our main result by combining all the ingredients we obtained in the previous sections. In particular, we will prove the following theorem.

\begin{thm}\label{thm: main theorem}
    There is an algorithm that, given access to a distribution $D$ over $\cube{d}\times \cube{}$ whose marginal on $\cube{d}$ is the uniform distribution, outputs, with probability at least $0.8$, a hypothesis $\widehat{h}:\cube{d} \to \cube{}$ such that
    \[
        \Pr_{(\x,y)\sim D}\Bigr[\widehat{h}(\x) \neq y\Bigr] \le O(\opt^{1/25}) + \eps\,, \text{ where }\opt = \min_{\substack{\vect v\in \R^d \\ \tau \in \R}} \Pr_{(\x,y)\sim D}\Bigr[\sign(\vect v\cdot \x+ \tau) \neq y \Bigr]\,.
    \]
    The algorithm has time and sample complexity $\poly(d, 1/\eps)$.
\end{thm}

If $\opt \le \eps^{25}/C$ for some sufficiently large constant $C$, then Algorithm \ref{algorithm:learn-halfspace}, run on input $\eps$, achieves the guarantee of \Cref{thm: main theorem}. Otherwise, we may choose the input parameter $\eps'$ to be such that $(\eps')^{25}/{(2C)}\le \opt \le (\eps')^{25}/C$ by trying all possible estimates for $\opt$ within a grid of appropriate size and choose the one that leads to the hypothesis with the minimum error.

\begin{remark}
    The probability of success can be amplified to $1-\delta$ through $O(\log 1/\delta)$ repetitions and selection of the hypothesis with the best accuracy.
\end{remark}

One key ingredient of our proof is the, by now well-known, critical index lemma which states that any halfspace over the boolean hypercube is either close to a sparse halfspace, or most of its coefficients are roughly of the same magnitude. Here, we use the following version of the critical index lemma appearing in \cite{gopalan2010fooling}.

\begin{lem}[\cite{gopalan2010fooling}]\label{lem: critical index}
    Let $\sign\left(\vect v\cdot\vect x+\tau\right)$ be a halfspace over $\left\{ \pm1\right\} ^{d}$. Define for every $j\in[d]$
\[
i_{j}=\argmax_{i\in[d]\setminus\left\{ i_{1},i_{2},\cdots i_{j-1}\right\} }\left[\abs{v_{i}}\right].
\]
Let $k$ be a positive integer. Then, one of the following two cases
must hold:
\begin{enumerate}
    \item {(Sparse).} There is a halfspace $\sign\left(\vect v_{\mathrm{sp}}\cdot\vect x
+\tau
\right)$ with $\vect v_{\mathrm{sp}}$ 
supported on at most $k$ indices $\left\{ i_{1},i_{2},\cdots i_{k}\right\} $ for which
\[
\Pr_{\vect x\sim\left\{ \pm1\right\} ^{d}}\left[\sign\left(\vect v\cdot\vect x
+\tau
\right)\neq\sign\left(\vect v_{\mathrm{sp}}\cdot\vect x
+\tau
\right)\right]\leq\frac{1}{k^{100}}.
\]
    \item {(Structured).} There exists some $\Delta\in\left\{ 0,1,\cdots,k-1\right\} $
for which $\frac{\norm{\vect u_{\mathrm{tail}}}_{4}^2}{\norm{\vect u_{\mathrm{tail}}}_{2}^2}\leq O\left(\frac{\log^2 k}{\sqrt{k}}\right)$,
where $\vect u_{\mathrm{tail}}$ denotes the restriction of $\vect v$
onto $[d]\setminus\left\{ i_{1},i_{2},\cdots i_{\Delta}\right\} $. 
\end{enumerate}
\end{lem}

We are now ready to prove \Cref{thm: main theorem}.

\begin{proof}[Proof of \Cref{thm: main theorem}]
     As we mentioned after the statement of \Cref{thm: main theorem}, without loss of generality, we may assume that $\opt \le \etatwo$ for  $\etatwo=\eps^{25}/C$.

     Recall that $D$ is a distribution over $\left\{ \pm1\right\} ^{d}\times\left\{ \pm 1\right\} $
that has a uniform $\left\{ \pm1\right\} ^{d}$-marginal and for some
halfspace $f^*(\vect x) = \sign\left(\vect v^{*}\cdot\vect x + \tau^*\right)$ satisfies
\begin{equation}
\Pr_{(\vect x,y)\sim D}\left[f^*(\vect x)\neq y\right]\leq\etatwo=\eps^{25}/C.\label{eq: u star has good accuracy}
\end{equation}

Without loss of generality $\vect v^{*}$ has integer weights whose
values are bounded by $2^{10d\log d}$ \cite{servedio2006}. The algorithm
computes 
\[
\hatI_{i}=\E_{(\vect x,y)\sim D}[y x_i]\pm\etatwo
\]
and obtains $H=\{\hati_{1},\hati_{2},\cdots\hati_{k}\}$ as 
\begin{equation}
\hati_{t}\leftarrow\argmax_{i\in[d]\setminus\left\{ \hati_{1},\hati_{2},\cdots\hati_{t-1}\right\} }\left[\abs{\hatI_{i}}\right].\label{eq: how i hat are defined}
\end{equation}

This allows us to apply \Cref{theorem:influential}
to obtain a halfspace $g^*(\vect x)=\sign(\tilde{\vect v}^{*} \cdot \vect x +\tau)$ so that the vector $\tilde{\vect v}^{*}$ satisfies 
\begin{enumerate}
\item The top $k$ weights of $\tilde{\vect v}^{*}$ (in absolute value) are located
respectively at the indices $\hati_{1},\hati_{2},\cdots\hati_{k}$.
\item All weights of $\tilde{\vect v}^{*}$are integers with weights bounded by
$2^{10d\log d}$ (as are the weights of $\vect v^{*}$).
\item It is the case that 
\[
\Pr_{\vect x\sim\left\{ \pm1\right\} ^{d}}\left[
g^*(\vect x)
\neq
f^*(\vect x)\right]\leq4 \etatwo k.
\]
\end{enumerate}
The last equation together with Equation \ref{eq: u star has good accuracy}
implies that 
\begin{equation}
\Pr_{(\vect x,y)\sim D}\left[g^*(\vect x)\neq y\right]\leq4(k+1)\etatwo\label{eq: v star has a good accuracy}
\end{equation}
Now, we can apply Lemma \ref{lem: critical index} with the halfspace defined by $\tilde{\vect v}^{*}$.
Suppose we are in case (1), then for some vector $\vect v_{\mathrm{sp}}$ supported of $\{\widehat{i}_1,\cdots, \widehat{i}_k\}$
we have
\[
\Pr_{\vect x\sim\left\{ \pm1\right\} ^{d}}\left[g^*(\vect x)\neq\sign\left(\vect v_{\mathrm{sp}}\cdot\vect x+\tau\right)\right]\leq\frac{1}{k^{100}},
\]
which together with the previous inequality implies that 
\[
\Pr_{(\vect x,y)\sim D}\left[\sign\left(\vect v_{\mathrm{sp}}\cdot\vect x+\tau\right)\neq y\right]\leq\frac{1}{k^{100}}+4(k+1)\etatwo.
\]
A union bound implies that with all examples in $S_2$
will be consistent with $\sign\left(\vect v_{\mathrm{sp}}\cdot\vect x+\tau\right)$
with probability 
\[
\left(\frac{1}{k^{100}}+4(k+1)\etatwo\right)\abs{S_2}.
\]
Substituting $k=\setk$ and $\etatwo=\eps^{25}/C$ and $S_2=\setSparseDatasetSize$,
we see that the above is at most $0.001$ for a sufficiently large
absolute constant $C$. Thus, the set $S_2$
will be linearly separable, and the algorithm will succeed in finding
a classifier $\hsparse(\vect x)=\sign(\hatvsparse\cdot x+\widehat{\tau})$ that perfectly classifies
the set $S_2$. Given the size of $S_2$,
the standard VC bound for linear classifiers implies that (for a
sufficiently large absolute constant $C$) with probability at least
0.999 we have
\[
\Pr_{(\vect x,y)\sim D}\left[\sign\left(\hatvsparse\cdot\vect x+\widehat{\tau}\right)\neq y\right]\leq\eps,
\]
concluding our consideration of the sparse case of Lemma \ref{lem: critical index}.

For the rest of the section we will consider the second case of Lemma
\ref{lem: critical index}, where $\tilde{\vect v}^* = \vhead^* + \vtail^*$ for some vector $\vtail^*$ whose projection outside $H_\Delta$ is regular. We will focus on the iteration $\Delta$
of our algorithm, fixing $\Delta$ to be the value given to us by
Lemma \ref{lem: critical index}. We tacitly will equate $\left\{ \pm1\right\} ^{d}$
with $\left\{ \pm1\right\} ^{\Delta}\times\left\{ \pm1\right\} ^{d-\Delta}$
using the bijective map
\[
\vect x\leftrightarrow\left(\x_{\mathrm{head}},\x_{\mathrm{tail}}\right),\text{ where }
\x_{\mathrm{head}}=\left[x_{\widehat{i}_{1}},x_{\widehat{i}_{2}},\cdots,x_{\widehat{i}_{\Delta}}\right]
\]
and $\x_{\mathrm{tail}}$ is the coordinates of $\vect x$ on the
rest of $[d]$ (ordered in some arbitrary but fixed way). In the notation
of our main algorithm, $\x_{\mathrm{head}} = \x_{H_\Delta}$. Moreover, we let $\vect u_{\mathrm{head}}^{*} = (\tilde{\vect v}^*)_{H_\Delta}$ and $\vect u_{\mathrm{tail}}^{*}= (\tilde{\vect v}^*)_{[d]\setminus H_\Delta}$. We have:

\begin{align}
g^*(\vect x) &=
\sign\left(\tilde{\vect v}^{*}\cdot\vect x
+\tau^*
\right)=\sign\left(\vect u_{\mathrm{head}}^{*}\cdot\x_{\mathrm{head}}+\vect u_{\mathrm{tail}}^{*}\cdot\x_{\mathrm{tail}}
+\tau^*
\right) \notag\\
&= \sign\left(\frac{\vect u_{\mathrm{head}}^{*}\cdot\x_{\mathrm{head}}}{\norm{\vect u_{\mathrm{tail}}^{*}}_{2}}+\frac{\vect u_{\mathrm{tail}}^{*}\cdot\x_{\mathrm{tail}}}{\norm{\vect u_{\mathrm{tail}}^{*}}_{2}}
+\frac{\tau^*}{\norm{\vect u_{\mathrm{tail}}^{*}}_{2}}
\right),\label{eq: breaking v star into head and tail}
\end{align}
where $\norm{\vect u_{\mathrm{tail}}^{*}}_{2}\neq0$ because otherwise
we would fall under the first case of Lemma \ref{lem: critical index}.
Since $\vect u_{\mathrm{tail}}^{*}$ has integer weights bounded by
$2^{10d\log d}$, the vector $\frac{\vect u_{\mathrm{head}}^{*}}{\norm{\vect u_{\mathrm{tail}}^{*}}_{2}}$
will likewise have weights bounded by $2^{10d\log d}$. Since the
second case of Lemma \ref{lem: critical index} holds, we see that
\begin{equation}
\frac{\norm{\vect u_{\mathrm{tail}}^{*}}_{4}^2}{\norm{\vect u_{\mathrm{tail}}^{*}}_{2}^2}\leq O\left(\frac{\log^{2}k}{\sqrt{k}}\right)\label{eq: v tail star is regular}
\end{equation}
and therefore Fact \ref{fact: Berry-Esseen for boolean} implies
that $\vect u_{\mathrm{tail}}^{*}$ is such that for all $\tau\in\R$
\begin{equation}
\Pr_{\vect x\sim\left\{ \pm1\right\} ^{d}}\left[\frac{\vect u_{\mathrm{tail}}^{*}\cdot\x_{\mathrm{tail}}}{\norm{\vect u_{\mathrm{tail}}^{*}}_{2}}\geq \tau\right]=\Pr_{z\sim\N(0,1)}\left[z\geq \tau\right]\pm O\left(\frac{\log^{2}k}{\sqrt{k}}\right),\label{eq: tail is similar to Gaussian}
\end{equation}
which implies that for any fixed value of $\x_{\mathrm{head}}$ we have:
\begin{multline}
\Pr_{\x_{\mathrm{tail}}\sim\left\{ \pm1\right\} ^{d-\Delta}}\left[\sign\left(
\frac{\vect u_{\mathrm{head}}^{*}\cdot\x_{\mathrm{head}}}{\norm{\vect u_{\mathrm{tail}}^{*}}_{2}}+\frac{\vect u_{\mathrm{tail}}^{*}\cdot\x_{\mathrm{tail}}}{\norm{\vect u_{\mathrm{tail}}^{*}}_{2}}
+
\frac{\tau^*}{\norm{\vect u_{\mathrm{tail}}^{*}}_{2}}
\right)=1\right] \\
= \Pr_{z\sim\N(0,1)}\left[\frac{\vect u_{\mathrm{head}}^{*}\cdot\x_{\mathrm{head}}}{\norm{\vect u_{\mathrm{tail}}^{*}}_{2}}+
\frac{\tau^*}{\norm{\vect u_{\mathrm{tail}}^{*}}_{2}}
+z\geq0\right]\pm O\left(\frac{\log^{2}k}{\sqrt{k}}\right)\label{eq: how the sign is distributed conditional on head}
\end{multline}

Define $\mathcal{D_{\mathrm{GLM}}}$ to be the probability distribution over $\R^{\Delta}\times \{\pm 1\}$ distributed as the tuple
of random variables $\left(\x_{\mathrm{head}},g^*(\x)
\right)$, where $\x$ is uniform on $\cube{d}$, and recall that $D_\Delta$ is the distribution of pairs of the form $(\x_{\mathrm{head}}, y)$, where $\x_{\mathrm{head}} = \x_{H_\Delta}$ and $(\x,y)\sim D$.
We will first consider what happens when we run $\findheavycoefficients(\eps_{\mathrm{hv}},\Delta,u_{\max},\mathcal{D_{\mathrm{GLM}}})$
and then use it to make conclusions about $\findheavycoefficients(\eps_{\mathrm{hv}}, \Delta, u_{\max}, D_\Delta)$. Recall that, in Algorithm \ref{algorithm:learn-halfspace}, we make the following choices for the input parameters of $\findheavycoefficients$: $\eps_{\mathrm{hv}} = C^{0.01} \log^2 (k) / \sqrt{k}$ and $u_{\max} = d2^{20d \log d}$.

The distribution $\mathcal{D_{\mathrm{GLM}}}$ is a generalized linear model such that:
\[
    \E_{(\x_{\mathrm{head}}, y) \sim \mathcal{D_{\mathrm{GLM}}}}\Bigr[y \;\Bigr|\; \x_{\mathrm{head}} \Bigr] = \sigma\left(
\frac{\vect u_{\mathrm{head}}^{*}\cdot\x_{\mathrm{head}}}{\norm{\vect u_{\mathrm{tail}}^{*}}_{2}}
+\frac{\tau^*}{\norm{\vect u_{\mathrm{tail}}^{*}}_{2}}
\right)\,,
\]
    where the activation function $\sigma(\cdot)$ is defined as follows:
\begin{align}
\sigma(\tautemp) &=\E_{x_{\text{tail}}\sim\{\pm 1\}^{d-\Delta}}\left[
\sign\left(\frac{\vect u_{\mathrm{tail}}^{*}\cdot\x_{\mathrm{tail}}}{\norm{\vect u_{\mathrm{tail}}^{*}}_{2}}+\tautemp
\right)
\right]
\notag \\
&= 2\Biggr(
\Pr_{x_{\text{tail}}\sim\{\pm 1\}^{d-\Delta}}\left[\frac{\vect u_{\mathrm{tail}}^{*}\cdot\x_{\mathrm{tail}}}{\norm{\vect u_{\mathrm{tail}}^{*}}_{2}}
+
\tautemp\geq0\right]
-1
\Biggr).\label{eq: definition of sigma for our case}
\end{align}
Recalling the definition of a regular activation function (Definition
\ref{def: regular activation}), we conclude that Equation \ref{eq: v tail star is regular}
implies that the function $\sigma$ is $O({\log^{2}(k)} / {\sqrt{k}})$-regular.
Recalling that we use a regularity parameter $\eps_{\mathrm{hv}} = C^{0.01} \log^2 (k) / \sqrt{k}$
when running $\findheavycoefficients$, we see that $\sigma$
is indeed $\eps_{\mathrm{hv}}$-regular when $C$ is a sufficiently
large absolute constant. Additionally, as we noted earlier, each coefficient
of ${\vect u_{\mathrm{head}}^{*}}/{\norm{\vect u_{\mathrm{tail}}^{*}}_{2}}$
is bounded by $2^{10d\log d}$ in absolute value. Thus, in this case
we can apply Corollary \ref{cor: the GLM component theorem nonhom}\footnote{ Without loss of generality, the bias $\tau^*$ of the ground truth half-space is bounded by $d 2^{10d\log d}$ as the inputs are from $\{\pm 1\}^d$.} to conclude
that with probability at least $0.95$ the algorithm $\findheavycoefficients(\eps_{\mathrm{hv}},\Delta,u_{\max},\mathcal{D_{\mathrm{GLM}}})$
will output a vector $\widehat{\vect u}$ and a scalar $\widehat{\tau}_\Delta$
for which 
\[
\E_{(\x_{\mathrm{head}},y)\sim\mathcal{D_{\mathrm{GLM}}}}\left[\left(
\sigma\left(
\widehat{\vect u}\cdot\x_{\mathrm{head}}+\widehat{\tau}_\Delta \right)-\sigma\left(\frac{\vect u_{\mathrm{head}}^{*}\cdot\x_{\mathrm{head}}}{\norm{\vect u_{\mathrm{tail}}^{*}}_{2}}
+\frac{\tau^*}{\norm{\vect u_{\mathrm{tail}}^{*}}_{2}}
\right)
\right)^{2}\right]\leq O\left(\eps_{\mathrm{hv}} \log(\Delta/\eps_{\mathrm{hv}})\right),
\]
which in particular implies (via Jensen's inequality)
\begin{equation}
\E_{\tilde{\x}\sim\left\{ \pm1\right\} ^{\Delta}}\left[\abs{
\sigma\left(
\widehat{\vect u}\cdot\tilde{\x}+\widehat{\tau}_\Delta \right)-\sigma\left(\frac{\vect u_{\mathrm{head}}^{*}}{\norm{\vect u_{\mathrm{tail}}^{*}}_{2}}\cdot\tilde{\x}
+\frac{\tau^*}{\norm{\vect u_{\mathrm{tail}}^{*}}_{2}}
\right)
}\right]\leq O\left(\frac{C^{0.005}\log^{1.5} (k/\eps)}{k^{1/4}}\right)\label{eq: v head is good}
\end{equation}

Furthermore, Equations \ref{eq: breaking v star into head and tail}
and \ref{eq: v star has a good accuracy} give us a bound of $4(k+1)\etatwo$
on the statistical distance between $\mathcal{D_{\mathrm{GLM}}}$ and
a sample $\left(\x_{\mathrm{head}},y\right)$ with $\left(\x_{\mathrm{head}},\x_{\mathrm{tail}},y\right)\sim D$
(in our main algorithm we call this distribution $D_{\Delta}$). Using
the bound on the sample complexity $m_{\mathrm{hv}}$ of $\findheavycoefficients$
allows us to bound the total variation distance between the output of $\findheavycoefficients(\eps_{\mathrm{hv}}, \Delta, u_{\max}, D_\Delta)$ and the output of $\findheavycoefficients(\eps_{\mathrm{hv}}, \Delta, u_{\max}, {\cal D}_{\mathrm{GLM}})$ by the following quantity:
\begin{align*}
    m_{\mathrm{hv}}\, \mathrm{dist}_{\mathrm{TV}}(D_\Delta, {\cal D}_{\mathrm{GLM}}) \le \tilde{O}(\Delta/\eps_{\mathrm{hv}}^2) (k+1) \etatwo \le \etatwo \, \tilde{O}(k^3)
\end{align*}
Thus, with probability at least $0.95-\etatwo \, \tilde{O}(k^3) \ge 0.94$, the output $(\widehat{\vect u},\widehat{\tau}_{\Delta})$
of $\findheavycoefficients$, obtained in iteration $\Delta$ of our main algorithm, satisfies
Equation \ref{eq: v head is good}.

Now, we observe that the definition of the activation function $\sigma$
(i.e. Equation \ref{eq: definition of sigma for our case}) implies that for any pair of
fixed values $a$ and $b$ in $\R$ with $a\leq b$ we have

\begin{align*}
\Pr_{\x_{\mathrm{tail}}\sim\{ \pm1\} ^{d-\Delta}} &\Bigr[\sign\Bigr(a+\frac{\vect u_{\mathrm{tail}}^{*}\cdot\x_{\mathrm{tail}}}{\norm{\vect u_{\mathrm{tail}}^{*}}_{2}}\Bigr)\neq\sign\Bigr(b+\frac{\vect u_{\mathrm{tail}}^{*}\cdot\x_{\mathrm{tail}}}{\norm{\vect u_{\mathrm{tail}}^{*}}_{2}}\Bigr)\Bigr]=
\\
&=\Bigr|{
\Pr_{\x_{\mathrm{tail}}\sim\{ \pm1\} ^{d-\Delta}}\Bigr[-b\leq\frac{\vect u_{\mathrm{tail}}^{*}\cdot\x_{\mathrm{tail}}}{\norm{\vect u_{\mathrm{tail}}^{*}}_{2}}<-a\Bigr]}\Bigr|
\\
&=\Bigr|{
\Pr_{x_{\text{tail}}\sim\{\pm 1\}^{d-\Delta}}\Bigr[\frac{\vect u_{\mathrm{tail}}^{*}\cdot\x_{\mathrm{tail}}}{\norm{\vect u_{\mathrm{tail}}^{*}}_{2}}+a\geq0\Bigr]
-
\Pr_{x_{\text{tail}}\sim\{\pm 1\}^{d-\Delta}}\Bigr[\frac{\vect u_{\mathrm{tail}}^{*}\cdot\x_{\mathrm{tail}}}{\norm{\vect u_{\mathrm{tail}}^{*}}_{2}}+b\geq0\Bigr]
}\Bigr| \\
&=\frac{1}{2}\abs{\sigma(a)-\sigma(b)}.
\end{align*}
Analogously, we also see that the same expression holds when $a\geq b$. We can
thus apply the identity above together with Equation \ref{eq: v head is good}
to obtain 
\begin{multline*}
\Pr_{\substack{
\x_{\mathrm{head}}\sim \{\pm 1\} ^{\Delta},\\ \x_{\mathrm{tail}}\sim \{\pm 1\} ^{d-\Delta}}
}\Bigr[\sign\Bigr(
\frac{\vect u_{\mathrm{head}}^{*}\cdot\x_{\mathrm{head}}}{\norm{\vect u_{\mathrm{tail}}^{*}}_{2}}+\frac{\vect u_{\mathrm{tail}}^{*}\cdot\x_{\mathrm{tail}}}{\norm{\vect u_{\mathrm{tail}}^{*}}_{2}}
+\frac{\tau^*}{\norm{\vect u_{\mathrm{tail}}^{*}}_{2}}
\Bigr)\neq
\sign\Bigr(
\widehat{\vect u}\cdot\x_{\mathrm{head}}+\frac{\vect u_{\mathrm{tail}}^{*}\cdot\x_{\mathrm{tail}}}{\norm{\vect u_{\mathrm{tail}}^{*}}_{2}}
+\widehat{\tau}_{\Delta}
\Bigr)\Bigr]=\\
\E_{\x_{\mathrm{head}}\sim\{\pm 1\} ^{\Delta}}\Bigr[\Bigr|{\sigma\Bigr(\frac{\vect u_{\mathrm{head}}^{*}\cdot\x_{\mathrm{head}}}{\norm{\vect u_{\mathrm{tail}}^{*}}_{2}}
+\frac{\tau^*}{\norm{\vect u_{\mathrm{tail}}^{*}}_{2}}
\Bigr)-\sigma\Bigr(\widehat{\vect u}\cdot\x_{\mathrm{head}}
+\widehat{\tau}_\Delta
\Bigr)}\Bigr|\Bigr]\leq O\Bigr(\frac{C^{0.005}\log k}{k^{1/4}}\Bigr),
\end{multline*}
which together with Equations \ref{eq: breaking v star into head and tail}
and \ref{eq: v star has a good accuracy} gives us 
\begin{align}
\Pr_{\vect x,y\sim D}\left[\sign\left(\widehat{\vect u}\cdot\x_{\mathrm{head}}+\frac{\vect u_{\mathrm{tail}}^{*}\cdot\x_{\mathrm{tail}}}{\norm{\vect u_{\mathrm{tail}}^{*}}_{2}}
+\widehat{\tau}_\Delta
\right)\neq y\right] &\leq O\left(\frac{C^{0.005}\log^{1.5} (k/\eps)}{k^{1/4}}\right)+(k+1)\etatwo \notag \\ 
&\leq\frac{\eps^{2}}{C^{0.05}\log^{1.5}\frac{1}{\eps}},\label{eq: v head hat is good}
\end{align}
where the last step holds for a sufficiently large value of $C$
when we substitute the values of $k,\etatwo$. 

Now, we will analyze the vector $\widehat{\vect v}_{\Delta}$ obtained as
the output of the algorithm for finding regular coefficients $\findregularcoefficients(\eps_{\mathrm{reg}},d, H_\Delta, {\hatvheavy}, \hattau_\Delta, D)$.
Recalling Theorem \ref{thm: with known head vars and weights} that
analyzes $\findregularcoefficients$, we see that it is
stated in terms of the following quantity, where we define ${\cal V}(H_{\Delta},{\eps_{\mathrm{reg}}}) = \{\vect v\in \R^{d}: \|\vect v\|_2=1, \|\vect v\|_4^2 \le \eps_{\mathrm{reg}}, \vect v_{H_\Delta} = 0\}$:
\begin{equation}
\opt_{\mathrm{reg}}:=\min_{\vect v \in {\cal V}(H_{\Delta},{\eps_{\mathrm{reg}}})} \Pr_{(\vect x,y)\sim D}\left[\sign(\hatv_{\Delta}\cdot\vect x+\vect v\cdot\x_{\mathrm{tail}}+\widehat{\tau}_\Delta)\neq y\right]\label{eq: def of opt reg tail-1}
\end{equation}
Comparing Equation \ref{eq: def of opt reg tail-1} and Equation \ref{eq: v head hat is good},
we observe that, 
taking $\vect v={\vect v_{\mathrm{tail}}^{*}}/{\norm{\vect v_{\mathrm{tail}}^{*}}_{2}}$,
using Equation \ref{eq: v tail star is regular} and substituting the value of $k$, we have
\[
\frac{\norm{\vect v_{\mathrm{tail}}^{*}}_{4}^2}{\norm{\vect v_{\mathrm{tail}}^{*}}_{2}^2}\leq O\left(\frac{\log^{2}k}{\sqrt{k}}\right)=O\left(C^{-0.05}\eps^{4}\right)\leq\eps/C^{0.01},
\]
where the last inequality holds for a sufficiently large absolute
constant $C$. Therefore, comparing Equation \ref{eq: def of opt reg tail-1}
and Equation \ref{eq: v head hat is good}, we see that $\opt_{\mathrm{reg}}\leq {\eps^{2}}/({C^{0.05}\log^{1.5}({1}/{\eps})})$.
This allows us to use Theorem \ref{thm: with known head vars and weights}
to conclude that when we invoke 
\[
\widehat{\vect v}_{\Delta}\leftarrow\findregularcoefficients\left(\eps_{\mathrm{reg}},d,H_{\Delta},\hatv_{\mathrm{hv}},\hattau_\Delta, D\right)
\]
with probability at least $0.95$ we obtain a vector $\widehat{\vect v}_\Delta$ satisfying
\begin{align*}
\Pr_{(\vect x,y)\sim D}\left[\sign(\widehat{\vect v}_{\Delta}\cdot\vect x
+\widehat{\tau}_\Delta
)\neq y\right] &\leq O\left(C^{0.01}\frac{\opt_{\mathrm{reg}}}{\eps}\log\frac{C^{0.01}}{\eps}+\frac{\eps}{C^{0.01}}\right) \\
&\le O\left(\frac{\eps}{C^{0.04}\log^{1.5}\frac{1}{\eps}}\log\frac{C^{0.01}}{\eps}+\frac{\eps}{C^{0.01}}\right)\leq\frac{\eps}{10},
\end{align*}
where the last inequality follows by taking $C$ to be a sufficiently
large absolute constant. 

Finally, we see that with overall probability of at least $0.81$,
in the iteration $\Delta$ of our main algorithm (with $\Delta$ given
by the second case of Lemma \ref{lem: critical index}) we add a hypothesis
$h_{\Delta}(\vect x)=\sign(\widehat{\vect v}_{\Delta}\cdot\vect x
+\widehat{\tau}_\Delta
)$ into the set $V$ whose prediction error
is at most $\eps/10$. The standard Hoeffding bound implies that
for a sufficiently large value of $C$, each estimate $\widehat{\text{err}}(h)$
produced by our algorithm is indeed $\eps/4$-accurate (with probability at least $0.995$). Since
there exists a hypothesis $h_{\Delta}$ in $V$ with error
at most $\eps/10$, we see that
\[
\min_{\vect h\in V}\widehat{\text{err}}(h)\leq\frac{3\eps}{4}
\]
and therefore, the minimizer $\widehat{h}$ that our algorithm
obtains satisfies
\[
\Pr_{(\vect x,y)\sim D}\left[\widehat{h}(\vect x)\neq y\right]\leq\widehat{\text{err}}(\widehat{h})+\frac{\eps}{4}\leq\eps,
\]
finishing the proof.
\end{proof}

\section{Learning with Contamination}\label{section:contamination}

In this section, we will show that our main result extends to a much more challenging noise model, called bounded contamination (or nasty noise) --- see \Cref{definition:bounded-contamination}.
While we aim to capture the bounded contamination model, it will be more convenient to phrase our proofs in terms of learning under a seemingly weaker noise model, which was recently shown to be equivalent to bounded contamination when the size of the domain is finite \cite{blanc2025adaptive}.

\begin{defn}[Learning with Oblivious Contamination]\label{definition:oblivious-contamination}
    We say that an algorithm learns a class $\mathcal{C}$ under oblivious contamination of rate $\eta\in(0,1)$ with respect to the uniform distribution over $\cube{d}$ if it satisfies the following specifications. The algorithm is given $\eps\in(0,1)$ and access to a large enough set of i.i.d. examples from some distribution $D$ over $\cube{d}\times\cube{}$ such that $\mathrm{dist}_{\mathrm{TV}}(D,D_{\mathrm{clean}}) \le \eta$ where $D_{\mathrm{clean}}$ is a distribution over samples $(\x,y)$ such that $\x$ is the uniform over $\cube{d}$ and $y = f^*(\x)$ for some unknown $f^*\in \mathcal{C}$. The output of the algorithm is some hypothesis $h:\cube{d}\times \cube{}$ such that, with probability at least $0.95$, we have:
    \[
        \Pr_{(\x, f^*(\x)) \sim D_{\mathrm{clean}}}[h(\x) \neq f^*(\x)] \le \poly(\eta) + \eps
    \]
\end{defn}

In particular, the following result from \cite{blanc2025adaptive} shows that learning under bounded contamination is essentially equivalent to learning under oblivious contamination of the same rate.

\begin{thm}[Special Case of Theorem 2 from \cite{blanc2025adaptive}]\label{theorem:adaptive-oblivious-equivalent}
    If a class $\mathcal{C}\subseteq\{\cube{d} \times \cube{}\}$ can be learned under oblivious contamination of rate $\eta$ up to error $\gamma = \gamma(\eta)$ with sample complexity $m$ and in time $T$, then it can be learned under bounded contamination of rate $\eta$ up to error $\gamma(\eta)+ \eps$ with sample complexity $\poly(m,d,1/\eps)$ in time $\poly(m,d,1/\eps) + T$.\footnote{The equivalence preserves the failure probability for any choice $\delta\in(0,1)$. This is important, since amplifying the success probability under oblivious contamination is straightforward, but the same is not true for bounded contamination, where the samples are not independent.}
\end{thm}

We will show the following theorem, which, together with Theorem \ref{theorem:adaptive-oblivious-equivalent}, yields a $\poly(d, 1/\eps)$-time algorithm in the bounded contamination case with error $\poly(\eta) + \eps$.

\begin{thm}\label{theorem:contaminated-oblivious}
    There is an algorithm that learns the class $\mathcal{C} = \{f(\x) = \sign(\vect v\cdot \x+ \tau) : \vect v\in \R^d, \tau\in \R\}$ under oblivious contamination of rate $\eta$ with respect to the uniform distribution over $\cube{d}$ up to error $O(\eta^{1/c})+ \eps$ for some sufficiently large constant $c\ge 1$. The algorithm has time and sample complexity $\poly(d, 1/\eps)$.
\end{thm}

The algorithm is essentially the same as the one used for the label noise case (Algorithm \ref{algorithm:learn-halfspace}), with the only difference being that in place of $\findregularcoefficients$, we use the algorithm with the specifications of the following theorem.

\begin{thm}
\label{thm: with known head vars and weights contaminated}
There exists an algorithm
$\findregularcontaminated(\eps,d,H,\vect v_{\mathrm{head}},\tau,D)$
that satisfies the following specifications. The algorithm takes as inputs $\eps\in(0,1)$, positive integer $d$, a
set $H\subset[d]$, a vector $\vect v_{\mathrm{head}}\in\R^{d}$ that is supported
on $H$, a real-valued $\tau$, and a sample access to a distribution $D$ supported on $\left\{ \pm1\right\} ^{d}\times\left\{ \pm1\right\} $
such that $\mathrm{dist}_{\mathrm{TV}}(D, D_{\mathrm{clean}})\leq \eta$, where for $(\vect{x}, y)\sim D_{\mathrm{clean}}$ we have (i) $\vect x$ is uniformly random on $\left\{ \pm1\right\} ^{d}$ (ii) $y=\sign(\vect v_{\mathrm{head}}\cdot\vect x+\vect v_{\mathrm{tail}}^*\cdot\vect x+\tau)$ where $
\norm{\vect v_{\mathrm{tail}}^*}_{2}=1$, for all  $i\in H$ we have $\left(\vect v_{\mathrm{tail}}^*\right)_{i}=0$ and $\norm{\vect v_{\mathrm{tail}}^*}_{4}^2\leq\eps$. 
The algorithm runs in time $\poly(d/\eps)$  and
outputs, with probability at least $0.995$, a vector
$\widehat{\vect v}\in\R^{d}$ for which 
\begin{equation}
\Pr_{(\vect x,y)\sim D}\left[\sign(\widehat{\vect v}\cdot\vect x+\tau)\neq y\right]\leq O\left(\frac{\eta}{\eps^3}+\eps\right),\label{eq: bound with known head vars and weights contaminated}
\end{equation}
\end{thm}

Compared to $\findregularcoefficients$, the algorithm $\findregularcontaminated$ performs an additional outlier removal step to ensure that the uniform distribution over the input samples is concentrated in every direction. In particular, we use the following outlier removal guarantee which originates to classical results from the literature of learning with contamination. We use a direct corollary of a version of the outlier removal guarantee appearing in the recent work of \cite{klivans2024learningac0}.

\begin{lem}[Outlier removal \cite{klivans2009learning,diakonikolas2018learning,klivans2024learningac0}]\label{lemma:outlier-removal}
    There is an algorithm that satisfies the following specifications. The algorithm is given $\delta\in(0,1)$, $r\in \mathbb{N}$, and a set $S$ of $m$ i.i.d. examples from some distribution over $\cube{d}$ whose total variation distance from the uniform is at most $\eta$. If $m \ge C {(3d)^{4r}} \log(1/\delta)$ for some sufficiently large universal constant $C\ge 1$, then the algorithm outputs a set $S' \subseteq S$ such that:
    \begin{enumerate}
        \item With probability at least $1-\delta$, we have $|S\setminus S'| \le 4\eta m$.
        \item For any $A> 0$ and any $\vect v\in \R^d$ such that $\E_{\x \sim \cube{d}}[(\vect v\cdot \x)^{2r}] \le A$, we have:
        \[
            \sum_{\x \in S'} (\vect v\cdot \x)^{2r} \le 8Am\; \text{ with probability at least }1-\delta.
        \]
    \end{enumerate}
\end{lem}

While the proof of \Cref{thm: with known head vars and weights contaminated} closely follows the proof of \Cref{thm: with known head vars and weights}, we provide its proof here for completeness.

\begin{proof}[Proof of \Cref{thm: with known head vars and weights contaminated}]
We define $\lambda(z) = \relu(1- z/\eps)$.
For some absolute constant $C$, the algorithm $\findregularcontaminated$
does the following
\begin{itemize}
\item Draw a $C\left(\frac{d}{\eps}\right)^{C}$-sized sample set $S$
from $D$.
\item Run the algorithm of \Cref{lemma:outlier-removal} on inputs $(\delta = 1/C, r = 1, S)$ and let $S'$ be its output.
\item Form a subset ${\Sout}$ of $S'$ as follows for $\phi_{\max} = C/\sqrt{\eps}$:
\[
{\Sout}\leftarrow\left\{ (\vect x,y)\in S':\ \abs{\vhead\cdot\vect x + \tau}\leq \phi_{\max}\right\} 
\]
\item Let ${\cal V} = \{\vect v\in \R^d: \|\vect v\|_2 \le 1 \text{ and } \vect v_{H} = 0\}$, where recall that $\vect v_H$ is a vector in $\R^{|H|}$ corresponding to the restriction of $\vect v$ to the coordinates in the set $H$.
\item Compute the following quantity:
\[
\vtailhat\leftarrow\argmin_{\vect v \in {\cal V}}\sum_{(\vect x,y)\in {\Sout}}\lambda\Bigr((\vhead\cdot\vect x+\vect v\cdot \vect x + \tau)\cdot y\Bigr),
\]
Note that the minimization task above can be performed in time $\poly(d/\eps)$
because the function $\lambda$ is convex, and therefore the minimization
above is a convex minimization task.
\item Output $\widehat{\vect v}:=\vtailhat+\vhead$. 
\end{itemize}

For the following, we let $\phi:\R^d \to \R$ be the function defined as follows. 
\begin{equation}
    \phi(\vect x) = \vhead\cdot \vect x + \tau\label{equation:definition-phi-contaminated}
\end{equation}
\begin{claim}
\label{claim: bound on hinge-like loss for ground truth tail contaminated}Let
$\vtail^{*}$ be a as specified in the statement of \Cref{thm: with known head vars and weights contaminated},
then with probability at least $0.999$, it is the case that 
\begin{equation}
Q:= \frac{1}{|S|}\sum_{(\vect x,y)\in {\Sout}}\left[\lambda\left((\phi(\vect x)+\vtail^{*}\cdot\vect x)y\right)\right]\leq O\left(\frac{\eta}{\eps^3}+\eps\right).\label{eq: bound on hingle-like loss contaminated}
\end{equation}
\end{claim}

\begin{proof}
We have that $\lambda(z)\le \frac{|z|}{\eps}+1$ for all $z$ and $\lambda(z) = 0,$ for all $z \ge \eps$. Moreover, $(\phi(\vect x)+\vtail^{*}\cdot\vect x)y<0$
if and only if $\sign(\vhead \cdot \vect x+\vtail^{*}\cdot \vect x+ \tau)\neq y$, since we have defined $\phi(\vect x) = \vhead \cdot \vect x + \tau$. Therefore:

\begin{align*}
    Q &\le \frac{1}{|S|} \sum_{(\x,y)\in \Sout} \biggr( 2\cdot\ind_{\phi(\x) + \vtail^*\cdot \x \in [0,\eps]} + \biggr(1+ \frac{|\phi(\x)| + |\vtail^*\cdot \x|}{\eps}  \biggr) \ind_{y\neq \sign(\phi(\x) + \vtail^*\cdot \x)}\biggr)
\end{align*}
We split the right side of the above inequality in two terms, $Q_1$ and $Q_2$. We have:
\begin{align*}
    Q_1 &:= \frac{2}{|S|} \sum_{(\x,y)\in \Sout} \ind_{\phi(\x) + \vtail^*\cdot \x \in [0,\eps]} \\
    &\le \frac{2}{|S|} \sum_{(\x,y)\in S} \ind_{\phi(\x) + \vtail^*\cdot \x \in [0,\eps]} \\
    &\le 2\Pr_{(\x,y)\sim D}\bigr[\phi(\x) + \vtail^*\cdot \x \in [0,\eps] \bigr] + O(\eps) \\
    &\le 2\Pr_{\x\sim \cube{d}}\bigr[\phi(\x) + \vtail^*\cdot \x \in [0,\eps] \bigr] + O(\eps + \eta)\\
    &\le O(\eps+\eta)\,,
\end{align*}
where the second inequality follows from a union bound, combined with the fact that $|S| \ge C(d/\eps)^C$, the third inequality follows from the total variation distance bound between $D$ and $D_{\mathrm{clean}}$, and the final inequality follows from the application of the Berry-Esseen theorem (Fact \ref{fact: Berry-Esseen for boolean}). Specifically, since $\norm{\vtail^{*}}_{4}^2\leq\eps$
and $\norm{\vtail^{*}}_{2}=1$, for any fixed $a\in\R$ we
have 
\[
\Pr_{\vect x\sim\left\{ \pm1\right\} ^{d}}\left[\abs{\vtail^{*}\cdot\vect x+a}\leq\eps\right]\leq O(\eps)
\]
Further recalling that $\vtail^{*}$ and $\vhead$
are supported on disjoint indices in $[d]$, we see that $\vtail^{*}\cdot\vect x$
and $\phi(\vect x)$ are indipendent random variables
when $\vect x$ is drawn uniformly from $\left\{ \pm1\right\} ^{d}$.
This allows us to average over $a=\phi(\vect x)$ to
conclude that 
\[
\Pr_{\vect x\sim\left\{ \pm1\right\} ^{d}}\left[\abs{\vtail^{*}\cdot\vect x+\phi(\vect x)}\leq\eps\right]\leq O(\eps)
\]

For the second term $Q_2$ of $Q$ we have the following:

\begin{align*}
    Q_2 &:= \frac{1}{|S|} \sum_{(\x,y)\in \Sout} \biggr(1+ \frac{|\phi(\x)| + |\vtail^*\cdot \x|}{\eps}  \biggr) \ind_{y\neq \sign(\phi(\x) + \vtail^*\cdot \x)} \\
    &\le {(1+\phi_{\max}/\eps)} \Pr_{(\x,y)\sim S}[y\neq \sign(\phi(\x) + \vtail^*\cdot \x)] + \frac{1}{|S|} \sum_{(\x,y)\in \Sout} \biggr(\frac{|\vtail^*\cdot \x|}{\eps}  \biggr) \ind_{y\neq \sign(\phi(\x) + \vtail^*\cdot \x)}\,,
\end{align*}
where the inequality follows from the definition of $\Sout$, where we have removed the points such that $|\phi(\x)| > \phi_{\max}$. Let $B = C/\eps^2$. We split the summation of the second term of the above expression in two parts regarding the value of $|\vtail^*\cdot \x|$: let $S_1$ contain the points such that $|\vtail^*\cdot \x|\le B$ and $S_2 = \Sout \setminus S_1$.
\begin{align*}
    Q_2 &\le \Bigr(1+\frac{\phi_{\max}+B}{\eps}\Bigr) \Pr_{(\x,y)\sim S}\Bigr[y\neq \sign(\phi(\x) + \vtail^*\cdot \x)\Bigr] + \frac{1}{|S|} \sum_{(\x,y)\in \Sout} \biggr(\frac{|\vtail^*\cdot \x|}{\eps}  \biggr) \ind_{(\x,y)\in S_2}
\end{align*}
For the second term of the above expression, we apply the Cauchy-Schwarz inequality to obtain:
\begin{align*}
    \frac{1}{|S|} \sum_{(\x,y)\in \Sout} \biggr(\frac{|\vtail^*\cdot \x|}{\eps}  \biggr) \ind_{(\x,y)\in S_2} &\le \frac{1}{\eps |S|} \sqrt{\sum_{(\x,y)\in \Sout} (\vtail^*\cdot \x)^2 \cdot \sum_{(\x,y)\in \Sout} \ind_{|\vtail\cdot \x|> B}} \\
    &\le \frac{1}{\eps} \sqrt{\frac{8}{|S|} \sum_{(\x,y)\in \Sout} \ind_{|\vtail\cdot \x|> B} } \\
    &\le \frac{1}{\eps } \sqrt{\frac{8}{|S| B^2} \sum_{(\x,y)\sim \Sout} (\vtail^*\cdot \x)^2} \\
    &\le \frac{8}{\eps B} \le O(\eps)\,,
\end{align*}
where the first inequality follows from Cauchy-Schwarz, the second follows from the guarantee of the outlier removal procedure (see \Cref{lemma:outlier-removal}) and the fact that $\Sout\subseteq S'$, the third inequality follows from Markov's inequality, and the fourth inequality follows again by the guarantee of the outlier removal procedure.

For the first term of the upper bound for $Q_2$, we use Hoeffding's inequality, as well as the bound on the total variation distance between $D$ and $D_{\mathrm{clean}}$ to overall obtain:
\[
    Q_2 \le \Bigr(1 + \frac{\phi_{\max} + B}{\eps} \Bigr)\cdot \eta + O(\eps) \le O(\eta / \eps^3) + O(\eps)\,,
\]
as desired.
\end{proof}
Using the conclusion of Claim \ref{claim: bound on hinge-like loss for ground truth tail contaminated}, and the fact that $\lambda\left((\phi(\vect x)+\vtailhat\cdot\vect x)y\right)\geq\indicator_{\sign\left(\phi(\vect x)+\vtailhat\cdot\vect x\right)\neq y}$,
we see that
\begin{align}
O\left(\frac{\eta}{\eps^3}+\eps\right) & \geq\frac{1}{|S|}\sum_{(\vect x,y)\in {\Sout}}\left[\lambda\left((\phi(\vect x)+\vtail^{*}\cdot\vect x)y\right)\right] \geq\frac{1}{|S|}\sum_{(\vect x,y)\in {\Sout}}\left[\lambda\left((\phi(\vect x)+\vtailhat\cdot\vect x)y\right)\right] \notag \\
&\geq \frac{1}{|S|}\sum_{(\vect x,y)\in {\Sout}}\ind_{\sign\left(\phi(\vect x)+\vtailhat\cdot\vect x\right)\neq y} \ge \frac{1}{|S|}\sum_{(\vect x,y)\in {S'}}\ind_{\sign\left(\phi(\vect x)+\vtailhat\cdot\vect x\right)\neq y} \ind_{|\phi(\x)| \le \phi_{\max}} \notag\\
&\ge \frac{1}{|S|}\sum_{(\vect x,y)\in {S}}\ind_{\sign\left(\phi(\vect x)+\vtailhat\cdot\vect x\right)\neq y} \ind_{|\phi(\x)| \le \phi_{\max}} - 4\eta \notag\\
&= \Pr_{(\vect x,y)\in S}\Bigr[\left(\abs{\phi(\vect x)}\leq \phi_{\max}\right)\wedge\left(\sign\left(\phi(\vect x)+\vtailhat\cdot\vect x\right)\neq y\right)\Bigr] - 4\eta\,,\label{eq: bound on loss inside the band contaminated}
\end{align}
where the last inequality follows from the guarantee of the outlier removal procedure that $|S\setminus S'| \le 4\eta |S|$.

Since $\norm{\vtail^{*}}_{2}$ and $\norm{\vtailhat}_{2}$
are bounded by $1$, we conclude that:
\begin{multline*}
\eta\geq\Pr_{(\vect x,y)\sim D}\left[\left(\abs{\phi(\vect x)}>{\phi_{\max}  }\right)\land\left(\sign(\phi(\vect x)+\vtail^{*}\cdot\vect x)\neq y\right)\right]=\\
\Pr_{(\vect x,y)\sim D}\left[\left(\abs{\phi(\vect x)}>{\phi_{\max}  }\right)\land\left(\sign(\phi(\vect x))\neq y\right)\right]\pm\underbrace{\Pr_{(\vect x,y)\sim D}\left[\abs{\vtail^{*}\cdot\vect x}\geq{\phi_{\max}  }\right]}_{\leq \eta + O(\eps)},
\end{multline*}
and similarly
\begin{multline*}
\Pr_{(\vect x,y)\sim D}\left[\left(\abs{\phi(\vect x)}>{\phi_{\max}  }\right)\land\left(\sign(\phi(\vect x)+\vtailhat\cdot\vect x)\neq y\right)\right]=\\
\Pr_{(\vect x,y)\sim D}\left[\left(\abs{\phi(\vect x)}>{\phi_{\max}  }\right)\land\left(\sign(\phi(\vect x))\neq y\right)\right]\pm\underbrace{\Pr_{(\vect x,y)\sim D}\left[\abs{\vtailhat\cdot\vect x}\geq{\phi_{\max}  }\right]}_{\leq \eta + O(\eps)},
\end{multline*}
which together tell us that 
\[
\Pr_{(\vect x,y)\sim D}\left[\left(\abs{\phi(\vect x)}>{\phi_{\max}  }\right)\land\left(\sign(\phi(\vect x)+\vtailhat\cdot\vect x)\neq y\right)\right]\leq O(\eps + \eta).
\]
Using the above together with Chebyshev's inequality we see that if
$C$ is a sufficiently large absolute constant, then the $C\left(\frac{d}{\eps}\right)^{C}$-sized
sample set $S$ with probability at least $0.999$ satisfies
\[
\Pr_{(\vect x,y)\sim S}\left[\left(\abs{\phi(\vect x)}>{\phi_{\max}  }\right)\land\left(\sign(\phi(\vect x)+\vtailhat\cdot\vect x)\neq y\right)\right]\leq O(\eps + \eta),
\]
which together with \ref{eq: bound on loss inside the band contaminated} and the definition of $\phi$ \eqref{equation:definition-phi-contaminated}
implies that 
\begin{equation}
\Pr_{(\vect x,y)\sim S}\left[\left(\sign(\vhead\cdot\vect x+\vtailhat\cdot\vect x + \tau)\neq y\right)\right]\leq O\left(\frac{\eta}{\eps^3}{ }+\eps\right).\label{eq: hypothesis good on sample set contaminated}
\end{equation}
Finally, a standard uniform-convergence VC-dimension argument for
halfspaces tells us that if $C$ is a sufficiently large absolute
constant, then the $C\left(\frac{d}{\eps}\right)^{C}$-sized sample
set $S$ with probability at least $0.999$ satisfies the uniform convergence property for the class of halfspaces,
which together with Equation \ref{eq: hypothesis good on sample set contaminated}
implies Equation \ref{eq: bound with known head vars and weights contaminated}
finishing the proof of Theorem \ref{thm: with known head vars and weights contaminated}.
\end{proof}

Another important ingredient for the proof of \Cref{theorem:contaminated-oblivious} is the following analogue of \Cref{theorem:influential-main}, which ensures that we may consider the positions of the $k$-largest magnitude coefficients to be known, by only incurring an $O(k)$ error amplification, even in the presence of contamination.

\begin{thm}\label{theorem:influential-contaminated}
    There is an algorithm that, given sample access to a distribution over pairs $(\x,y)\sim \cube{d}\times \cube{}$ that has total variation distance at most $\eta$ from another distribution over $(\x,y)$ such that $\x$ is uniform over $\cube{d}$ and $y = \sign(\vect v^*\cdot \x + \tau^*)$ for some unknown $\vect v^*, \tau^*$, outputs, with probability at least $0.995$, a set of $k$ indices $\{j_1,\dots, j_k\}$ such that:
    \[
        \Pr_{\x\sim \cube{d}}\Bigr[ \sign(\vect v^*\cdot \x + \tau^*) \neq \sign(\tilde{\vect v}^*\cdot \x + \tau^*) \Bigr] \le 4\eta k\,,
    \]
    where $\tilde{\vect v}^*$ is such that the $k$ largest (in absolute value) coefficients correspond to the indices $j_1,\dots,j_k$. The algorithm has time and sample complexity $\poly(d,1/\eta)$.
\end{thm}

\begin{proof}
    The algorithm draws a set $S$ of sufficiently many samples from the input distribution and computes the quantities $\widehat{I}_i = \E_{(\x,y)\sim S}[y x_i]$ and let $j_1,j_2,\dots,j_k$ be the $k$ indices $i$ corresponding to the $k$ largest values of $|\widehat{I}_i|$. Since $|y x_i|= 1$ and the total variation distance between the input distribution and the clean distribution is $\eta$, we have that the choice $(\widehat{i}_1, \dots, \widehat{i}_k) = (j_1,\dots,j_k)$ satisfies, with high probability, the premise of \Cref{lem: almost-maximal influences are ok} for the choice $\etatwo = 4\eta$, which concludes the proof.
\end{proof}

We are now ready to prove \Cref{theorem:contaminated-oblivious}.

\begin{proof}[Proof of \Cref{theorem:contaminated-oblivious}]
        As in the proof of \Cref{thm: main theorem}, we may, without loss of generality assume that $\eta \le \eps^{49}/C$ for some sufficiently large constant $C\ge 1$.
        The algorithm that achieves the guarantee of \Cref{theorem:contaminated-oblivious} is very similar to Algorithm \ref{algorithm:learn-halfspace}, with the main differences being that (1) the choice of the parameter $k$ for the critical index lemma is $\tilde{O}(1/\eps^{c_1})$ for $c_1 = 49/3$, and (2) in order to learn the regular coefficients (part II of the structured case), the algorithm uses $\findregularcontaminated$ in place of $\findregularcoefficients$. The proof of correctness closely follows the proof of \Cref{thm: main theorem}. Since, in this case, the noise is not only on the labels, but also on the features, we use \Cref{theorem:influential-contaminated} to find the influential coordinates. For the sparse case, as well as learning the heavy coefficients in the structured case, recall that we used appropriate total variation distance arguments to show that the sample complexity is small enough with respect to the noise rate, so that we can safely assume that the noise is negligible, and the performance of the algorithms will be, with high probability, identical to the case where the labels are realized by the corresponding sparse or structured halfspace. In fact, these total variation distance arguments are even stronger and they hold identically in the setting of learning with contamination. Finally, to learn the regular coefficients, we use \Cref{thm: with known head vars and weights contaminated}, which tolerates contamination but incurs a worse error dependence compared to its analogue for the label noise case. In particular, we overall achieve an error guarantee of $O(\eta^{1/49})+ \eps$. 
\end{proof}

\begin{remark}
    In order to obtain an error guarantee that nearly matches the guarantee of $O(\opt^{1/25})+\eps$ that we achieved in the label noise case, we may apply \Cref{lemma:outlier-removal} for some large enough constant $r$ and obtain a tighter version of \Cref{thm: with known head vars and weights contaminated}. However, this would also lead to a worse runtime of $O(d)^{4r}$. Here, we chose $r = 1$ for simplicity of presentation.
\end{remark}

\newpage
\bibliographystyle{alpha}
\bibliography{refs}

\appendix
\input{appendix}

\end{document}

%% file: intro.tex
\section{Introduction}

Distribution-specific supervised learning is a central and long-standing problem in computational learning theory. In the case of \emph{continuous} high-dimensional distributions, a wide range of polynomial-time learning algorithms are known, including~\cite{vempala2010random, awasthi2017power, daniely2015ptas, DiakonikolasKS18a, DiakonikolasKKT21, DiakonikolasZ24}. For example, the celebrated work of~\cite{awasthi2017power} gives a polynomial-time algorithm for agnostically learning linear classifiers under log-concave distributions. These results highlight a mature algorithmic understanding of learning in well-studied continuous settings.

For learning under \emph{discrete} distributions, however, our algorithmic understanding is much less complete. None of the aforementioned results extend to the most basic and well-studied discrete high-dimensional distribution: the uniform distribution over the hypercube $\{\pm 1\}^d$.  In fact, with very few exceptions (such as monotone $O(\log n)$ juntas), the only known function classes that admit polynomial-time learnability with respect to discrete distributions seem to be function classes that can be learned with respect to {\em all} distributions. 

The main reason for this striking gap between the continuous and discrete settings is that the existing algorithms rely crucially on \emph{geometric} tools of high-dimensional distributions---such as the \emph{anti-concentration property} and the \emph{localization technique}~\cite{awasthi2017power}---which fail for high-dimensional discrete distributions. These geometric tools underpin nearly all known polynomial-time learning algorithms for continuous distributions, but they do not hold, even approximately, in discrete domains such as the hypercube.

This motivates the following fundamental question:

\begin{center}
\fbox{
    \parbox{0.7\linewidth}{
        \centering
        \textbf{Is supervised learning under discrete distributions inherently harder than learning under continuous distributions?}
    }
}
\end{center}
In this work, we take a major step toward answering this question by giving the \emph{first polynomial-time algorithm} for arguably the most fundamental problem in supervised learning under discrete high-dimensional distributions --- the problem of agnostically learning halfspaces under the uniform distribution on the Boolean hypercube. Halfspaces constitute one of the most extensively studied concept classes in learning theory, while the uniform distribution on the Boolean cube is among the most canonical and well-analyzed distributions in this context.

Specifically, our algorithm achieves classification error
$
    \opt^{O(1)} + \epsilon,
$
where $\opt$ is the optimal halfspace error. All previous polynomial-time algorithms for related problems in continuous settings relied heavily on anti-concentration properties of the underlying distribution. Our result bypasses this limitation entirely: the uniform distribution on $\{\pm 1\}^d$ is not anti-concentrated, yet our algorithm still achieves a fully polynomial-time guarantee.

The fastest previously known algorithm for this problem was given by~\cite{o2008chow}, which runs in time $\poly(d) \cdot 2^{\poly(1/\epsilon)}$. Our algorithm has an exponentially better runtime while maintaining the  $\opt^{O(1)} + \epsilon$ 
accuracy guarantee. (See Table~\ref{table:hc} for a detailed comparison.)

Attempts to improve on~\cite{o2008chow} have led to significant follow-up research. In particular,~\cite{de2014nearly} developed a faster (though not polynomial-time) algorithm based on \emph{Chow parameter matching}. Their algorithm achieves quasi-polynomial accuracy blow-up $2^{-\Omega(\sqrt[3]{\log(1/\opt)})} + \epsilon$ and runs in quasi-polynomial time $\poly(d)\cdot 2^{O(\log^3(1/\epsilon))}$. Compared to this result, our algorithm achieves a quasi-polynomial improvement in both runtime and accuracy. Moreover,~\cite{de2014nearly} prove that even with improved Chow-reconstruction techniques, an $\opt^{O(1)} + \epsilon$ error bound \emph{cannot} be achieved by any approach based solely on Chow parameter matching. Thus, our results fundamentally go beyond the limitations of prior methods.

To overcome these obstacles, we develop a new technique for learning under high-dimensional discrete distributions. At a high level, our approach combines ideas from \emph{Generalized Linear Models (GLMs)} with the analysis of \emph{high-influence variables}.  As a byproduct of our analysis, we also obtain the first polynomial-time algorithm for the fundamental task of distribution-free learning of \emph{sigmoidal Generalized Linear Models}  whose runtime scales \emph{polylogarithmically} with the Lipschitz constant of the activation function. This is an exponential improvement over the runtime of previous algorithms, all of which scaled \emph{linearly} with the Lipschitz constant.

Finally, our algorithm extends naturally to the more challenging setting of learning under contamination (a.k.a.\ \emph{nasty noise}). When an $\eta$-fraction of the samples are arbitrarily corrupted, our algorithm outputs a halfspace with classification error
$
    \eta^{O(1)} + \epsilon.
$

\begin{table*}[h]\begin{center}
\begin{tabular}{c c c} 
 \toprule
 \textbf{Error Guarantee} & \textbf{Runtime} & \textbf{Reference} \\ \midrule
$\opt + \epsilon$ & $d^{O(1/\epsilon^2)}$ & \cite{kalai2008agnostically} \\ \midrule
$\opt^{\Omega(1)} + \epsilon$ & $\poly(d) \cdot 2^{\poly(1/\epsilon)}$ & \cite{o2008chow} \\ \midrule
$2^{-\Omega(\sqrt[3]{\log(1/\opt)})}+ \epsilon$ & $\poly(d)\cdot 2^{O(\log^3(1/\epsilon))}$ & \cite{de2014nearly}
 \\ \midrule
$\opt^{\Omega(1)} + \epsilon$ & $\poly(d, 1/\epsilon)$ & {\color{utburnt} \textbf{This work}}
 \\ \bottomrule
\end{tabular}
\end{center}
\caption{Upper bounds for agnostic learning of halfspaces over the uniform distribution on $\{\pm 1\}^{d}$.}
\label{table:hc}
\end{table*}

\subsection*{Related Work}

 \paragraph{Polynomial-Time Robust Learning Algorithms under Continuous Distributions.} 
There has been a large body of work on $\poly(d/\epsilon)$-time algorithms for robust learning over continuous high-dimensional distributions, such as the Gaussian distribution and log-concave distributions.

The work of \cite{vempala2010random} gives polynomial-time algorithms for (non-robust) learning intersections of halfspaces with respect to the Gaussian distribution. The work of \cite{awasthi2017power} gives a polynomial-time algorithm for learning halfspaces under log-concave distributions. The work of \cite{diakonikolas2020non} extends the result of \cite{awasthi2017power} to a broader family of high-dimensional continuous distributions satisfying certain well-behavedness assumptions, such as concentration, anti-concentration, and so-called anti-anti-concentration.

The work of \cite{daniely2015ptas} gives a polynomial-time approximation scheme (PTAS) for learning halfspaces under the Gaussian distribution, achieving error $(1+\mu)\opt+\epsilon$ for any small constant $\mu$. The work of \cite{DiakonikolasKKT21} further improves on \cite{daniely2015ptas} by giving a proper PTAS under the Gaussian distribution that outputs a halfspace as the learned hypothesis.

The works of \cite{klivans2009learning, awasthi2017power, DiakonikolasKS18a} give polynomial-time learning algorithms for learning halfspaces over log-concave distributions under contamination (also known as nasty noise). The work of \cite{DiakonikolasKS18a} also applies to learning constant-degree polynomial threshold functions (PTFs) under log-concave distributions with contamination. The more recent work of \cite{diakonikolas2024super} gives a polynomial-time algorithm for the same task with improved robustness under the Gaussian distribution.

There has also been a long line of work on polynomial-time algorithms for learning halfspaces in the Massart noise model under log-concave distributions \cite{awasthi2015efficient, awasthi2016learning, yan2017revisiting, zhang2017hitting, mangoubi2019nonconvex, diakonikolas2020learning}. The algorithm of \cite{diakonikolas2020learning} runs in time polynomial in $d/\epsilon$ and achieves an error of $\opt+\epsilon$.

None of the polynomial-time algorithms mentioned above extend to discrete high-dimensional distributions, such as the uniform distribution over the Boolean cube. This demonstrates a large gap in our understanding between robust PAC learning under continuous and discrete high-dimensional distributions.

\paragraph{Robust Learning of Halfspaces over the Hypercube.}
The algorithm of \cite{kalai2008agnostically,diakonikolas2010bounded} runs in time $d^{\Tilde{O}(1/\epsilon^2)}$ and achieves agnostic learning of halfspaces over the Boolean cube up to error $\opt+\epsilon$. There is a strong evidence that this run-time is optimal: as noted in \cite{kalai2008agnostically} even an improvement in the run-time of $d^{\Tilde{O}(1/\epsilon^{(2-\beta)})}$ for positive constant $\beta$ will yield faster algorithms for the parity with noise problem. See also \cite{dachman2014approximate} for a statistical query lower bound.

As mentioned earlier, the work of
\cite{o2008chow} gives an algorithm that runs in time $\poly(d) \cdot 2^{\poly(1/\epsilon)}$
and achieves an error of $\opt^{O(1)}+\epsilon$ in the agnostic learning setting. 

The subsequent work of \cite{de2014nearly} introduced a faster --- though still not polynomial-time --- algorithm based on \emph{Chow parameter matching}. This algorithm achieves a quasi-polynomial accuracy guarantee of $2^{-\Omega(\sqrt[3]{\log(1/\mathrm{opt})})} + \epsilon$ and operates in quasi-polynomial time $\mathrm{poly}(d) \cdot 2^{O(\log^3(1/\epsilon))}$. 

The limitations of the Chow parameter matching technique were also explored in \cite{de2014nearly}, and it was shown that no method based solely on Chow parameter matching can achieve an error of $\poly(\opt)+\epsilon$. Even more, it was shown that Chow parameter matching cannot yield an accuracy better than $\opt^{\omega(1/(\log\log \opt))}$. Our techniques break through this barrier.

The approach of \cite{diakonikolas2019degree} can be used to extend the result of \cite{de2014nearly} into the setting of learning with contaminated data (nasty noise).

\paragraph{Learning Generalized Linear Models.} Generalized Linear Models have a long history in the literature of machine learning \cite{rosenblatt1958perceptron,nelder1972generalized,dobson2018introduction}. They correspond to the scenario where the input samples $(\x,y)\in {\cal X}\times \cube{}$ are generated by some distribution such that $\E[y|\x] = \sigma(\vect w^*\cdot \x)$ for some known, monotone activation function $\sigma:\R\to [-1,1]$ and some unknown vector $\vect w^*$. Learning halfspaces is a special case of GLM learning with the sign activation, which is, however, discontinuous. Several works study efficient learnability under continuous GLMs \cite{kalai2009isotron,kakade2011efficient,klivans2017learning,chen2020classification}. Many of these and other related results, including ours, make use of the notion of the matching loss, initially proposed by \cite{auer1995exponentially}. Most related to our work, the classical result of \cite{kakade2011efficient} showed that GLMtron --- an algorithmic variant of the perceptron algorithm --- learns any GLM with known, Lipschitz activation with runtime that scales polynomially with $\|\vect w^*\|_2$.
Our results exponentially improve the run-time dependence on $\|\vect w^*\|_2$ compared to the GLMtron \cite{kakade2011efficient}, by establishing, for the first time, that a poly-logarithmic dependence on $\|\vect w^*\|_2$ is achievable for the sigmoidal activation.

A wide range of related results on learning GLMs with continuous activations have been provided in the literature. Several works consider the case where the labels are real-valued and noisy, and the goal of the learner is to output a hypothesis whose error competes with the optimum error achieved by a function of the form $\sigma(\vect w\cdot \x)$ (often referred to as \emph{robustly learning a single neuron}). Such results fall into two main categories, depending on whether the activation $\sigma$ is known \cite{ShalevShwartz2010LearningKH,diakonikolas2020approximation,diakonikolas2022learning,awasthi2023agnostic,wang2023robustly,wang2024sample,Guo2024AgnosticLO} or unknown \cite{pmlr-v75-dudeja18a,gollakota2023agnostically,pmlr-v235-zarifis24a,pmlr-v291-rohatgi25a,Wang2025RobustlyLM}. These results typically make strong assumptions on the marginal distribution on $\x$. There is strong evidence that strong assumptions are necessary when the labels are noisy \cite{diakonikolas22hardness}. In fact, most of the aforementioned results assume that $\x$ follows some continuous distribution, such as the Gaussian. It remains an open question whether our results on robust halfspace learning over the hypercube --- which corresponds to agnostic GLM learning with respect to the sign activation --- can be extended to other activations as well.

%% file: appendix.tex
\section{Auxiliary Tools}

\begin{lem}[Berry--Esseen]\label{lemma:berry-esseen}
    Let $X_1, \dots, X_d$ be independent random bariables such that $\E[X_i] = 0$, $\sum_{i\in[d]}X_i^2 \le \sigma^2$, and $\sum_{i\in[d]} |X_i|^3 \le \beta$. Let $S = \frac{1}{\sigma}\sum_{i\in[d]}X_i$ and denote with $D_S$ the distribution of $S$. Then we have the following:
    \[
        \Bigr|\Pr_{S\sim D_S}[S\ge \tau] - \Pr_{z\in \mathcal{N}(0,1)}[z \ge \tau] \Bigr| \le \frac{\beta}{\sigma^3},\, \text{ for all }\tau\in\R
    \]
\end{lem}

\begin{corollary}\label{fact: Berry-Esseen for boolean}
    Let $\vect v\in \R^d$ such that $\|\vect v\|_4^2 \le \gamma \|\vect v\|_2^2$. Then, for every $\tau\in\R$ we have:
    \[
        \Bigr| \Pr_{\vect x \sim \{\pm 1\}^d}\Bigr[\frac{\vect x \cdot \vect v}{\|\vect v\|_2} \ge \tau\Bigr] - \Pr_{z\sim\N(0,1)}\left[z\geq \tau\right] \Bigr| \le {\gamma}
    \]
\end{corollary}

\begin{proof}[Proof of Corollary \ref{fact: Berry-Esseen for boolean}]
    Let $X_i = v_i x_i$. We have that $\sum_i X_i^2 = \|\vect v\|_2^2 =: \sigma^2$. Moreover, by applying the Cauchy-Schwarz inequality, we obtain the following:
    \begin{align*}
        \sum_{i} |X_i|^3 \le \sqrt{\sum_{i}X_i^2} \cdot \sqrt{ \sum_{i} X_i^4} = \|\vect v\|_2 \cdot \|\vect v\|_4^2 \le {\gamma} \|\vect v\|_2^3 = {\gamma} \sigma^3 =: \beta
    \end{align*}
    We may now apply Lemma \ref{lemma:berry-esseen} to obtain the desired result for $S = \frac{\vect x\cdot \vect v}{\|\vect v\|_2}$.
\end{proof}

\section{Theorems on GLM learning}
This section contains all the theorems and proofs relevant to GLM learning that were omitted in \Cref{section:heavy}.
\subsection{Realizable GLM with truncated activation}
We show a weaker form of \Cref{thm: misspecified truncated GLM} which has no misspecification. This will be a helpful warmup to read before \Cref{thm: misspecified truncated GLM} and contains many of the ideas that lead to the improved dependence on $w_{\text{max}}$. Another feature is that this weaker theorem is sufficient for proving our result for learning sigmoidal GLMs (\Cref{theorem:sigmoid-main}). We are now ready to state the theorem on learning realizable truncated GLMs.
\begin{thm}
\label{thm:realize_truncated_glm}
Let $\zalphazero\in\R_{\geq1}$ and suppose $\sigma$ is a an activation
function $\R\rightarrow[-1,1]$ that satisfies the following three
conditions:
\begin{itemize}
\item $\sigma$ is monotone non-decreasing.
\item $\sigma$ is 1-Lipschitz.
\item $\sigma(z)=\sign(z)$ whenever $\abs z\geq\zalphazero$.
\end{itemize}
Then, there exists an algorithm $\AGLM^{\sigma}\left(\eps,w_{\max},\zalphazero,D\right)$ 
that
\begin{itemize}
\item Draws $O\left(\frac{\Delta\zalphazero^{2}}{\eps^{2}}\ln\frac{\Delta\zalphazero}{\eps}\right)$
samples $\left\{ x_{i}\right\} $ from a distribution $D$ over $\left\{ \vect x\in\R^{\Delta}:\norm{\vect x}_2\leq\sqrt{\Delta}\right\} \times\left\{ \pm1\right\} $
and outputs a vector $\vect v\in\R^{\Delta}$.
\item Runs in time $\poly\left(\frac{\zalphazero\Delta}{\eps}\log w_{\max}\right)$.
\item When $D$ is a GLM with an arbitrary $\R^{\Delta}$-marginal which is supported
on $\{ \vect x\in\R^{\Delta}:\norm{\vect x}_2\leq\sqrt{\Delta}\} $,
and an activation $\sigma(\vect{v^{*}}\cdot\vect x),$ where $\norm{\vect v^{*}}_2\leq w_{\max}$,
then with probability at least $0.99$ the algorithm $\AGLM^{\sigma}\left(\eps,w_{\max},\zalphazero,\dpairs\right)$
will output a vector $\vect v_{0}$ satisfying 
\begin{equation}
\E_{\vect x\sim D_{X}}\left[\left(\sigma(\vect x\cdot\vect{v^{*}})-\sigma(\vect x\cdot\vect v_{0})\right)^{2}\right]\leq\eps\label{eq: goal of thm for realizable truncated sigma}
\end{equation}
where $D_{X}$ is the $\R^{\Delta}$-marginal of $D$. (Note that no assumption
is made on $D_{X}$.)
\end{itemize}
\end{thm}
\begin{proof}
We will follow closely the Varun Kanade's notes \cite{kanade2018note}. Define the surrogate loss $\ell^{\sigma}$ as
follows:

\[
\ell^{\sigma}(\vect v;\vect x,y)=\int_{0}^{\vect v\cdot\vect x}(\sigma(z)-y)\d z.
\]
The surrogate loss is convex because $\sigma$ is monotone. The algorithm
$\AGLM^{\sigma}\left(\eps,w_{\max},\dpairs\right)$
does the following:
\begin{itemize}
\item Draw a collection $S$ of $C\frac{\Delta\zalphazero^{2}}{\eps^{2}}\ln\frac{\Delta\zalphazero}{\eps}+C$
samples from $D$, where $C$ is a sufficiently large absolute constant.
\item Using the Ellipsoid algorithm, find
\[
\vect v_{0}\leftarrow\argmin_{\vect v:\ \norm{\vect v}\leq w_{\max}}\E_{(\vect x,y)\in S}\left[\ell^{\sigma}(\vect v;\vect x,y)\right]
\]
\item Output $\vect v_{0}$.
\end{itemize}
The run-time is dominated by the second step, which runs in time $\poly\left(\frac{\zalphazero\Delta}{\eps}\log w_{\max}\right)$.

To analyze the algorithm, we define the following auxilliary functions:
\[
\text{trunc}(z):=\min(\zalphazero,\max(-\zalphazero,z))
\]
\[
\ell_{\text{trunc}}^{\sigma}(\vect v;\vect x,y)=\int_{0}^{\text{trunc}(\vect v\cdot\vect x)}(\sigma(z)-y)\d z.
\]
Since $\sigma(z)=\sign(z)$ whenever $\abs z\geq\zalphazero$, for
every $z\in\R$ we have
\begin{equation}
\sigma(z)=\sigma(\text{trunc}(z)).\label{eq: sigma compatible with truncation}
\end{equation}
We now make the following observation:
\begin{claim}
\label{claim: truncation does nothing for ground truth surrogate loss}For
every $\left(\vect x,y\right)$ in the support of $D$ we have
\[
\ell^{\sigma}(\vect{v^{*}};\vect x,y)=\ell_{\text{trunc}}^{\sigma}(\vect{v^{*}};\vect x,y)\in[-2\zalphazero,2\zalphazero]
\]
\end{claim}

\begin{proof}
Since $\sigma(z)=\sign(x)$ whenever $\abs x\geq\zalphazero$, for
every $\left(\vect x,y\right)$ in the support of $D$, if $\abs{\vect x\cdot\vect{v^{*}}}>\zalphazero$
we have $y=\sign(\zalphazero)$. Therefore,
\begin{itemize}
\item If $\vect x\cdot\vect{v^{*}}>\zalphazero$, it has to be that $y=\sign(\vect x\cdot\vect{v^{*}})=1$
and therefore 
\[
\ell^{\sigma}(\vect{v^{*}};\vect x,y)=\int_{0}^{\vect{v^{*}}\cdot\vect x}(\sigma(z)-1)\d z=\int_{0}^{\zalphazero}(\sigma(z)-1)\d z=\int_{0}^{\text{trunc}(\vect{v^{*}}\cdot\vect x)}(\sigma(z)-1)\d z=\ell_{\text{trunc}}^{\sigma}(\vect{v^{*}};\vect x,y).
\]
\\
We also see that $\int_{0}^{\zalphazero}(\sigma(z)-1)\d z$ is in $[-2\zalphazero,0]$
since $\sigma(z)\in[-1,1]$ for all $z$.
\item Analogously, if $\vect x\cdot\vect{v^{*}}<-\zalphazero$, then analogously
$y=-1$ and therefore
\[
\ell^{\sigma}(\vect{v^{*}};\vect x,y)=\int_{0}^{\vect{v^{*}}\cdot\vect x}(\sigma(z)+1)\d z=\int_{0}^{-\zalphazero}(\sigma(z)+1)\d z=\int_{0}^{\text{trunc}(\vect{v^{*}}\cdot\vect x)}(\sigma(z)+1)\d z=\ell_{\text{trunc}}^{\sigma}(\vect{v^{*}};\vect x,y).
\]
We also see that $\int_{0}^{-\zalphazero}(\sigma(z)+1)\d z=-\int_{-\zalphazero}^{0}(\sigma(z)+1)\d z$
is in $[-2\zalphazero,0]$ since $\sigma(z)\in[-1,1]$ for all $z$.
\item In the remaining case $\vect x\cdot\vect{v^{*}}\in[-\zalphazero,\zalphazero]$
we have 
\[
\ell^{\sigma}(\vect{v^{*}};\vect x,y)=\int_{0}^{\vect{v^{*}}\cdot\vect x}(\sigma(z)-y)\d z=\int_{0}^{\text{trunc}(\vect{v^{*}}\cdot\vect x)}(\sigma(z)-1)\d z=\ell_{\text{trunc}}^{\sigma}(\vect{v^{*}};\vect x,y).
\]
\\
And again, we see that since $\vect x\cdot\vect{v^{*}}\in[-\zalphazero,\zalphazero]$
we have 
\[
\abs{\int_{0}^{\vect{v^{*}}\cdot\vect x}(\sigma(z)-y)\d z}\leq\int_{0}^{\vect{v^{*}}\cdot\vect x}\abs{\sigma(z)-y})\d z\leq2\alpha.
\]
\end{itemize}
\end{proof}
Now, since $\abs{\ell^{\sigma}(\vect{v^{*}};\vect x,y)}\leq2\zalphazero$,
we can use the Hoeffding bound to conclude that w.p. 0.99 over the
choice of our sample set $S$ satisfies 

\[
\E_{(\vect x,y)\sim S}\left[\ell^{\sigma}(\vect{v^{*}};\vect x,y)\right]=\E_{(\vect x,y)\sim D}\left[\ell^{\sigma}(\vect{v^{*}};\vect x,y)\right]\pm O\left(\frac{\zalphazero}{\sqrt{|S|}}\right),
\]
Since $\vect v_{0}$ minimizes the empirical surrogate loss, among all vectors
whose norm is at most $w_{\max}$, and $\norm{\vect v^{*}}_2\leq w_{\max}$
we also have 
\begin{multline}
\E_{(\vect x,y)\sim S}\left[\ell^{\sigma}(\vect v_{0};\vect x,y)\right]\leq\E_{(\vect x,y)\sim S}\left[\ell^{\sigma}(\vect{v^{*}};\vect x,y)\right]=\E_{(\vect x,y)\sim D}\left[\ell^{\sigma}(\vect{v^{*}};\vect x,y)\right]\pm O\left(\frac{\zalphazero}{\sqrt{|S|}}\right)=.\\
\E_{(\vect x,y)\sim D}\left[\ell^{\sigma}(\vect{v^{*}};\vect x,y)\right]\pm\frac{\eps}{4},\label{eq: the minimizing vector has good empirical surrogate loss}
\end{multline}
where in the last step we substituted $|S|=C\frac{\Delta\zalphazero^{2}}{\eps^{2}}\ln\frac{\Delta\zalphazero}{\eps}+C$
and took $C$ to be a sufficiently large absolute constant.

We make the following further observations regarding the truncated
loss $\ell_{\text{trunc}}^{\sigma}(\vect v;\vect x,y)$:
\begin{itemize}
\item For every $\vect v,\vect x$ and $y$ we have 
\[
\abs{\ell_{\text{trunc}}^{\sigma}(\vect v;\vect x,y)}\leq\int_{0}^{\text{trunc}(\vect v\cdot\vect x)}\abs{\sigma(z)-y}\d z\leq2\zalphazero.
\]
\item For every $\vect v,\vect x$ and $y$ we have 
\begin{equation}
\ell_{\text{trunc}}^{\sigma}(\vect v;\vect x,y)\leq\ell^{\sigma}(\vect v;\vect x,y)\label{eq: truncated surrate loss at most surrogate loss}
\end{equation}
because
\begin{itemize}
\item If $\vect v\cdot\vect x\in[-\zalphazero,\zalphazero]$, we have $\text{trunc}(\vect v\cdot\vect x)=\vect v\cdot\vect x$
and we have $\ell_{\text{trunc}}^{\sigma}(\vect v;\vect x,y)=\ell^{\sigma}(\vect v;\vect x,y)$.
\item If $\vect v\cdot\vect x>\zalphazero$ we have 
\[
\ell^{\sigma}(\vect v;\vect x,y)=\int_{0}^{\text{trunc}(\vect v\cdot\vect x)}(\sigma(z)-y)\d z+\int_{\text{trunc}(\vect v\cdot\vect x)}^{\vect v\cdot\vect x}(1-y)\d z\geq\int_{0}^{\text{trunc}(\vect v\cdot\vect x)}(\sigma(z)-y)\d z
\]
\item If $\vect v\cdot\vect x<-\zalphazero$ we have 
\[
\ell^{\sigma}(\vect v;\vect x,y)=\int_{0}^{\text{trunc}(\vect v\cdot\vect x)}(\sigma(z)-y)\d z+\int_{\text{trunc}(\vect v\cdot\vect x)}^{\vect v\cdot\vect x}(-1-y)\d z\geq\int_{0}^{\text{trunc}(\vect v\cdot\vect x)}(\sigma(z)-y)\d z
\]
\end{itemize}
\end{itemize}
We also need the following claim:
\begin{claim}
For a sufficiently large value of the absolute constant $C$, with
probability at least $0.999$ over the choice of the set $S$, we
have the following uniform convergence bound for $\ell_{\text{trunc}}^{\sigma}$:
\begin{equation}
\max_{\vect v'\in\R^{d}}\abs{\E_{(\vect x,y)\sim D}\left[\ell_{\text{trunc}}^{\sigma}(\vect v';\vect x,y)\right]-\E_{(\vect x,y)\sim S}\left[\ell_{\text{trunc}}^{\sigma}(\vect v';\vect x,y)\right]}\leq\tilde{O}\left(\zalphazero\sqrt{\frac{\Delta}{|S|}}\right)\leq\frac{\eps}{4}\label{eq: uniform convergence for truncated surrogate loss}
\end{equation}
\end{claim}

\begin{proof}
We have 
\begin{align*}
\int_{0}^{\text{trunc}(\vect v'\cdot\vect x)}(\sigma(z)-y)\d z&=\int_{0}^{\zalphazero}(\sigma(z)-y)\indicator_{z\leq\text{trunc}(\vect v'\cdot\vect x)}\d z-\int_{-\zalphazero}^{0}(\sigma(z)-y)\indicator_{z>\text{trunc}(\vect v'\cdot\vect x)}\d z\\
&=\int_{0}^{\zalphazero}(\sigma(z)-y)\indicator_{z\leq\vect v'\cdot\vect x}\d z-\int_{-\zalphazero}^{0}(\sigma(z)-y)\indicator_{z>\vect v'\cdot\vect x}\d z
\end{align*}
which implies
\begin{align}
&\left|\E_{(\vect x,y)\sim D}\left[\int_{0}^{\text{trunc}(\vect v'\cdot\vect x)}(\sigma(z)-y)\d z\right]-\E_{(\vect x,y)\sim S}\left[\int_{0}^{\text{trunc}(\vect v'\cdot\vect x)}(\sigma(z)-y)\d z\right]\right|\notag\\
\leq &\zalphazero\max_{\vect v'\in\R^{d},z\in\R}\left[\abs{\E_{(\vect x,y)\sim D}\left[(\sigma(z)-y)\indicator_{z\leq\vect v'\cdot\vect x}\right]-\E_{(\vect x,y)\sim S}\left[(\sigma(z)-y)\indicator_{z\leq\vect v'\cdot\vect x}\right]}\right]\notag\\
&+\zalphazero\max_{\vect v'\in\R^{d},z\in\R}\left[\abs{\E_{(\vect x,y)\sim D}\left[(\sigma(z)-y)\indicator_{z>\vect v'\cdot\vect x}\right]-\E_{(\vect x,y)\sim S}\left[(\sigma(z)-y)\indicator_{z>\vect v'\cdot\vect x}\right]}\right].\label{eq: triangle inequality for truncated loss}
\end{align}
Standard VC-dimension bounds for linear classifiers tell us that for
a sufficiently large absolute constant $C$, if $\abs S\geq C\frac{\Delta\zalphazero^{2}}{\eps^{2}}\ln\frac{\Delta\zalphazero}{\eps}+C$
then with probability 0.999 for all $\vect v'$ and $z$ we have
\begin{enumerate}
    \item $\abs{\E_{(\vect x,y)\sim D}\left[\indicator_{z\leq\vect v'\cdot\vect x}\right]-\E_{(\vect x,y)\sim S}\left[\indicator_{z\leq\vect v'\cdot\vect x}\right]}  \leq\frac{\eps}{16\zalphazero}$,
    \item $\abs{\E_{(\vect x,y)\sim D}\left[\indicator_{z>\vect v'\cdot\vect x}\right]-\E_{(\vect x,y)\sim S}\left[\indicator_{z>\vect v'\cdot\vect x}\right]}  \leq\frac{\eps}{16\zalphazero}$,
    \item $\abs{\E_{(\vect x,y)\sim D}\left[y\cdot\indicator_{z\leq\vect v'\cdot\vect x}\right]-\E_{(\vect x,y)\sim S}\left[y\cdot\indicator_{z\leq\vect v'\cdot\vect x}\right]}  \leq\frac{\eps}{16\zalphazero}$, and
    \item $\abs{\E_{(\vect x,y)\sim D}\left[y\cdot\indicator_{z>\vect v'\cdot\vect x}\right]-\E_{(\vect x,y)\sim S}\left[y\cdot\indicator_{z>\vect v'\cdot\vect x}\right]}  \leq\frac{\eps}{16\zalphazero}$.
\end{enumerate}
Therefore, combining the bounds above with Equation \ref{eq: triangle inequality for truncated loss},
using a triangle inequality and recalling $\abs{\sigma(z)}\leq1$
for every $z$, we obtain Equation \ref{eq: uniform convergence for truncated surrogate loss}
.
\end{proof}
Thus, we have

\begin{align}
\E_{y\sim D_{y\vert \x}}&\left[\ell_{\text{trunc}}^{\sigma}(\vect v;\vect x,y)\right]-\E_{y\sim D_{y\vert \x}}\left[\ell_{\text{trunc}}^{\sigma}(\vect{v^{*}};\vect x,y)\right]=\E_{y\sim D_{y\vert \x}}\left[\int_{\text{trunc}(\vect{v^{*}}\cdot\vect x)}^{\text{trunc}(\vect v\cdot\vect x)}(\sigma(z)-y)\d z\right]\notag\\
&=\int_{\text{trunc}(\vect{v^{*}}\cdot\vect x)}^{\text{trunc}(\vect v\cdot\vect x)}(\sigma(z)-\sigma(\vect{v^{*}}\cdot\vect x))\d z=\int_{\text{trunc}(\vect{v^{*}}\cdot\vect x)}^{\text{trunc}(\vect v\cdot\vect x)}(\sigma(z)-\sigma(\text{trunc}(\vect{v^{*}}\cdot\vect x)))\d z\notag\\
&\geq \frac{1}{2}\left(\sigma(\text{trunc}(\vect v\cdot\vect x))-\sigma(\text{trunc}(\vect{v^{*}}\cdot\vect x)))\right)^{2}=\frac{1}{2}\left(\sigma(\vect v\cdot\vect x)-\sigma(\vect{v^{*}}\cdot\vect x))\right)^{2}.\label{eq: kanade notes step}
\end{align}
 The first equality is due to linearity of expectation and the definition
of $\ell_{\text{trunc}}^{\sigma}$. The second equality is because $D_{x}$ is a GLM with vector $\vect{v^{*}}$
and activation $\sigma$. The third equality is due to Equation \ref{eq: sigma compatible with truncation}. The inequality $\int_{\text{trunc}(\vect{v^{*}}\cdot\vect x)}^{\text{trunc}(\vect v\cdot\vect x)}(\sigma(z)-\sigma(\text{trunc}(\vect{v^{*}}\cdot\vect x)))\d z\geq\frac{1}{2}\left(\sigma(\text{trunc}(\vect v\cdot\vect x))-\sigma(\text{trunc}(\vect{v^{*}}\cdot\vect x)))\right)^{2}$
is due to $\sigma$ being montonically increasing and $1$-Lipschitz
(same as in page 5 of \cite{kanade2018note}). The last equality is due to Equation \ref{eq: sigma compatible with truncation}.
We therefore have for any vector $v\in\R^{d}$
\begin{align*}
\frac{1}{2}\E_{(\vect x,y)\sim D}\left[\left(\sigma(\vect v\cdot\vect x)-\sigma(\vect{v^{*}}\cdot\vect x))\right)^{2}\right]&\leq\E_{(\vect x,y)\sim D}\left[\ell_{\text{trunc}}^{\sigma}(\vect v;\vect x,y)-\ell_{\text{trunc}}^{\sigma}(\vect{v^{*}};\vect x,y)\right]\\
&=\E_{(\vect x,y)\sim D}\left[\ell_{\text{trunc}}^{\sigma}(\vect v;\vect x,y)-\ell^{\sigma}(\vect{v^{*}};\vect x,y)\right]\\
&\leq\E_{(\vect x,y)\sim S}\left[\ell_{\text{trunc}}^{\sigma}(\vect v;\vect x,y)\right]-\E_{(\vect x,y)\sim D}\left[\ell^{\sigma}(\vect{v^{*}};\vect x,y)\right]+\eps/4\\
&\leq\E_{(\vect x,y)\sim S}\left[\ell^{\sigma}(\vect v;\vect x,y)\right]-\E_{(\vect x,y)\sim D}\left[\ell^{\sigma}(\vect{v^{*}};\vect x,y)\right]+\eps/4.
\end{align*}
The first step uses Equation \ref{eq: kanade notes step}. The second step uses Claim \ref{claim: truncation does nothing for ground truth surrogate loss}. The third step uses Equation \ref{eq: uniform convergence for truncated surrogate loss}. The fourth step uses Equation \ref{eq: truncated surrate loss at most surrogate loss}.
Finally, combining with Equation \ref{eq: the minimizing vector has good empirical surrogate loss},
we get the desired Equation \ref{eq: goal of thm for realizable truncated sigma}.
\end{proof}

\subsection{Learning Sigmoidal Activations}\label{appendix:sigmoid}
We are now ready to prove \Cref{theorem:sigmoid-main} on the learnability of sigmoidal activations.
\begin{proof}[Proof of \Cref{theorem:sigmoid-main}]
    The proof proceeds by reducing to the task of learning from GLMs with clipped activation. Let $\rho=C(\log \frac{\Delta}{\epsilon})$ be some parameter for a sufficiently large constant $C$. Let $D_X$ be the marginal distribution of $D$. We define the "clipped" activation $\sigmoidclipped$ as 
    \begin{equation}
    \sigmoidclipped(z)
    =
    \begin{cases}
        \sigmoid(z) & \text{if } \abs{z}\leq \clippingthreshold -1\\
        \sigmoid(\rho-1)+(1-\sigmoid(\rho-1))\cdot (z-(\rho-1)) & \text{if } z\in (\rho-1,\rho]\\
        \sigmoid(-(\rho-1))-(1+\sigmoid(-(\rho-1)))\cdot (-z-(\rho-1)) & \text{if } z\in [-\rho,-\rho+1)\\
        \sign(z) & \text{if } \abs{z} > \clippingthreshold 
    \end{cases}
\end{equation}
    Observe that for all $t\in \R$
    \begin{equation}
    \label{eqn:sigmoid_approx_clipped}
        |\sigmoid(t)-\sigmoidclipped(t)|\leq |1-(1-e^{-\rho})/(1+e^{-\rho})|\leq e^{-\rho}/2.
    \end{equation}
    Let $D_{\text{clipped}}$ be a GLM on $\R^{\Delta+1}\times \{\pm 1\}$ defined as follows:
    \begin{itemize}
        \item The marginal of $D_{\text{clipped}}$ is sampled as follows: Sample $\vect{x}$ from the marginal of $D$ and output $(\vect{x},1)$.
        \item The activation function is $\sigmoidclipped$.
    \end{itemize}
    
\begin{claim}
\label{claim: running sigmoid on truncated}
    Let $\widehat{\vect{v}}=(\widehat{\vect{w}},\hattau)$ be obtained by running $\AGLM^{\sigmoidclipped}(4\eps,w_{\max},D_{\text{clipped}})$. Then, with probability at least $0.99$, $(\widehat{\vect{w}},\hattau)$ satisfies the equation
     \begin{equation}
     \label{eqn:sigmoid_goal}
        \E_{(\vect{x},y)\sim D}\Bigr[\bigr(\sigmoid(\vect w^* \cdot \x + \tau^*) - \sigmoid(\widehat{\vect w} \cdot \x + \hattau)\bigr)^2\Bigr] \le C'\eps\,.
    \end{equation} for universal constant $C'$.
\end{claim}
\begin{proof}
We see that the choice of $\sigma=\sigmoidclipped$ satisfies the requirements on $\sigma$ in Theorem \ref{thm:realize_truncated_glm}. Indeed, it holds that : (1) $\sigmoidclipped$ is monotone non-decreasing, (2) the function $\sigmoidclipped$ is $1$-Lipschitz, and (3) by construction, we have $\sigmoidclipped(z)=\sign(z)$ whenever $z$ is not in $[-\rho, \rho]$. Thus, we have that 
\[
\E_{\vect{x}\sim D_{X}}\Bigr[\bigr(\sigmoidclipped(\vect w^* \cdot \x + \tau^*) - \sigmoidclipped(\widehat{\vect w} \cdot \x + \hattau)\bigr)^2\Bigr].
\]
Recall from \Cref{eqn:sigmoid_approx_clipped} that $\max_{t\in \R}|\sigmoid(t)-\sigmoidclipped(t)|\leq e^{-\rho/2}$. Thus, we have 
\[
\E_{\vect{x}\sim D_{X}}\Bigr[\bigr(\sigmoid(\vect w^* \cdot \x + \tau^*) - \sigmoid(\widehat{\vect w} \cdot \x + \hattau)\bigr)^2\Bigr]\leq 2\epsilon+2e^{-\rho}\leq C'\epsilon
\]
for $\rho=C\log (\Delta/\epsilon)$.
\end{proof}
Let $\tilde{D}$ be the distribution formed from $D$ by concatenating one to the marginal. A sample from $\tilde{D}$ is generated as: sample $(\vect{x},y)$ from $D$ and output $((\vect{x},1),y)$. We compare the distributions $\tilde{D}$ and $D_{\text{clipped}}$. First, for every specific $z$, the TV distance between a pair of $\{\pm 1\}$-valued random variables with expectations $\sigmoidclipped(z)$ and $\sigmoid(z)$ respectively is bounded by $\abs{\sigmoid(z)-\sigmoidclipped(z)}/2$, which is at most $e^{-\rho}$ by Equation \ref{eqn:sigmoid_approx_clipped}. Comparing this with the definitions of $\tilde{D}$ and $D_{\text{clipped}}$ we see that 
\[
\mathrm{dist}_{\mathrm{TV}}
(\tilde{D},D_{\text{clipped}})
\leq \frac{1}{2} \max_{z
    \in \R}
    \abs{\sigmoidclipped(z)-\sigmoid(z)} \leq e^{-\rho}.
\]
Theorem \ref{thm:realize_truncated_glm} tells us that the algorithm  $\AGLM^{\sigmoidclipped}$ uses at most $O\left(\frac{\Delta\rho^{2}}{\eps^{2}}\ln\frac{\Delta\rho}{\eps}\right)$
samples from the distribution it accesses. This, together with the data-processing inequality and the bound above, lets us conclude that
\begin{multline*}
\mathrm{dist}_{\mathrm{TV}}
\left(\AGLM^{\sigmoidclipped}(\eps,w_{\max},\rho,\tilde{D}),
\AGLM^{\sigmoidclipped}(\eps,w_{\max},\rho,D_{\text{clipped}})
\right)
\leq
O\left(\frac{\Delta\rho^{2}}{\eps^{2}}\ln\frac{\Delta\rho}{\eps}\right)
\mathrm{dist}_{\mathrm{TV}}
(\tilde{D},D_{\text{clipped}})
 \\
\leq O\left(\frac{\Delta\rho^{2}}{\eps^{2}}\ln\frac{\Delta\rho}{\eps}\right)
\cdot 
e^{-\rho}.
\end{multline*}
Substituting $\clippingthreshold=\setalphazero$, we see that for sufficiently large absolute constant $C$ the expression above is at most $0.05$. Combining this with Claim \ref{claim: running sigmoid on truncated}, we see that the vector $(\widehat{\vect{w}},\hattau)$ given by $\AGLM^{\sigmoidclipped}(\eps,\rho,w_{\max},\tilde{D})$ with probability at least $0.95$ satisfies \Cref{eqn:sigmoid_goal}. To complete the proof, we observe that samples from $\tilde{D}$ can be easy generated given samples from $D$.
\end{proof}